\documentclass[11pt]{article}

\usepackage[english]{babel}

\usepackage[letterpaper,top=2cm,bottom=2cm,left=3cm,right=3cm,marginparwidth=1.75cm]{geometry}

\usepackage{amsmath}
\usepackage{graphicx}
\usepackage[colorlinks=true, allcolors=blue]{hyperref}
\usepackage{setspace}
\usepackage{natbib}
\usepackage{booktabs}
\usepackage{multirow}
\usepackage{array}
\usepackage{amsthm}
\usepackage{amssymb}
\usepackage{caption} 
\usepackage[perpage]{footmisc}
\setcitestyle{authoryear,round}
\newcolumntype{L}[1]{>{\raggedright\arraybackslash}p{#1}}
\newcolumntype{C}[1]{>{\centering\arraybackslash}p{#1}}
\newtheorem{theorem}{Theorem}  
\newtheorem{proposition}{Proposition}  
\newtheorem{Hypothesis}{Hypothesis}

\title{How Big Data Dilutes Cognitive Resources, Interferes with Rational Decision-making and Affects Wealth Distribution ?}
\author{Yongheng Hu\footnote{School of International Business, Zhejiang International Studies University, Liuhe Road, Hangzhou 310023, China. Correspondence to: Yongheng Hu (22030101043@st.zisu.edu.cn). As a new working paper, parts of this article have had been publicly reported in \textit{The Academic Forum of the Institute of New Structural Economics of Peking University in 2025}. \textit{The International Conference on Basic Sciences (ICBS) in 2025}. I am very grateful to Dr.Mingyi Yang of Washington State University for his guidance and many helpful discussions. This working paper is incomplete and still being improved, all comments and opinions about the article are welcome. Of course all remaining omissions and errors in statement or technique are mine.}}

\begin{document}
\maketitle

\begin{abstract}
Big data has exponentially dilated consumption demand and speed, but can they all be converted to utility? We argue about the measures of consumption and utility acquisition in CRRA utility function under the condition of big data interaction, we indicate its weakness, i.e., irrational consumption does not lead to the acquisition of utility. We consider that big data, which is different from macro and micro economic signals, formed by general information entropy, affects agents' rational cognition, which makes a part of their consumption ineffective. We preliminarily propose the theory that how dilution mechanism driven by big data will affect agents' cognitive resources. Based on theoretical and empirical analysis, we construct the Consumption Adjustment Weight Function (CAWF) of agents interacting with big data and further apply it to a model of firm wealth distribution with financial frictions, we get analytical solutions according to the Mean Field Game (MFG) and find: Lower financial friction increases the average wealth of firms but also leads to greater wealth inequality. When agents convert effective consumption into utility, which is a weight of total consumption, the average wealth of firms increases with the weight increasing. Meanwhile, wealth inequality follows a U-shaped trend, and it will be the lowest level when the weight approaches to 0.5. In conclusion, we try to provide a new complementary hypothesis to refine the “Lucas Critique” according to the cognitive resources as endowments involved in the decision-making of agents.


\end{abstract}

\newpage

\section{Introduction}
“The medium is the information” (Marshall McLuhan) profoundly reveals the decisive influence of the form of the medium itself on social development and human cognition, whose importance far exceeds the specific content carried by the medium. In the era of digital economy, information is “materialized” into data\footnote{Big data in this article is regarded as an aggregate of information in the broad sense, the value and scale of big data are determined by the uncertainty of information (information entropy) and the volume of information, respectively. In reality, it can be understood that big data comes from “information particles”, i.e., “Token”, the content and scale of Token determine the two properties of big data.}, which becomes structured data like “hot media” and unstructured data like “cold media”. These two different types of data shift the material basis of knowledge and information production from traditional discourse practices to quantifiable and tradable data resources, completing the dataization transformation of knowledge, information, and factors of production. \citep{jones2020nonrivalry}. Therefore, the most interesting issue of this paper is to try to update the individual's information interaction behavior from the traditional face-to-face information exchange and belief updating to the demand for big data consisting of more micro information quantified by general information entropy, to construct a structural model of the agent's consumption decision-making with the interaction of big data, and to explore how the agent's consumption decision-making changes as he or she interacts more deeply and persistently with big data.

Our core assumptions: (\textbf{A}) The uncertainty of information determines the value of information, and numerous pieces of information with different values constitute big data and form data values. (\textbf{B}) The degree of rationality of an agent is determined by the level of cognitive resources the agent possesses\footnote{We consider that agents' cognitive resources endowments determine their rational states, and for each rational state, the agents will undergo a certain belief transition.}. (\textbf{C}) The agent's interaction with big data is an interaction between “data scale” and “data value”: Data scale affects the agent's cognitive resources, which determines the validity of consumption adjustment (i.e., only the consumption adjustments made in a fully rational situation will be converted into 100$\%$ incremental utility). Data value affects the direction of the agent's consumption adjustment (i.e., whether the agent chooses to increase or decrease consumption is determined by the level of data value). (\textbf{D}) The agent's interaction with big data will generate a consumption-to-utility weight, i.e., in the interaction with big data, the agent's consumption is not completely effective, only a certain weight of consumption is effective and can be converted to utility, and this weight will evolve and converge to a fixed non-zero value as the agents continue to interact with big data\footnote{This can be easily explained by platform economy: Agents browse through various types of networks, e-commerce, video and other platforms, accepting the scale and different value of the data, affecting their cognitive resources, and then they make consumption decisions after the cognitive resources have been affected, their decisions will be irrational, and can't be fully converted into utility.}.

The literature about data as information affecting the economy: \citet{farboodi2021model} and \citet{farboodi2019big} view data as an information resource that reduces production uncertainty, and that firms use data elements to acquire forward-looking knowledge that improves the accuracy of their predictions of optimal production techniques and increases productivity levels. E-commerce online platforms use consumer information materialized by data to change the “gloss” of product quality, causing consumers to misjudge the true quality of the product and inducing unwanted consumption behavior \citep{acemoglu2025big}. Moreover, platforms may not only use “big data” pricing mechanisms to make profits but may also redistribute information in the form of “filter bubbles” that target heterogeneous personal data to maximize platform engagement \citep{acemoglu2024model}. Then, \citet{ding2024consumer} argues that data-transformed consumer benefits, such as digital vouchers or shopping subsidies, have a stimulative effect on consumption growth that stems from increased consumer spending in the targeted category, rather than crowding out consumption spending in other categories. Some studies have incorporated information uncertainty into the “efficiency-equity” research framework, pointing out that information frictions link data generation and economic activity, mainly in the following ways: Micro-level uncertainty creates resource mismatches in the macro-system, and makes macro total factor productivity is endogenous to the data collection behaviors of micro-subjects (\citealp{farboodi2021model}; \citealp{david2016information}; \citealp{benhabib2016endogenous}), which can be attributed to the widespread use of dataization mobile communication tools, whose availability of data allows for more random and rapid changes in individual behavior \citep{fabregas2025digital}, thus leads to the inability of traditional data selection mechanisms to accurately identify imperfect information and decision-making errors, as \citet{gans2025ai} points out that AI's analytical and decision-making capabilities are excellent in data-rich domains but less trustworthy in judgment-intensive data environments. Therefore, \citet{caplin2025data} introduces new forms of data to identify agents' preferences, beliefs, etc. by constructing data engineering models. Notably, \citet{jones2025how} and \citet{jones2024the} explore the possibility that big data-driven AI technologies may pose a threat to human survival while promoting economic growth at the level of heterogeneous agent utility acquisition. However, as mentioned in the opening section, he simply attributes the factors affecting utility acquisition to the heterogeneity of agents' risk aversion coefficients, ignoring the relationship between consumption and effective consumption (consumption that delivers utility) in the presence of big data (technology like “AI”) interactions.

Other thought-provoking literature on the impact of information on individual decision-making includes \citet{handel2018frictions} analyze “frictions” and “mental gaps” in the use of information, revealing their impact on economic decision-making. \citet{epley2016mechanics} explored the mechanism of motivated reasoning and analyzed how people adjust information to maintain belief consistency, and \citet{benabou2016mindful} proposed the framework of “economics of beliefs” to analyze the production and consumption of beliefs and their intrinsic values. \citet{gino2016motivated} study motivated Bayesian behavior, revealing how individuals balance between moral sense and selfish behavior. \citet{grubb2015overconfident} reveals consumers' overpurchasing and choice mistakes due to overconfidence in the market. \citet{barberis2013thirty} systematically reviews the application of prospect theory in behavioral economics.

Further, studies on deviations from rational expectations of complete information are also relevant to our paper, e.g., \citet{bordalo2020overreaction}, \citet{coibion2015information}, \citet{coibion2012what}, \citet{carroll2003macroeconomic}, \citet{mankiw2003disagreement}. There is also a literature on constructing quantitative models based on the spread of information, e.g., \citet{carroll2020sticky}, \citet{mackowiak2015business}, \citet{woodford2013macroeconomic}, \citet{mankiw2007sticky}, and applying them to expectancy inference \citep{adam2019stock}, measuring confidence fluctuations \citep{angeletos2018quantifying} and ambiguity (\citealp{baqaee2020asymmetric}; \citealp{bianchi2018uncertainty}; \citealp{bidder2012robust}). At the technical perception, the most classic research could be traced back to \citet{Lucas}, after that, \citet{bhandari2025survey}, \citet{hansen2016sets}, \citet{strzalecki2011axiomatic}, and \citet{hansen2001a}, \citet{hansen2001b} have been more sophisticated in their research on decision theory and subjective belief updating models. They generally place the research problem in a dynamic analytical framework, arguing that agents update beliefs through subjective probabilistic distortions that deviate from rational expectations. By exogenizing the belief distortion parameter (e.g., setting it as an AR process), constructing a belief distortion operator to measure the degree of belief distortion, and further defining the subjective probability distortion measure, calculating the discounted value of the future expectation, and placing it into a recursive equation to measure the continuation utility. The basic core is to set the belief distortion parameter exogenously, but its impact on the economy is endogenous, i.e., the continuation utility of agents is affected by economic factors, such as unemployment and inflation, which lead to changes in the belief distortion parameter, thus further amplifying the economic fluctuations under the influence of the belief distortion parameter. DSGE models are usually popular in such research frameworks.

In contrast, our study is based on a static framework, and we have a central exogenous assumption, we need big data interactions to occur: Our big data interaction environment is optionally decided by the agent, and the “consumption-adjusted weights” as defined in this paper exist only when the agent's behavior of interacting with big data exists. Therefore, it can be understood that the theory of our article is not parallel to the traditional decision theory and subjective belief updating model, but is a further continuation, i.e., we try to analyze: when each dynamic time node arrives (implying that it is a static situation at that node), the agent's subjective beliefs finish updating, and the big data interactions have occurred for a period, then, agent makes consumption decision at this time node, how much of the utility is available according to the consumption decision? 

This would imply that: Firstly, big data comes from more micro information elements (like unstructured data) that are not capable of having accurate value judgments, which are not similar to macroeconomic information (like structured data)\footnote{Macroeconomic fluctuations can be directly represented by the quantification of an economic variable and in this way influence the agent's beliefs, pessimistic or optimistic.}, they (big data) are generally disseminated through online media and digital platforms forming data elements that affect the cognitive resources of agents only at the time they have received the big data. Since the elements that drive belief updating come more from macro information, for the agent, the former can actively choose whether to accept it or not, while the latter's acceptance of macro information is passive and unavailable to the agent due to the objective existence of macro-economic dynamics. Secondly, when the agent undergoes a big data interaction, its cognitive resources are affected, which affects the weight of its effective consumption over the total consumption, and when the agent stops big data interaction, the change of its cognitive resource level will also stop, at which time the weight of effective consumption to total consumption will be fixed. Throughout the process, the agent's subjective belief updating will be accompanied by macroeconomic dynamics all the time, meanwhile, if there is a big data interaction, then our theory will explain how much consumption is effective, i.e., providing utility. If there are no big data interactions, then it will be useful to use traditional decision theory and subjective belief updating models\footnote{In other words, our theory would be more applicable to people who are data-preferential, i.e., agents who have long relied on online platforms to assist in their consumption behaviors.}. Finally, it is also important to point out that the “big data” in our paper is different from the “big data from and applied to agents”. As \citet{acemoglu2025big} talked in their paper: Platforms obtain and integrate big data through the online transaction behavior of users to form the “gloss” of products and accordingly choose products that maximize the platform's revenue rather than truly matching the users' needs, which reduces the users' welfare. In contrast, our measurement of data value is based on a more general perspective, i.e., data value comes from, and only comes from the uncertainty of information (information entropy). It does not come from macroeconomic information as previously explored, i.e., the value of the data is not measured in terms of fluctuations in economic variables. Literature related to this point can be traced back to the idea of “Economy of Knowledge”\citep{hayek1945use}, which takes the price system as the dissemination mechanism of information. Based on this theory, \citet{sims2006rational} and \citet{sims2003implications} began his research about the economic implications of rational inattention. Then, some theoretical studies (\citealp{angeletos2025inattentive}; \citealp{hebert2023information}) and experimental researches (\citealp{pomatto2023cost}; \citealp{dean2023experimental}; \citealp{caplin2022rationally}; \citealp{hebert2021neighborhood}) have expanded the theory and given insightful conclusions. Among them, \citet{pomatto2023cost} have proposed the maximum log-likelihood ratio information acquisition cost (LLR) function, and provided a more concise and elegant axiomatic formulation of the information cost structure. In their research framework, there is a cost to acquire information, which is different from our article. In normal information economics framework, scholars generally believe that information is scarce and valuable, and that the collection, analysis, learning and utilization of information incur costs(\citealp{biglaiser2025information}, \citealp{chatterjee2025bargaining}; \citealp{vong2025reputation}; \citealp{gentzkow2025ideological}). However, the big data formed based on information in our article does not require material costs because it does not come from any macro or micro signals of economics. The only cost incurred is the “cognitive” cost. Hence, we treat cognitive resources as a natural attribute of an agent that accumulates through experience and learning of himself, i.e., cognitive resource endowment. And as for information, according to \citet{hayek1945use}, they make it dependent on price, which implies that information has the “entity” of an economic variable. In our discussion, we set the aggregation of information becomes big data, and the value of big data comes and only comes from the uncertainty of information, i.e., general information entropy. That means it doesn't come from whether the data is an accurate, error-free and complete measure of macro market or product quality, i.e., it is not a kind of knowledge in “Economy of Knowledge”. We are more interested in how big data, more generally defined, affects the cognitive resources of the agent, and what proportion of consumption decisions made by the agent based on cognitive resources alone provide utility, i.e., the agent does not have an expectation of utility to the consumption decision ex ante, demand arises and consumption occurs only at the moment after the agent interacting with the big data for a period. In other words, we will study the weighting relationship between consumption and effective consumption due to big data interactions, rather than the recursive relationship between expected utility and current utility formed by the agent through information acquisition and belief updating.

Firstly, we explore the dynamics of agents' cognitive resources affected by the continuous time and increasing scale of big data through constructing a differential dynamics system. Secondly, we define data value variable by introducing general information entropy into the measurement system. Then, based on prospect theory and empirical analysis, we obtain the amount and direction of consumption adjustment of agents with different rationality, based on this, we establish the “Consumption Adjustment Weight Function” (CAWF) under the big data interaction of agents. Finally, we use empirical analysis to demonstrate the relationship between uncertainty and financial friction, and apply the CAWF to the model of firms' wealth distribution with financial friction to explore the economic influence of big data interaction on wealth distribution according to the \textit{Mean Field Game} (MFG).

\section{The Prerequisite Theory: Big Data and Dilution}
In this section, we will set up and prove the antecedent basic theory of the dilution of agents' cognitive resources by big data interactions, and we point out that since the process of accumulating data possesses the character of continuous time, and the result of the accumulation possesses the character of scale, then when an agent chooses to interact with big data, his or her cognitive resources, which have originally remained at a certain level, should also be affected by the dynamics of time and the scale of the data at the same time.

\subsection{Dilution with Continuous Time}

\textbf{Differential Equation}: Considering the cognitive resource dilution and learning recovery mechanism in big data interactions, the dynamics of the cognition retention coefficient $r^c$ after an agent chooses a big data interaction is as follows:
\begin{align}
\frac{\partial r^c}{\partial t}=\underbrace{-\lambda^c r(t) s^c}_{\text{\textit{Dilution item}}}+\underbrace{v^c r(t)(1-r(t))}_{\textit{Recovery item}\label{(1)}}
\end{align}

The “\textit{Dilution item}” represents the dilution of the individual's cognitive level by the big data interaction, which is proportional ($\lambda^c \in (0,+\infty)$) to the degree of big data interaction $s^c \in (0,+\infty)$ and the current cognitive level $r(t)$. The “\textit{Recovery item}” represents the individual's thinking and learning during the big data interaction, thus restoring a certain degree of cognitive level. Logistic-type moderators are used to ensure that the cognitive retention level $r\in[0,1]$, prevents the cognitive level from exceeding the reasonable range: In the case of higher cognitive level $r(t) \rightarrow 1$ or lower $r(t) \rightarrow 0$, the individual's rational recovery efficiency is slow. $v^c \in (0,+\infty)$ is the rate of recovery that reflects the strength of individual self-correction.

\begin{theorem}
\textit{Once an agent starts interacting with big data, regardless of the depth of its interaction with big data, the agent's cognitive resources will always continue to decrease and converge over time, with the level of convergence influenced only by the relative magnitude of the agent's cognitive resource dilution and recovery}.
\end{theorem}

\begin{proof}
Let ${r}_{1}=r(t)^{-1}$, we get: 
\[\begin{gathered}r^c=r(t)=\frac{1}{r_1}\\\dot{r}(t)=-r(t)^2\dot{r}_{1}\\\dot{r}(t)+r(t)[\lambda^c s^c-v^c]=vr(t)^2\\\dot{r}_{1}-r_{1}[\lambda^c s^c-v^c]=-v^c\end{gathered}\]

Set the integration factor $\mu^c(t)$, and then we solve the equation:
\[\begin{gathered}\mu^c(t)=e^{\int-[\lambda^c s^c-v^c]dt}=e^{v^c t}e^{-\lambda^c\int s^c dt}\\\dot{r}_{1}\mu^c(t)-[\lambda^c s^c-v^c]r_{1}\mu^c(t)=-v^c\mu^c(t)\end{gathered}\]

According to FOC:
\[\begin{gathered}\frac{\partial}{\partial t}[\mu^c(t)r_{1}]=-v^c\mu^c(t)\end{gathered}\]

Solving it, we get:
\[\begin{gathered}\frac{\partial}{\partial t}[e^{v^ct}e^{-\lambda^c\int s^c dt}r_{1}]=-v^c e^{v^ct}e^{-\lambda^c\int s^c dt}\\\int_0^t\frac{\partial}{\partial\tau}\left[e^{v^c\tau}e^{-\lambda^c\int_0^\tau s^c(\varepsilon)d\varepsilon}r_{1}\right]d\tau=-v\int_0^te^{v^c\tau}e^{-\lambda^c\int_0^\tau s^c(\varepsilon)d\varepsilon}d\tau\end{gathered}\]

Then, we get $r_1$:
\[\begin{gathered}r_{1}=e^{-v^c\tau}e^{\lambda^c\int_0^ts^c(\varepsilon)d\varepsilon}\left[r_{1}(0)-v^c\int_0^t e^{v^c\tau}e^{-\lambda^c\int_0^\tau s^c(\varepsilon)d\varepsilon}d\tau\right]\\r_{1}=r_{1}(0)e^{-v^c\tau}e^{\lambda^c\int_{0}^{t}s^c(\varepsilon)d\varepsilon}-v^c e^{-v^c\tau}e^{\lambda^c\int_{0}^{t}s^c(\varepsilon)d\varepsilon}\int_{0}^{t}e^{v^c\tau}e^{-\lambda^c\int_{0}^{\tau}s^c(\varepsilon)d\varepsilon}d\tau\end{gathered}\]

Setting the initial cognitive level $r(0)$ to a constant value $r_0$, and the big data interaction level $s^c(\varepsilon)$ also to a constant value $s_0$. We get:
\begin{align}
r(t)=\frac{r_0e^{\int_0^t[v^c-\lambda^c s^c(\varepsilon)]d\varepsilon}}{1+r_0v^c\int_0^te^{\int_0^t[v^c-\lambda^c s^c(\varepsilon)]d\varepsilon}d\tau}=\frac{r_0e^{(v^c-\lambda^c s_0)t}}{1+\left(\frac{r_0v^c}{v^c-\lambda^c s_0}\right)(e^{(v^c-\lambda^c s_0)t}-1)}\label{(2)}
\end{align}

When $\frac{\partial r^c}{\partial t}=0$, then $\lambda^cs_0=v^c$, we get: 
\begin{align}
r^c=r^*=\frac{r_0}{1+r_0v^ct}\label{(3)}
\end{align}

Figure 1 illustrates the dynamics of function (2): 
\begin{figure}[htbp]
\centering
\includegraphics[width=11cm]{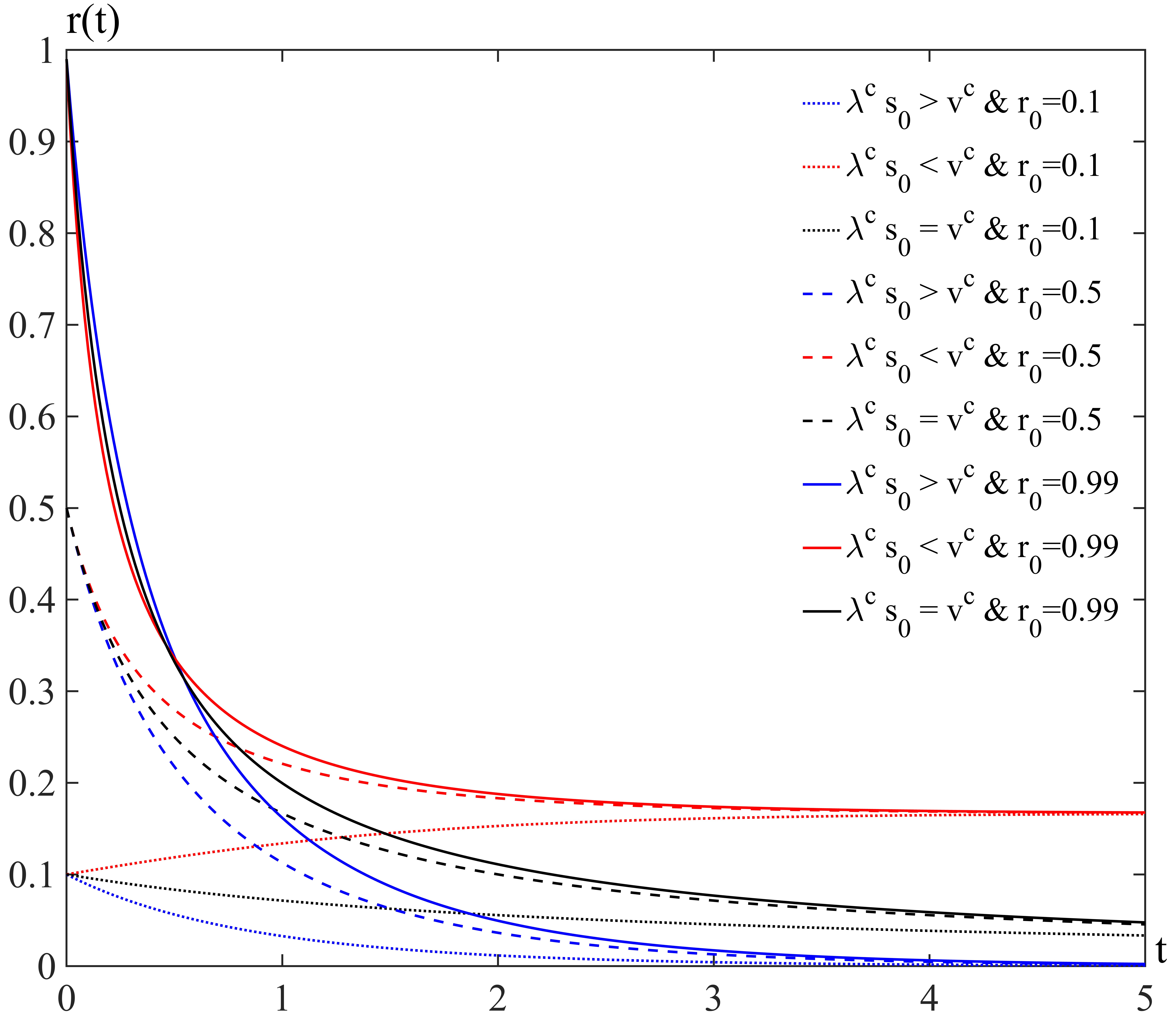}
\caption{\label{fig:F1}Cognition Retention Dynamic with Continuous Time}
\end{figure}

Figure 1 exhibits the dynamic trend of cognition retention level $r(t)$ with time when the initial cognition level $r_0$, as well as the dilution efficiency and recovery efficiency coefficient gap $\lambda^c s_0-v^c$ are both different. Where the blue curve represents $\lambda^c s_0-v^c>0$, at which $\lambda^cs_0=3,v^c=2$. The red curve represents $\lambda^c s_0-v^c<0$, at which $\lambda^c s_0=5,v^c=6$. The black curve represents $\lambda^c s_0-v^c=0$, at which $\lambda^c s_0=4,v^c=4$. The solid line represents the high initial cognition level $r_0=0.99$, the dashed line represents medium initial cognition level $r_0=0.5$, and the dotted line represents low initial cognition level $r_0=0.1$.

Hence, for any individual, the cognition retention level changes and eventually converges after big data interaction, and the level of convergence is only related to $\lambda^c s_0-v^c$. When $\lambda^c s_0-v^c>0$, the final cognition retention level is the highest. When the individual's initial cognition level $r_0>0.16$, which is in the middle or high level, the individual cognition retention level are gradually reduced. And when the individual's initial cognition level $r_0<0.16$ is low, there is a slight increasing trend of r(t) in the case of $\lambda^c s_0-v^c>0$. For the rest of the cases $\lambda^c s_0-v^c=0$ and $\lambda^c s_0-v^c<0$, $r(t)$ is gradually decreasing. 

Ultimately, we can conclude that for a rational agent (i.e., with a high or normal level of initial cognitive resources, $r _0> 0.2$), the agent's cognitive resources gradually decrease and equilibrate over time as long as the agent chooses to interact with big data, regardless of the size of the big data it receives (i.e., the degree of data interaction is constant).
\end{proof}

\subsection{Dilution with the Increasing Scale of Data}
\textbf{Cognitive Resource}: Based on the conclusions of the cognitive resource retention analysis in 2.1, we formally introduce the agent's cognitive resources $R^c=R(t)$ in this section\footnote {Cognitive retention is the endowment state of the agent's cognitive resources: The higher the cognitive retention, $r(t)\rightarrow1$, the more cognitive resources $R^c$ the agent has, and the lower the cognitive retention, $r\rightarrow0$, the less cognitive resources $R^c$ the agent has.}, the following six assumptions are proposed as the prerequisites of our model: 

(\textbf{A}) Finite Cognitive Resource: The agent has a maximum cognitive resource $R_{max}=R_0$ when making decisions independently (Does not exist or before big data interactions). (\textbf{B}) Linear Dilution: When there is interaction behavior between the agent and big data, the agent's cognitive resources are linearly diluted with each increase in the size of the big data interaction. (\textbf{C}) Nonlinear Marginal Load: The cognitive load induced by big data per unit size decreases with the overall size of big data increasing. (\textbf{D}) Linear Recovery: The agent has the ability to learn to recover cognitive resources while interacting with big data, and the recovery term follows a linear gradient flow, i.e., the rate (intensity) of recovery is proportional to the degree of deviation from the initial cognitive resources\footnote {For the cognitive resources $R(t)$, we emphasize that it serves as a resource that the agent can actually allocate based on the cognitive retention state, and assume that the lower an agent possesses cognitive resources are, the faster their recovery rate (intensity) will be, i.e., a linear gradient flow. This can also be derived from the “Recovery” item of function (1): $r(t)=R(t)/R_0, recovery=r(t)(1-r(t))=(R(t)/R_0)(1-(R(t)/R_0))=(R(t)/R_0)- (R(t)/R_0)^2=(R(t)(R_0-R(t)))/(R_0)^2 \varpropto (R_0-R(t))$. Therefore, $(R_0-R(t))$ is the main recovery item of the cognitive resource dynamic.}. (\textbf{E}) Homogeneous Interaction: Big data exists without any economic variable entities, therefore, big data is homogeneous, and agents' cognitive resources allocate to big data are evenly distributed. (\textbf{F}) Algorithmic Intelligentsia: We assume big data as algorithm-driven intelligentsia, and the intelligentsia can get the equivalent cognitive resources by learning through algorithm after diluting the cognitive resources of the agent. Therefore, the agent interacts with big data, which is essentially an agent having cognitive resources $R_i^c=R_i(t)$ interacting with an intelligent agent having cognitive resources $R_j^c=R_j(t)$. According to the homogeneous interaction assumption\footnote {There are $n_{number}\rightarrow\infty$ agents in the society that have homogenized cognitive resources $R_i^c$, and at the same time evenly allocate cognitive resources to big data of scale $n_{data} \rightarrow \infty$, then the intelligent agents will have $(R_i^c/n_{data})*n_{number}=R_j^c \approx R_i^c$ cognitive resources.}, we set $R_i^c \approx R_j^c$.

Then, the cognitive resources available to agent $i$ under the condition of discrete time $\Delta t \neq 0$ are:
\begin{align}
{R_i(t)=R_0-\sum\Theta_{ij}R_j(t-\Delta t)\label{(4)}}
\end{align}

$\Theta_{ij}\in[0,1]$ is the proportion of cognitive resources that agent $i$ interacts with big data $j$ and thus allocates to big data $j$, which will lead to a decrease in cognitive resources for agent $i$. Define $\sigma^c\in (0,1)$ as the cognitive load generated by the unit size of big data, according to the assumption, there is a marginal decreasing effect of big data size on the cognitive resource load of the agent, so the proportion of cognitive resource allocation $\Theta_{ij}$ is as follows:

\[\Theta_{ij}=\frac{\sigma^c}{(n-1)^{\gamma^c}}\]

$\gamma^c\in (0,1)$ controls the degree of nonlinearity of the big data scale effect. Since $R_i(t) \approx R_j(t)$, the cognitive resources of agent $i$ in steady state are as follows:
\[\lim_{t\to\infty}R_i(t)=R_0-\left[\frac{\sigma^c}{(n-1)^{\gamma^c}}\right]\left[(n-1)\lim_{t\to\infty}R_i(t)\right]=R_0-\sigma^c(n-1)^{1-\gamma^c}\lim_{t\to\infty}R_i(t)\]

When $\Delta t \rightarrow 0$, considering the linear gradient flow of recovery item, the complete cognitive resource dynamics according to the assumptions is as follows:
\[\begin{gathered}\frac{\partial R_i(t)}{\partial t}=\mu^c(R_0-R_i(t))-\eta^c\sigma^c(n-1)^{1-\gamma^c}R_i(t)\\\frac{\partial R_i(t)}{\partial t}+\left[\mu^c+\eta^c\sigma^c(n-1)^{1-\gamma^c}\right]R_i(t)=\mu^c R_0\\R_{i}(t)=\frac{\mu^c R_{0}}{\mu^c+\eta^c\sigma^c(n-1)^{1-\gamma^c}}+Ce^{-[\mu^c+\eta^c\sigma^c(n-1)^{1-\gamma^c}]t}\\\lim_{t\to\infty}R_i(t)=R_i^*=\frac{\mu^c R_0}{\mu^c+\eta^c\sigma^c(n-1)^{1-\gamma^c}}\end{gathered}\]

$\mu^c$ is the recovery rate of cognitive resource when consider the recovery item, $\eta^c$ is the dilution factor of cognitive resource. If the cognitive resource of agent $i$ in the steady state is $R_{i}^*$, and the dilution effect of big data interaction on cognitive resource reaches 50\% or more is specified as significant dilution, then for the cognitive retention level $r(t)$, we solve the threshold scale $n^*$ of big data:
\[\begin{gathered}1-r(t)=1-\frac{R_i^*}{R_0}=\frac{\eta^c\sigma^c(n-1)^{1-\gamma^c}}{\mu^c+\eta^c\sigma^c(n-1)^{1-\gamma^c}}\\1-r(t)>\frac{1}{2}\Rightarrow n^*=1+\left(\frac{\mu^c}{\eta^c\sigma^c}\right)^{\frac{1}{1-\gamma^c}}\end{gathered}\]
\textbf{Stochastic Process}: In reality, due to the objective existence of various uncertainties, the generation of social information and the dissemination of data elements are often unpredictable, which are similar to stochastic processes. Therefore, this paper sets the generation of big data to follow the standard Brownian Motion $W(t), dW(t)=\epsilon(t)\sqrt{dt}, \epsilon(t)\thicksim(0,1)$, which measures how big data affects the dynamics of an individual's cognitive resources under the consideration of exogenous uncertainty factors. Based on the analysis and conclusions of 2.1 Cognitive Retention Level and 2.2 Cognitive Resource Dynamics, we further set: (\textbf{A}) Big Data Dilution is Stable: This implies that there is a leverage between the agent's initial cognitive resource $R_0$ and instant cognitive resource $R_i(t)$, i.e., $R_0=\theta^cR_i(t), \theta^c>1$. For example: An agent with initial cognitive resource of $R_0$ starts to interact with big data at moment t, and the cognitive resource decreases from $R_0$ to $R_t$, and $R_0>R_t$. When the agent interacts with big data at moment $t+\delta$ again, then the agent's initial cognitive resource in this interaction is $R_t$. When $\delta\rightarrow0$, we set the leverage of the agent's initial cognitive resource in each big data interaction to be $\theta^c=R_0/R_i(t)>1$. (\textbf{B}) Endogenous Big Data Dilution Intensity: The impact of big data on cognitive resources follows the standard Brownian Motion $dW(t)$. According to the conclusion of the cognitive retention analysis (Figure 1), it can be seen that the dilution intensity of big data on agents' cognitive resources is embodied in the gap between the initial cognitive resources $R_0$ and the equilibrium cognitive resources $R_i^*$ (i.e., the lower the equilibrium cognitive retention $r^*$, the greater the dilution intensity of big data). Hence, we set the dilution intensity of big data for the agent's cognitive resources to be endogenous, i.e., the big data shock is $\psi^c(R_0-R_i^*)dW(t)$, $\psi^c>0$ is the fluctuation coefficient of standard Brownian Motion. Accordingly, Theorem 2 is formulated:

\begin{theorem}
\textit{When an agent interacts with big data, regardless of the level of cognitive resource recovery and dilution, the expectation value of agent's cognitive resource distribution becomes lower with the increasing scale of big data.}
\end{theorem}

\begin{proof}
The complete agent's cognitive resource dynamics process is given as follows:
\begin{align}
dR_i(t)=\left[\mu^c(R_0-R_i(t))-\eta^c\sigma^c(n-1)^{1-\gamma^c}R_i(t)\right]dt+\psi^c(R_0-R_i^*)dW(t)\label{(5)}
\end{align}

Substituting $R_0=\theta^cR_i(t)$, we get:
\[\frac{dR_i(t)}{R_i(t)}=\left[\mu^c(\theta^c-1)-\eta^c\sigma^c(n-1)^{1-\gamma^c}\right]dt+\psi^c\left[\theta^c-\frac{\mu^c\theta^c}{\mu^c+\eta^c\sigma^c(n-1)^{1-\gamma^c}}\right]dW(t)\]

Let $\Phi^c=\mu^c(\theta^c-1)-\eta^c\sigma^c(n-1)^{1-\gamma^c},\Omega^c=\psi^c\left[\theta^c-\frac{\mu^c\theta^c}{\mu^c+\eta^c\sigma^c(n-1)^{1-\gamma^c}}\right]$, we get:
\begin{align}
\frac{dR_i(t)}{R_i(t)}=\Phi^cdt+\Omega^cdW(t)\label{(6)}
\end{align}

According to \textit{Ito's Lemma}, the cognitive resource dynamics $dR_i(t)$ satisfies the standard \textit{Geometric Brownian Motion} (GBM), therefore, let $x_i^c=log(R_i(t))$, where $x_i^c$ represents the logarithm of the cognitive resource $R_i(t)$ for the agent $i$, the stochastic process could be written as:
\begin{align}
dx_i^c=\left(\Phi^c-\frac{1}{2}(\Omega^c)^2\right)dt+\Omega^cdW(t)\label{(7)}
\end{align}

Let $\Sigma^c=\Phi^c-\frac{1}{2}(\Omega^c)^2$, set the probability density function of $x_i^c=log(R_i(t))$ to be $P(x_i^c)$ and satisfy the following \textit{Kolmogorov Forward Equation} (KFE):
\[\begin{gathered}0=-\Sigma^c\frac{\partial P(x_i^c)}{\partial x_i^c}+\frac{1}{2}(\Omega^c)^2\left(\frac{\partial^2P(x_i^c)}{\partial(x_i^c)^2}\right)-\beta^cP(x_i^c)+\beta^c\delta(x_i^c)\\0=-\Sigma^cP^{(1)}(x_i^c)+\frac{1}{2}(\Omega^c)^2P^{(2)}(x_i^c)-\beta^cP(x_i^c)+\beta^c\delta(x_i^c)\end{gathered}\]

Where $\Delta P(x_i^c)$ denotes the proportion of cognitive resources lost by the agent as a result of its involvement in big data interactions, which is subtracted from the distribution of cognitive resources at the rate of $\beta^c$. $\delta(x_i^c)$ is a Dirac function with a function value equal to zero at all points except zero, and its integral over the entire domain of definition is equal to one. $\beta^c \delta(x_i^c)$ represents the cognitive resources that have just been diluted by big data, and after re-entering the distribution of cognitive resources at the rate of $\beta^c$, its effect on the cognitive influence of an agent is 0. This is more in line with the characteristics of the cognitive resources “diluted” by big data, i.e., the agent still possesses the diluted cognitive resources, but this part of the cognitive resources is vague and ineffective, the agent cannot make rational decisions based on this part of cognitive resources. The analysis form of $P(x_i^c)$ can be obtained by solving KFE:
\[\begin{gathered}P(x_i^c)=A\exp\left(\frac{\Sigma^c+\sqrt{(\Sigma^c)^2+2\beta^c(\Omega^c)^2}}{(\Omega^c)^2}x_i^c\right),x_i^c<0\\P(x_i^c)=A\exp\left(\frac{\Sigma^c-\sqrt{(\Sigma^c)^2+2\beta^c(\Omega^c)^2}}{(\Omega^c)^2}x_i^c\right),x_i^c>0\end{gathered}\]

Normalizing $P(x_i^c)$ and solving it:
\[\begin{gathered}A\left[\int_{-\infty}^0\exp\left(\frac{\Sigma^c+\sqrt{(\Sigma^c)^2+2\beta^c(\Omega^c)^2}}{(\Omega^c)^2}x_i^c\right)dx_i^c+\int_0^\infty \exp\left(\frac{\Sigma^c-\sqrt{(\Sigma^c)^2+2\beta^c(\Omega^c)^2}}{(\Omega^c)^2}x_i^c\right)dx_i^c\right]=1\\A=\frac{\beta^c}{\sqrt{(\Sigma^c)^2+2\beta^c(\Omega^c)^2}}\end{gathered}\]

Hence:
\begin{align}
P(x_i^c)=\frac{\beta^c}{\sqrt{(\Sigma^c)^2+2\beta^c(\Omega^c)^2}}\exp\left(\frac{\Sigma^c+\sqrt{(\Sigma^c)^2+2\beta^c(\Omega^c)^2}}{(\Omega^c)^2}x_i^c\right),x_i^c<0\label{8}\\P(x_i^c)=\frac{\beta^c}{\sqrt{(\Sigma^c)^2+2\beta^c(\Omega^c)^2}}\exp\left(\frac{\Sigma^c-\sqrt{(\Sigma^c)^2+2\beta^c(\Omega^c)^2}}{(\Omega^c)^2}x_i^c\right),x_i^c>0\label{(9)}
\end{align}

In this part, according to functions (8) and (9) above, the parameters are set as follows: $\sigma^c=0.4,\gamma^c=0.4,\psi^c=0.4,\beta^c=0.8,\theta^c=2$. Two sets of $\mu^c,\eta^c$ are set up: when $\mu^c>\eta^c,\mu^c=2,\eta^c=1$, when $\mu^c<\eta^c,\mu^c=2,\eta^c=2.1$. The big data scale variables $n=10,15,20,25$. The simulation results are shown in Figure 2 and Figure 3:

From Figure 2, when the cognitive resource recovery effect of the agent is stronger than the dilution effect, i.e., $\mu^c>\eta^c$, with the increase of the big data size $n$, the cognitive resource distribution of the agent will be gradually transitioned from a right-skewed distribution to a normal distribution and eventually stabilized to a left-skewed distribution, which means that the cognitive resource expectation value of the agent will be gradually reduced. 

\begin{figure}[htbp]
\centering
\includegraphics[width=11cm]{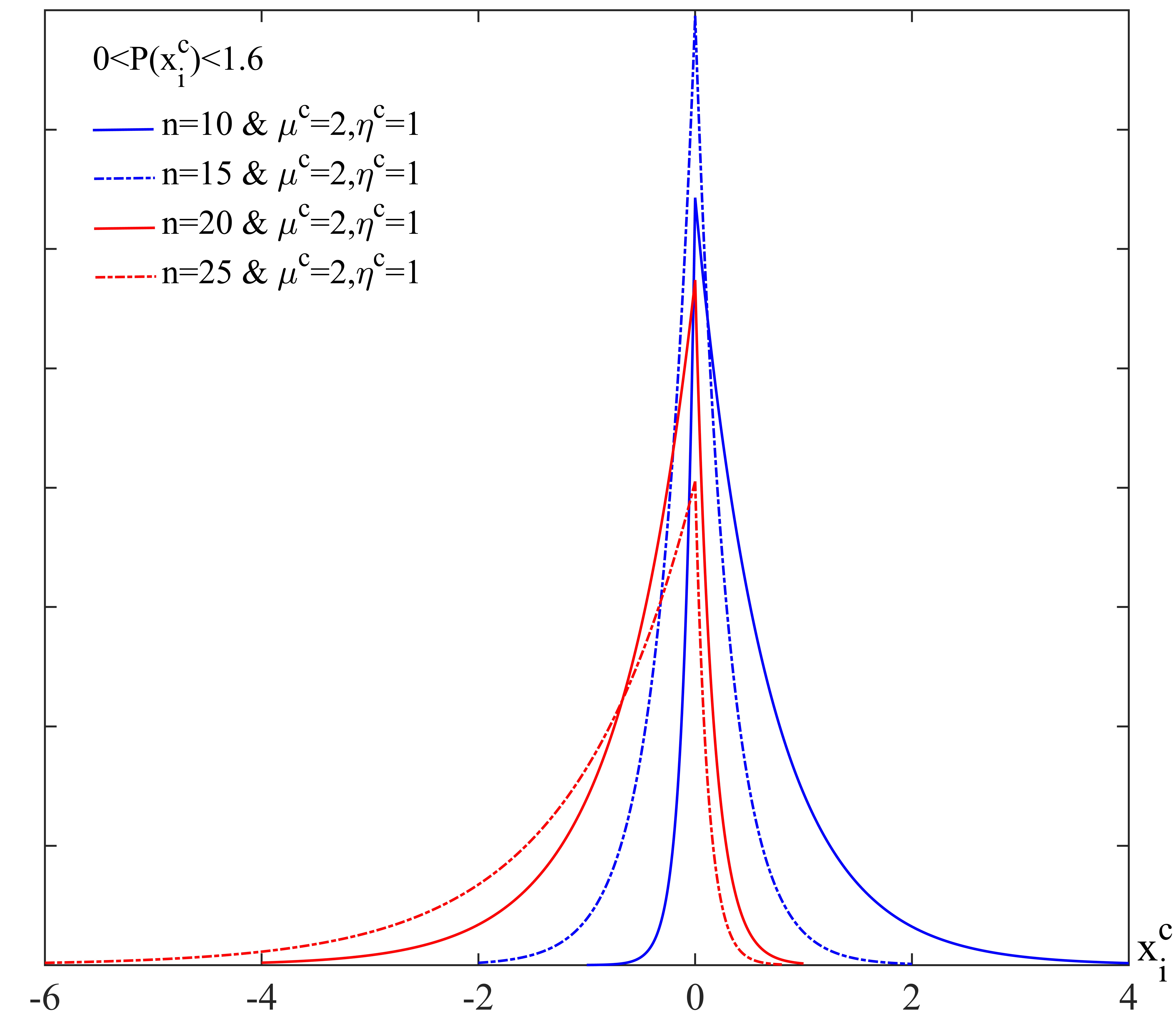}
\caption{\label{fig:F2}Cognitive Resource Distribution with Increasing $n$ and $\mu^c>\eta^c$}
\end{figure}
\begin{figure}[htbp]
\centering
\includegraphics[width=11cm]{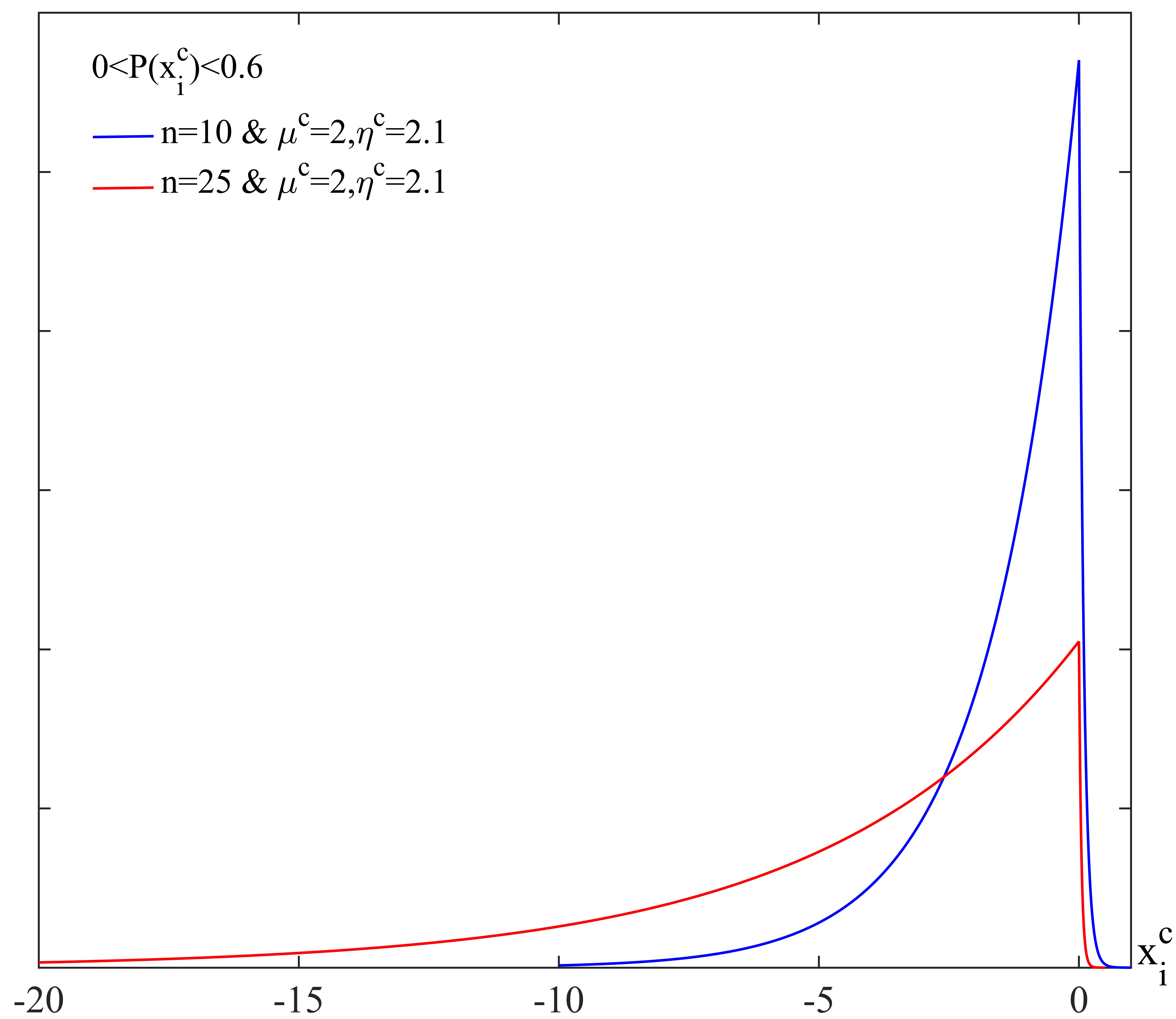}
\caption{\label{fig:F3}Cognitive Resource Distribution with Increasing $n$ and $\mu^c<\eta^c$}
\end{figure}

\newpage

Then, as can be seen from Figure 3, when the agent's cognitive resource recovery effect is weaker than the dilution effect, i.e., $\mu^c<\eta^c$, the agent's cognitive resource distribution is a stable left-skewed distribution, and with the increase of the big data size $n$, the peak and expectation value of this cognitive resource distribution both decrease\footnote{ The right tail of the agent's cognitive resource distribution in this case is weakly affected by increasing $n$, so the simulation results are presented only for $n=10$ and $n=25$ to clearly present the conclusion.}.

Therefore, according to the figures, regardless of the relative magnitude of $\mu^c$ and $\eta^c$, when the scale n of the agent's interaction with big data increases, the agent's cognitive resource distribution will be left-skewed eventually, and the expectation value of cognitive resources decreases.
\end{proof}

\subsection{Information Entropy and the Value of Big Data}
\textbf{Information Value and Data Value}: for this part, we will use a general variable that measures the uncertainty of information, i.e., information entropy, to measure the value of information and construct a measurement system of data value based on the value of information, and take the data value as an important variable for subsequent analysis. Our main ideas are as follows:

(\textbf{A}) The Space of Information: We represent the value of information in terms of three-dimensional vectors from the mathematical level, and this setting can be referred to \citet{angeletos2025inattentive} and \citet{caplin2022rationally}, who put the information variable or the state of decision problem in the three-dimensional space $\mathbb{R}^3$ for further analysis in their articles, and our paper refers to this. However, unlike them, we do not intend to use a particular three-dimensional vector as a proxy for the value of information, we are going to measure the value of information in the general case by setting up a special information value dimension, i.e., the three axes of 3D coordinates, and calculating the information value vector in the general case, i.e., an arbitrary vector in 3D space. Meanwhile, this definition also fits people's perception of the value of information. When people assess the value of information, they subconsciously consider it from multiple dimensions and weigh the importance of each dimension according to the specific situation. We believe that the value of information mainly comes from three dimensions, i.e. \textit{Timeliness}, \textit{Accuracy} and \textit{Relevance}, and under different application scenarios or user needs, the focus of information value is different, so the weights of different types of information in each value dimension vary greatly, showing diversity and complexity, which are similar to vectors in any direction in a three-dimensional space. Accordingly, we analogize the general value of information as a vector of arbitrary directions and the specific value as a vector coinciding with the coordinate axis.

(\textbf{B}) Information Values and Eigenvalues: We encode information vectors in $\mathbb{R}^3$ space with matrices, and the matrix elements encode the various types of attributes of information and their interrelationships. People use information to assist decision-making, which is actually a linear transformation of subjective decision-making to make decision-making more objective and rationalized\footnote{After encoding subjective decision-making, the eigenvalues of information vectors are used to perform a linear transformation of decision-making vectors to migrate the decision-making vectors in the direction of rationality, and the significance or effect of migration is up to the value of information.}. The eigenvalues of the matrix have good stability and invariance, which make it a reliable index to measure the value of information: Among various linear transformations, such as similar transformation, contractual transformation, etc., the eigenvalues of the matrix remain unchanged or have a specific transformation law, and this invariance makes the eigenvalues able to accurately capture the essential characteristics of the information, and are not affected by the change of the external manifestation of the information, so that the eigenvalues of the information matrix can accurately reflect the strength and characteristics of the linear transformation represented by the matrix in a specific direction, which is the main embodiment of the value of the information. And Eigenvalue Decomposition (EVD) is an effective tool for exploring the intrinsic structure of information, which reveals the distribution of information contained in the matrix in the direction of different eigenvectors, and the magnitude of the eigenvalue quantifies the relative importance of the information in the corresponding direction, so in the complex and changing information environment, the eigenvalue can stably measure the core value of the information, and provide a unified and reliable measurement standard for the comparison, analysis and integration of information. In terms of mathematical principles, eigenvalues can effectively extract the core features and inner structure of information matrix to quantify the value of information. Taking Principal Component Analysis (PCA) as an example, the size of eigenvalues in PCA reflects the richness of information carried by the components, i.e., the information value.

(\textbf{C}) Data Value: Data elements enable efficient, fast and decentralized dissemination of information by materializing diverse information. Now we set data as consisting of multiple independent pieces of information $\mathbb{D}\subseteq\{x_1,x_2,x_3...x_n\}$, and according to the information paradigm theory of \citet{stiglitz2000contributions}, the economic value of information stems from its ability to reduce decision uncertainty and make decisions more rational. Assuming that the decision maker faces a state space $\mathcal{S}$ with prior beliefs of a probability distribution $P(\omega)$, and that information $x_i$ updates beliefs through a posterior distribution $P(\omega|x_i)$, the value of information $x_i$ is reflected in the gain in the agent's expected utility, as $V_{\sigma_t}(x_i)=\mathbb{E}[U(P(\omega|x_i))]-\mathbb{E}[U(P(\omega))]$. $V_{\sigma_t}(x_i)$ demonstrates the additivity of the information value, so the total utility gain that the agent obtains by demanding the data elements is the sum of the values of each independent information that the data element $\mathbb{D}$ can provide, which is defined in this paper as the data value.

However, different information sources $x_i,i=\{1,2,3...n\}$ cannot be directly summed up due to the difference in magnitude, and the marginal utility of information aggregation may diminish with the increase in data size \citep{varian2004economics}, referring to the logic of standardized pricing of goods in general equilibrium theory \citep{arrow1954existence}, in order to build a unified data circulation mechanism, and to avoid overestimation of the value of large sample datasets due to simple summing, the normalization function $\phi(V)$ is used to achieve “dimensionless” and “scale-neutrality” of the accumulation of information value. In addition, linear normalization may ignore synergistic or antagonistic effects between information \citep{agrawal2019economics}, with synergistic effects reflecting the superadditivity of information, i.e., the combined value of information is more than the sum of the individuals, and antagonistic effects reflecting the increasing marginal cost of redundancy of information, with highly correlated information incurring an additional cleansing cost that reduces net utility. Neglecting these two kinds of utility will lead to inaccurate assessment of the total data value, therefore, in this paper, we will introduce the interaction strength of information based on linear normalization of data value to measure the gain and loss of data value caused by information synergy and antagonistic effect, respectively.

Now, based on the above theoretical analysis, we will specifically measure the information value and construct data value variables.

\begin{theorem}
\textit{The data value variable $D_t\in(0,1)$ measures the whole informational value state of society. When the information entropy of society is high, the data value is low and $D_t\rightarrow0$, when the information entropy of society is low, the data value is high and $D_t\rightarrow1$.}
\end{theorem}

\begin{proof}
Let $(\mathcal{S},\mathcal{F},\mathbb{P})$ be a complete probability space, where $\mathcal{S}$ is a sample space representing the set of all possible trial outcomes. $\mathcal{F}$ is a $\sigma-$ algebra on $\mathcal{S}$, which is a non-empty set class satisfying the closure of pairwise complementary and countable concatenation operations for defining measurable events. $\mathbb{P}:\mathcal{F}\rightarrow[0,1]$ is a probability measure assigning a probability value to each measurable event satisfying $\mathbb{P}(\mathcal{S})=1$, and for a sequence of mutually exclusive sequence of events $\{A_n\}_{n=1}^{\infty}\subseteq\mathcal{F}$, there exists $\mathbb{P}(\bigcup_{n=1}^{\infty}A_n)=\sum_{n=1}^{\infty}\mathbb{P}(A_n)$. Let $X:\mathcal{S}\rightarrow\mathbb{R}$ be a continuous-type random variable defined on this probability space with $\{\mathit{s}\in\mathcal{S}:X(\mathit{s})\in\mathcal{B}\}\subseteq\mathcal{F}$ for an arbitrary \textit{Borel} set $\mathcal{B}(\mathbb{R})$.The probability density function of the continuous-type random variable $X$, denoted $p_{X}:\mathbb{R}\rightarrow[0,+\infty)$, is a non-negative measurable function with respect to the \textit{Lebesgue} measure $\lambda$ satisfying the following conditions:
\[\begin{gathered}x\in\mathbb{R},p_X(x)\geq0,\lambda(\{x\in\mathbb{R}{:}p_X(x)\leq0\})=0\\\int_\mathbb{R}p_X(x)d\lambda(x)=1\\\mathbb{P}(X\in\mathbb{R})=1\\\mathbb{P}(X\in\mathcal{B}(\mathbb{R}))=\int_\mathbb{R}p_X(x)d\lambda(x)\end{gathered}\]

Denote the probability density of a continuous random variable $X$ as $p_X(x),x\in\mathbb{R}$ and the information entropy of $X$ is $h(X)$ as follows:
\begin{align}
h(X)=-\int_\mathbb{R}p_X(x)lnp_X(x)d\lambda(x)\label{10}
\end{align}

Now, according to function (10), we define the generation of information $x$ is a continuous random variable whose probability density is given by $p(x),x\in\mathbb{R}^+$, then the information entropy is $H(x)=\sigma_t=-\int_{\mathbb{R}}p(x)lnp(x)dx,\sigma_t\in(-\infty,+\infty)$ measures the degree of chaos of the information, and the larger the value indicates that the worse the quality of the information and the higher the uncertainty. Considering the information entropy as a set of vectors of information after combination, $\vec{\sigma}\in\mathbb{R}^3$, and $\mathbb{R}^3$ consists of three dimensions of \textit{Timeliness}, \textit{Accuracy} and \textit{Relevance}. It is defined that when the information entropy is high, the effective value that this combination of information vectors can provide is low, as measured by the vector operator $\vec{\sigma}_-$, and when the information entropy is low, the effective value that this combination of information vectors can provide is high, as measured by the vector operator $\vec{\sigma}_+$. There are three two-dimensional subspaces within the three-dimensional space $\mathbb{R}^3$ i.e. $\mathbb{R}_{xy}^2,\mathbb{R}_{yz}^2,\mathbb{R}_{xz}^2$. Before considering the value of $\vec{\sigma}$ in any direction in the three-dimensional space, analyze the value of $\vec{\sigma}$ in a particular direction within each two-dimensional subspace. Define the second-order matrix $A_i(i=x,y,z)$ as the value matrix of $\vec{\sigma}$ in a particular direction, $\vec{\sigma}_{i\pm}(i=x,y,z)$ as the vector operator of the value of $\vec{\sigma}$ in a particular direction, $\pm|\vec{\sigma}_{i-value}|(i=x,y,z)$ is the eigenvalue of $\vec{\sigma}$ in a particular direction, measuring the value of $\vec{\sigma}$.
\[\begin{gathered}A_x\vec{\sigma}_{x+}=+|\vec{\sigma}_{x-value}|\vec{\sigma}_{x+}\\A_x\vec{\sigma}_{x-}=-|\vec{\sigma}_{x-value}|\vec{\sigma}_{x-}\end{gathered}\]

$\vec{\sigma}_{x+}$ is orthogonal to $\vec{\sigma}_{x-}$, and the vector is described by a unit matrix consisting of the two simplest components:
\[\vec{\sigma}_{x+}=(1,0)^T,\vec{\sigma}_{x-}=(0,1)^T\]

Then, we get:
\[A_{x}=[\vec{\sigma}_{x+}\quad\vec{\sigma}_{x-}]\begin{bmatrix}+|\vec{\sigma}_{x-value}|&0\\0&-|\vec{\sigma}_{x-value}|\end{bmatrix}[\vec{\sigma}_{x+}\quad\vec{\sigma}_{x-}]^{-1}\]

That is:
\[A_{x}=\begin{bmatrix}1&0\\0&1\end{bmatrix}\begin{bmatrix}+|\vec{\sigma}_{x-value}|&0\\0&-|\vec{\sigma}_{x-value}|\end{bmatrix}\begin{bmatrix}1&0\\0&1\end{bmatrix}\]

Since the three two-dimensional subspaces $\mathbb{R}_{xy}^2,\mathbb{R}_{yz}^2,\mathbb{R}_{xz}^2$ exist in $\mathbb{R}^3$ have orthogonal relationship with each other, when the space $x$, space $y$, and space $z$ are orthogonal to each other, we can get the vector operator in the other space by orthogonal decomposition of the vector base in one space.

\[\begin{gathered}\vec{\sigma}_{y+}=\frac{\sqrt{2}}{2}\vec{\sigma}_{x+}+\frac{\sqrt{2}}{2}\vec{\sigma}_{x-}=(\frac{\sqrt{2}}{2},\frac{\sqrt{2}}{2})^T\\\vec{\sigma}_{y-}=-\frac{\sqrt{2}}{2}\vec{\sigma}_{x+}+\frac{\sqrt{2}}{2}\vec{\sigma}_{x-}=(-\frac{\sqrt{2}}{2},\frac{\sqrt{2}}{2})^T\\A_y=[\vec{\sigma}_{y+}\quad\vec{\sigma}_{y-}]\begin{bmatrix}+|\vec{\sigma}_{y-value}|&0\\0&-|\vec{\sigma}_{y-value}|\end{bmatrix}[\vec{\sigma}_{y+}\quad\vec{\sigma}_{y-}]^{-1}\\A_y=\begin{bmatrix}\frac{\sqrt2}{2}&-\frac{\sqrt2}{2}\\\frac{\sqrt2}{2}&\frac{\sqrt2}{2}\end{bmatrix}\begin{bmatrix}+|\vec{\sigma}_{y-value}|&0\\0&-|\vec{\sigma}_{y-value}|\end{bmatrix}\begin{bmatrix}\frac{\sqrt2}{2}&\frac{\sqrt2}{2}\\-\frac{\sqrt2}{2}&\frac{\sqrt2}{2}\end{bmatrix}\end{gathered}\]

Then:
\[\begin{gathered}\vec{\sigma}_{z+}=\frac{\sqrt{2}}{2}\vec{\sigma}_{x+}+\frac{\sqrt{2}i}{2}\vec{\sigma}_{x-}=(\frac{\sqrt{2}}{2},\frac{\sqrt{2}i}{2})^T\\\vec{\sigma}_{z-}=-\frac{\sqrt{2}}{2}\vec{\sigma}_{x+}+\frac{\sqrt{2}i}{2}\vec{\sigma}_{x-}=(-\frac{\sqrt{2}}{2},\frac{\sqrt{2}i}{2})^T\\A_z=[\vec{\sigma}_{z+}\quad\vec{\sigma}_{z-}]\begin{bmatrix}+|\vec{\sigma}_{z-value}|&0\\0&-|\vec{\sigma}_{z-value}|\end{bmatrix}[\vec{\sigma}_{z+}\quad\vec{\sigma}_{z-}]^{-1}\\A_{z}=\begin{bmatrix}\frac{\sqrt{2}}{2}&-\frac{\sqrt{2}}{2}\\\frac{\sqrt{2}i}{2}&\frac{\sqrt{2}i}{2}\end{bmatrix}\begin{bmatrix}+|\vec{\sigma}_{z-value}|&0\\0&-|\vec{\sigma}_{z-value}|\end{bmatrix}\begin{bmatrix}\frac{\sqrt{2}}{2}&-\frac{\sqrt{2}i}{2}\\-\frac{\sqrt{2}}{2}&-\frac{\sqrt{2}i}{2}\end{bmatrix}\end{gathered}\]

Now, we get:
\[\begin{gathered}A_{x}=|\vec{\sigma}_{x-value}|\begin{bmatrix}1&0\\0&-1\end{bmatrix},A_y=|\vec{\sigma}_{y-value}|\begin{bmatrix}0&1\\1&0\end{bmatrix},A_z=|\vec{\sigma}_{z-value}|\begin{bmatrix}0&i\\-i&0\end{bmatrix}\end{gathered}\]

Perform a linear combination of $A_x,A_y,A_z$ to analyze the value $\vec{\sigma}_{d-value}$ of the information on any direction $d=(a,b,c)$ within $\mathbb{R}^3$ ($d$ needs to satisfy the normalization requirement):
\[\begin{gathered}A_{d}\vec{\sigma}_{d+}=+|\vec{\sigma}_{d-value}|\vec{\sigma}_{d+}\\A_d\vec{\sigma}_{d-}=-|\vec{\sigma}_{d-value}|\vec{\sigma}_{d-}\end{gathered}\]

Then:
\[\begin{gathered}M_x=\begin{bmatrix}1&0\\0&-1\end{bmatrix},M_y=\begin{bmatrix}0&1\\1&0\end{bmatrix},M_z=\begin{bmatrix}0&i\\-i&0\end{bmatrix}\\a^2+b^2+c^2=1\\A_{d}=|\vec{\sigma}_{d-value}|(aM_{x}+bM_{y}+cM_{z})=|\vec{\sigma}_{d-value}|\begin{bmatrix}c&a-bi\\a+bi&-c\end{bmatrix}\\|\vec{\sigma}_{d-value}|=|0\pm\sqrt{a^2+b^2+c^2}|=|\pm1|\end{gathered}\]

Hence:
\[\begin{gathered}A_d\vec{\sigma}_{d+}=\{+[1]\}\vec{\sigma}_{d+}\\A_d\vec{\sigma}_{d-}=\{-[1]\}\vec{\sigma}_{d-}\end{gathered}\]

It follows that a combination of information in any direction $\vec{\sigma}\in\mathbb{R}^3$ when measured by the vector operator $\vec{\sigma}_{d-}$ has a value $-|\vec{\sigma}_{d-value}|=-1$, and when measured by the vector operator $\vec{\sigma}_{d+}$ has a value $+|\vec{\sigma}_{d-value}|=+1$. Thus, the information entropy $\sigma_t\in(-\infty,+\infty)$ represents a combination of a set of information, and when the information entropy is larger $\sigma_t\rightarrow+\infty$(Extreme uncertainty, $\sigma_t$ degenerates to the theoretical maximum entropy value), at which point that combination of information provides very little effective value and is measured by “$-1$”. When the information entropy is smaller $\sigma_t\rightarrow-\infty$(Extreme certainty, $\sigma_t$ degraded to 0, similar to the \textit{Dirac Delta} function), at this time the information combination can provide the effective value is extremely high, with “$+1$” to measure. In reality, however, the distribution of a random variable $x\in\mathcal{X}$ cannot completely converge to the \textit{Dirac Delta} function, and its support set $\mathcal{X}$ is restricted to the interval of minimum width $\varepsilon>0$. That is, when $x$ obeys the uniform distribution $x\sim U[a,a+\varepsilon]$, its differential entropy $H(x)=\sigma_{min}=log\varepsilon$. Since both negative entropy and zero entropy characterize the case of extreme certainty of the information, in this paper, we set $\varepsilon=1$, then when the theoretical entropy of the information is small $\vec{\sigma}_t\rightarrow-\infty$, $\sigma_t$ degenerates into $\sigma_{min}=0$, which restricts that $H(x)\geq0$. In addition, for the fixed support set $\mathcal{X}$, when $x$ obeys the Gaussian distribution $x\sim \mathcal{N}(\mu,v^2)$, the maximum entropy $H(x)=\sigma_{max}$ is reached, and the Lagrange equation is constructed by using the conditions $\int p(x)dx=1$ and $\int(x-\mu)^2p(x)dx=v^2$, then, use the calculus of variations, the maximum entropy $H(x)=\sigma_{max}=\frac{1}{2}ln(2\pi ev^2)$ can be found, so that, when the theoretically greater information entropy $\sigma_t\rightarrow+\infty$, $\sigma_t$ degenerates to $\sigma_{max}$. Accordingly, normalizing the differential entropy to the interval $[\sigma_{min},\sigma_{max}]$ and mapping it to $V_{\sigma_t }\in[-1,+1]$, we get the information value function $V_{\sigma_t }$:
\[\begin{gathered}V_{\sigma_t}(x)=1-2\frac{\sigma_t(x)}{\sigma_{max}}\\\lim_{\sigma_t\to-\infty}V_{\sigma_t}(x)=V_{\sigma_t}(\sigma_t(x)=\sigma_{min})=V_{max}=+1\\\operatorname*{lim}_{\sigma_{t}\to+\infty}V_{\sigma_{t}}(x)=V_{\sigma_{t}}(\sigma_{t}(x)=\sigma_{max})=V_{min}=-1\end{gathered}\]

Data elements $\mathbb{D}\subseteq\{x_1,x_2,x_3...x_n\}$, information value $V_{\sigma_t}(x_i)\in[-1,1]$, and a linear normalization function $\phi(V_{\sigma_t})=\frac{1}{2}(V_{\sigma_t}+1)$ mapping information value $V_{\sigma_t}(x_i)$ to data value $D(\sigma_t)$, with synergistic coefficients $a_{ij}\geq0$, and antagonistic coefficients $b_{ij}\geq0$. Describing the information combinations $\{x_i,x_j\}\subset\mathbb{D}(i,j\in\{1,2,3...n\},i\neq j)$ of the interactions, then the joint value addition of information $\Delta_a=a_{ij}\phi(V_{\sigma_t}(x_i))\phi(V_{\sigma_t} (x_j))$ and the joint value loss $\Delta_b=b_{ij}\phi(V_{\sigma_t}(x_i))\phi(V_{\sigma_t}(x_j))$, setting the interaction strength coefficient $J\geq0$, the data value $D(\sigma_t)$ is as follows:
\[D(\sigma_t)=\frac{1}{n}\sum_{i=1}^n\phi(V_{\sigma_t}(x_i))+J\left(\frac{\sum_{i<j}(a_{ij}-b_{ij})\phi(V_{\sigma_t}(x_i))\phi(V_{\sigma_t}(x_j))}{n(n-1)/2}\right)\]

Since the data value variable symbolizes the whole value state of the information accumulating with time in society, we set data value variable $D_t$ is the sum of $D(\sigma_t)$, i.e., $D_t=\sum D(\sigma_t)$. And the output range of $\sum D(\sigma_t)$ is compressed by the \textit{Sigmoid} function: 

\[\mathbb{T}(D(\sigma_t))=\frac{1}{1+\exp(-(\sum D(\sigma_t)))}\]

Then we get:
\[D_t=\mathbb{T}(D(\sigma_t))\in(0,1)\begin{cases}\lim_{\sum D(\sigma_t)\to+\infty}D_t\to1\\\lim_{\sum D(\sigma_t)\to-\infty}D_t\to0\end{cases}\]
\end{proof}

Theorem 1, Theorem2 and Theorem3 reveal and conclude the prerequisites of our paper: (\textbf{A}) Big data value variable is measured by $(0,1)$ in our analysis framework. (\textbf{B}) The rationality of agents is decided by the cognitive resources they have. The big data, which accumulates with continuous time, will profoundly dilute the cognitive resources of agents and decrease the probability that they make rational decisions if the big data interactions of agents exist. 

\section{The Model: Consumption and Utility Acquisition}
For the third section, we specifically analyze and prove the amount and direction of agent's consumption adjustment when the agent interacts with big data, and construct a consumption adjustment weight function (CAWF) based on the obtained conclusions. Applying the CAWF we can find that: When a big data interaction exists, the agent's cognitive resources are diluted, and the decision-making of this agent will be irrational. This ultimately results in the fact that the agent's effective consumption, which is able to acquire utility, will become a weight of total consumption.

\subsection{The Amount of Consumption Adjustment with Uncertainty}

\textbf{Prospect Theory}: According to \citet{kai1979prospect}, most people are risk-averse when they are faced with a profit situation, preferring small definite returns. Whereas, when they are faced with a loss situation, they are risk-averse, wishing to avoid the loss as much as possible. Based on this, we set the information that the agent can receive with the data elements as a binary state variable $w\in\{0,1\}$, and the agent's consumption decision $s\in\mathcal{S}$ is determined by the likelihood function $p(s|w)$.

We assume that there are two types of agents, one type of agent whose cognitive resources have not been diluted by big data, so its consumption decision is rational, and the other type of agent whose cognitive resources have been significantly diluted by big data for a long period of time, so its consumption decision is irrational. The rational agent can accurately judge the direction of consumption adjustment and the exact amount of adjustment according to the information, which is defined as “Bayesian Agent”, and its consumption decision is given by the set $s=(s_d,s_q)$. Where $s_d$ determines the direction of consumption adjustment, whether it is an increase or a decrease, and $s_q\in\mathbb{R}$ determines the exact amount of change in consumption, the magnitude of $s_q$ is determined by the information uncertainty. The Bayesian agent's decision about the amount of consumption change is rational: When information uncertainty is high (information entropy $\sigma_t$ is large), the difference between $p(s|w=1)$ and $p(s|w=0)$ is small, and the agent's change in consumption will be small and unbiased. When the information uncertainty is low (the information entropy $\sigma_t$ is low), the difference between $p(s|w=1)$ and $p(s|w=0)$ is very large, and the agent's decision based on the information is relatively certain about its own benefit, i.e., it is rational, so its change in consumption will increase. Referring to \citet{augenblick2025overinference}, set $f(x)=ln(\frac{x}{1-x})$, the Bayesian agent's consumption is $c(s)$, the baseline consumption is $c_0$, and the consumption adjustment term is $\mathbb{S}(s)$, which is determined by the uncertainty of information:
\[\begin{gathered}\underbrace{f(c(s))}_{Posterior\ Consume}=\underbrace{f(c_0)}_{Prior\ Consume}\quad\underbrace{\pm}_{s_d}\quad\underbrace{S(s)}_{s_q|s_d}\\\mathbb{S}(s)=\left|ln\left(\frac{p(s|w=1)}{p(s|w=0)}\right)\right|\propto\frac{1}{\sigma_t}\\p(s|w=1)+p(s|w=0)=1,p(s|w=1)\geq p(s|w=0)\end{gathered}\]

On the contrary, although irrational agents also have the awareness of judging the economic situation based on the information of the current data and thus assisting the consumption decision, they are not rational enough. They can judge whether the consumption environment is good or bad based on the information, therefore, they can be aware of whether consumption needs to be increased or decreased, but since they do not inquire more deeply into the truth or falsity of the information, i.e., subjectively, the benefit they gain through the adjustment amount decision they make based on any piece of information is equiprobable, so the irrational agent does not know exactly how much consumption should be adjusted. In this paper, we define an irrational agent as a non-Bayesian Agent, this type of agent knows how to adjust the direction of consumption based on the information but does not know the exact amount that their consumption should be adjusted. The consumption adjustment decision of this agent is $\widehat{\mathbb{S}}(\hat{s})$, which is given by the set $\hat{s}=(s_d,s_n),s_n\in\mathbb{R}$. And satisfies three types of conditions: (\textbf{A}) $\mathbb{E}[s_n|\mathbb{S}]=\mathbb{S}$, indicating that the non-Bayesian agent's estimate of $s_n$ fluctuates in a statistically significant way around the true consumption adjustment, $\mathbb{S}$, and its expectation is equal to the true consumption adjustment $\mathbb{S}$. (\textbf{B}) $P(\mathbb{S}|s_n)\neq1$, i.e., the adjustment to consumption, $\mathbb{S}$, is not degenerate in any case, and we cannot infer the complete true adjustment to consumption accurately by estimating $s_n$. (\textbf{C}) $\mathbb{S}_2>\mathbb{S}_1,\frac{\partial}{\partial s_n}\left(\frac{p(s_n|s_d,\mathbb{S}=\mathbb{S}_2)}{p(s_n|s_d,\mathbb{S}=\mathbb{S}_1}\right)>0$, i.e., there is an ordered correspondence between the magnitude of the adjustment of consumption, $s_n$, and the magnitude of the true adjustment of consumption, $\mathbb{S}$, such that the higher the estimation of $s_n$ is, the larger the adjustment of consumption is, and vice versa. Accordingly, we propose Theorem 4.

\begin{theorem}
\textit{Compared to rational agents, irrational agents will decide to underestimate the amount of consumption adjustment when information uncertainty is low (lower information entropy). And they will decide to overestimate the amount of consumption adjustment when information uncertainty is high (higher information entropy).}
\end{theorem}

\begin{proof}
In this paper, we argue that, depending on the individual cognitive resources subject to different dilution effects of big data of different scales, the consumption adjustment of non-Bayesian agents is heterogeneous, therefore, they do not always update the exact amount of their consumption adjustment to a certain size of $s_n$, they adjust the value of $s_n$ at any time according to the scale of big data and the generation of information which they have received. Referring to the findings of \citet{augenblick2025overinference} and \citet{chambers2012updating}, it is set that a non-Bayesian agent's estimation of the economic situation before acquiring any information in the present can generate a subjective decision $\widehat{\mathbb{S}}_0$ about consumption, and when the agent understands after acquiring the information whether it is advisable to reduce or increase consumption, his decision is updated to $\widehat{\mathbb{S}}(s_d)$ firstly. Since this type of agents are irrational in their judgment of consumption and cannot accurately perceive what consumption $s_n$ should be adjusted to, hence, the true consumption adjustment decision, $\widehat{\mathbb{S}}(\hat{s})$, which consists of the non-Bayesian agent's actual change in the direction and the exact amount of the consumption, will between $\widehat{\mathbb{S}}(s_d)$ and $s_n$:
\begin{align}
\widehat{\mathbb{S}}(\hat{s})=\varepsilon s_n+(1-\varepsilon)\widehat{\mathbb{S}}\left(s_d\right),\varepsilon\in(0,1)\label{(11)}
\end{align}

Therefore, when receiving the same type of information $s=s_0$ at the social level, the Bayesian agent will take the error-free consumption adjustment decision $s=(s_d,s_q)$, whose complete consumption adjustment magnitude is $\mathbb{S}(s)=\left|ln\left(\frac{p(s|w=1)}{p(s|w=0)}\right)\right|$, which stems from the type of agent's accurate and rational judgment of the consumption adjustment magnitude $s_q$. On the other hand, non-Bayesian agents take the consumption adjustment decision $\hat{s}=(s_d,s_n)$, whose consumption adjustment magnitude needs to satisfy function $(11)$, which is due to the fact that the consumption adjustment term $s_n$ of this type of agent will be endogenous to the cognitive resources that are subjected to the dilution of big data. In this paper, we set that the Bayesian agent's consumption adjustment size $\mathbb{S}(s)$ is the criterion for judging whether the consumption decision is rational or not, and the non-Bayesian agent's consumption adjustment size will always deviate from $\mathbb{S}(s)$, which in turn produces over-adjustment or under-adjustment decisions of the consumption size, $\widehat{\mathbb{S}}(\hat{s})$, i.e., $s_n$ always fluctuates around $s_q$. Specifically: when receiving information $s=s_0$, if $\mathbb{E}[\widehat{\mathbb{S}}(\hat{s})|s]>\mathbb{S}(s)$, the non-Bayesian agent overestimates the size of adjustment to consumption, and if $\mathbb{E}[\widehat{\mathbb{S}}(\hat{s})|s]<\mathbb{S}(s)$, the non-Bayesian agent underestimates the size of adjustment to consumption.

Setting that both the true consumption adjustment S and the estimate of the consumption adjustment $s_n$ obey a lognormal distribution when the message $s=s_0$ is received, and the estimate $s_n$ fluctuates around the true consumption adjustment S, i.e., $\mathbb{E}[s_n|\mathbb{S}]=\mathbb{S}$, then:
\[\begin{gathered}ln\mathbb{S}{\sim}\mathcal{N}(\mu_{\mathbb{S}},\sigma_{\mathbb{S}}^{2})\\lns_n{\sim}\mathcal{N}\left(ln\mathbb{S}-\frac{\sigma_n^2}{2},\sigma_n^2\right)\\ln\widehat{\mathbb{S}}(\hat{s}){\sim}\mathcal{N}\left(\mathbb{E}[ln\widehat{\mathbb{S}}(\hat{s})|s],\mathrm{Var}[ln\widehat{\mathbb{S}}(\hat{s})|s]\right)\end{gathered}\]

And:
\[\begin{gathered}ln\mathbb{E}[s_n]=\mathbb{E}[lns_n]+\frac{\sigma_n^2}{2}\\\widehat{\mathbb{S}}(s_d)=\mathbb{E}[\mathbb{S}|s_d]=\exp\left(\mu_\mathbb{S}+\frac{\sigma_\mathbb{S}^2}{2}\right)\end{gathered}\]

We get:
\[\begin{gathered}ln\widehat{\mathbb{S}}(\hat{s})=\left(\frac{\sigma_{\mathbb{S}}^{2}}{\sigma_{\mathbb{S}}^{2}+\sigma_{n}^{2}}ln(s_{n})+\left(1-\frac{\sigma_{\mathbb{S}}^{2}}{\sigma_{\mathbb{S}}^{2}+\sigma_{n}^{2}}\right)ln(\widehat{\mathbb{S}}(s_{d}))\right)\\ln\widehat{\mathbb{S}}(\hat{s})=\left(\underbrace{\frac{\sigma_{\mathbb{S}}^{2}}{\sigma_{\mathbb{S}}^{2}+\sigma_{n}^{2}}}_{Weight\ on\ estimate}\times\underbrace{\left(lns_{n}+\frac{\sigma_{n}^{2}}{2}\right)}_{LN\ Adjusted\ estimate}+\underbrace{\left(1-\frac{\sigma_{\mathbb{S}}^{2}}{\sigma_{\mathbb{S}}^{2}+\sigma_{n}^{2}}\right)}_{Weight\ on\ prior}\times\underbrace{ln\exp\left(\mu_{\mathbb{S}}+\frac{\sigma_{\mathbb{S}}^{2}}{2}\right)}_{LN\ Adjusted\ prior}\right)\end{gathered}\]

We set $b\in\{Bayes,non-Bayes\},\beta^b=\frac{\sigma_{\mathbb{S}}^2}{\sigma_{\mathbb{S}}^2+\sigma_n^2}$, then: 
\[ln\widehat{\mathbb{S}}(\hat{s})=\left(\beta^b\left(lns_n+\frac{\sigma_n^2}{2}\right)+(1-\beta^b)ln\exp\left(\mu_{\mathbb{S}}+\frac{\sigma_{\mathbb{S}}^2}{2}\right)\right)\]

Hence:
\[\begin{gathered}\mathbb{E}\left[ln\widehat{\mathbb{S}}(\widehat{s})|s\right]=\beta^{b}\mathbb{E}\left[\left(lns_n+\frac{\sigma_{n}^{2}}{2}\right)|s\right]+(1-\beta^{b})\mathbb{E}\left[\left(ln\exp\left(\mu_{\mathbb{S}}+\frac{\sigma_{\mathbb{S}}^{2}}{2}\right)\right)|s\right]\\\mathbb{E}[ln\widehat{\mathbb{S}}(\hat{s})|s]=\beta^b\left(ln\mathbb{S}-\frac{\sigma_n^2}{2}+\frac{\sigma_n^2}{2}\right)+(1-\beta^b)\left(ln\exp\left(\mu_\mathbb{S}+\frac{\sigma_\mathbb{S}^2}{2}\right)\right)\\\mathbb{E}[ln\widehat{\mathbb{S}}(\hat{s})|s]=\left(\beta^b(ln\mathbb{S})+(1-\beta^b)\left(ln\exp\left(\mu_\mathbb{S}+\frac{\sigma_\mathbb{S}^2}{2}\right)\right)\right)\end{gathered}\]

Given:
\[\begin{gathered}\mathrm{Var}[ln\widehat{\mathbb{S}}(\hat{s})|s]=(\beta^b)^2\mathrm{Var}[lns_n|s]=(\beta^b)^2\sigma_n^2\\ln\widehat{\mathbb{S}}(\hat{s})\thicksim\mathcal{N}\left(\left(\beta^b(ln\mathbb{S})+(1-\beta^b)\left(ln\exp\left(\mu_\mathbb{S}+\frac{\sigma_\mathbb{S}^2}{2}\right)\right)\right),(\beta^b)^2\sigma_n^2\right)\end{gathered}\]

Hence:
\[\begin{gathered}\widehat{\mathbb{S}}(\hat{s})=\mathbb{E}[\widehat{\mathbb{S}}(\hat{s})|s]=\exp\left(\left(\beta^b(ln\mathbb{S})+(1-\beta^b)\left(ln\exp\left(\mu_\mathbb{S}+\frac{\sigma_\mathbb{S}^2}{2}\right)\right)\right)+\frac{(\beta^b)^2\sigma_n^2}{2}\right)\\\widehat{\mathbb{S}}(\hat{s})=\exp\left(\beta^b(ln\mathbb{S})\right)\times\exp\left((1-\beta^b)\left(ln\exp\left(\mu_\mathbb{S}+\frac{\sigma_\mathbb{S}^2}{2}\right)\right)\right)\times\exp\left(\frac{(\beta^b)^2\sigma_n^2}{2}\right)\\\widehat{\mathbb{S}}(\hat{s})=\mathbb{S}^{\beta^b}\times\left(\exp\left(\mu_\mathbb{S}+\frac{\sigma_\mathbb{S}^2}{2}\right)\right)^{(1-\beta^b)}\times\exp\left(\frac{(\beta^b)^2\sigma_n^2}{2}\right)\end{gathered}\]

Let $\mu^b=\left(\exp\left(\mu_\mathbb{S}+\frac{\sigma_\mathbb{S}^2}{2}\right)\right)^{(1-\beta^b)}\times\exp\left(\frac{(\beta^b)^2\sigma_n^2}{2}\right)$, then: 
\[\begin{gathered}\widehat{\mathbb{S}}(\hat{s})=\mu^b\mathbb{S}^{\beta^b}\\ln(\widehat{\mathbb{S}}(\hat{s}))=\beta^bln(\mathbb{S})+ln(\mu^b)\end{gathered}\]

For Bayesian agents, since they always make rational decisions that make consumption adjustments rational and efficient, so the consumption adjustments from their decisions are unbiased and the variance of the estimation is 0 with $\beta^{Bayes}=1,\mu^{Bayes}=0$:

\[ln(\mathbb{S}(\hat{s}))=ln(\mathbb{S})\]

Since the magnitude of consumption adjustment of non-Bayesian agents will always deviate from $\mathbb{S}(s)$, i.e., if $0<t_1<1$, $ln(\widehat{\mathbb{S}}(\hat{s}))_{low}=ln(t_1\times\mathbb{S}(s))=ln(\mathbb{S})+ln(t_1)<ln(\mathbb{S}(s))$, at this time, the non-Bayesian agent underestimates the consumption that should be adjusted. If $1<t_2$, $ln(\widehat{\mathbb{S}}(\hat{s}))_{high}=ln(t_2\times\mathbb{S}(s))=ln(\mathbb{S})+ln(t_2)>ln(\mathbb{S}(s))$, at this time the non-Bayesian agent overestimates the consumption that should be adjusted. 

Accordingly, we set the parameters $\beta^{non-Bayes}=0.8,\mu^{non-Bayes}=0.9$, and the simulation is shown in Figure 4:
\begin{figure}[htbp]
\centering
\includegraphics[width=11cm]{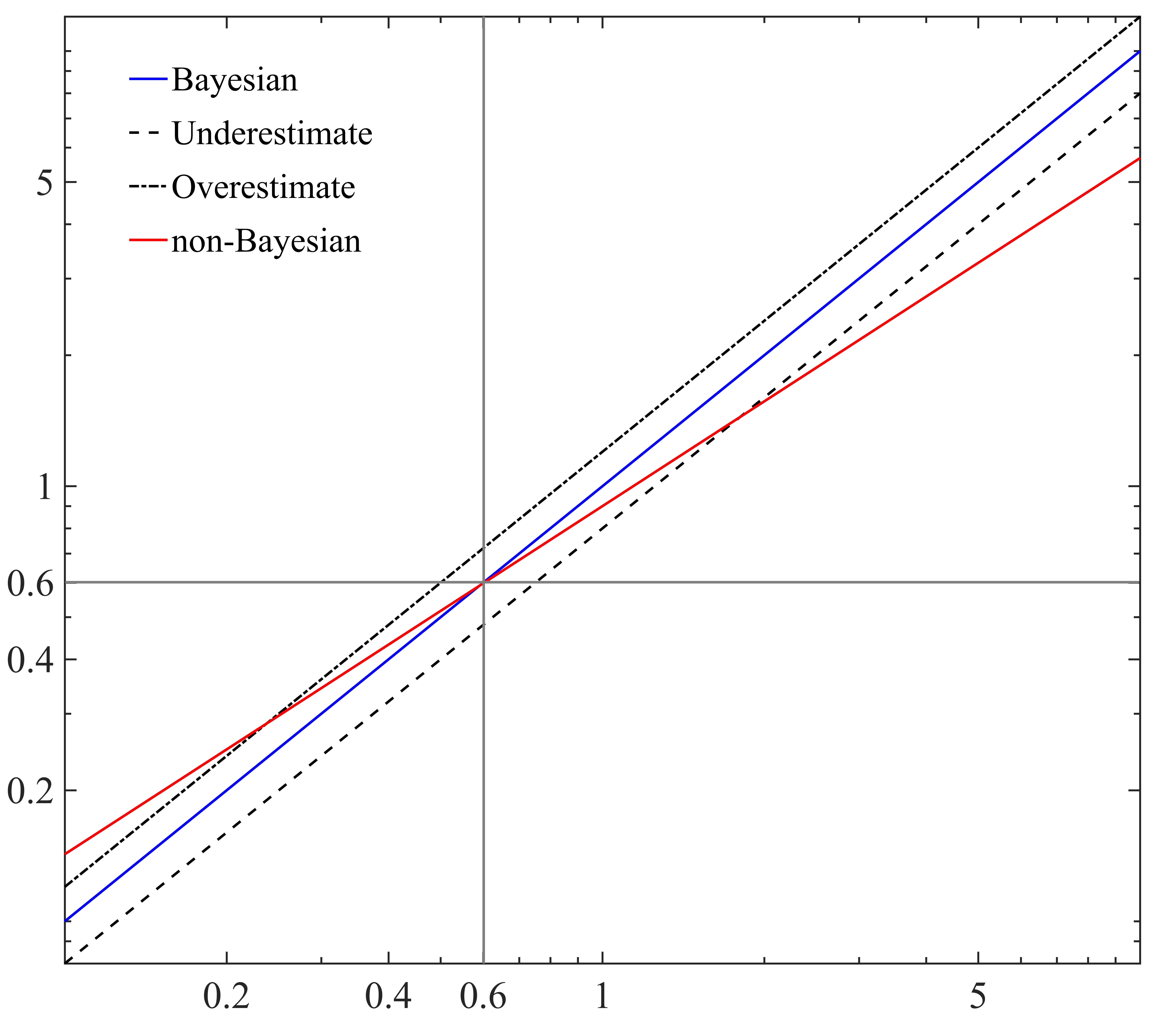}
\caption{\label{fig:F4}Consumption Adjustment Amount with Different Uncertainty}
\end{figure}

In Figure 4, the vertical axis is $ln(\widehat{\mathbb{S}}(\hat{s}))$ and the horizontal axis is $ln(\mathbb{S}(s))$. Therefore, when $0<ln(\mathbb{S})<0.6$, the non-Bayesian agent's consumption adjustment $ln(\widehat{\mathbb{S}}(\hat{s}))$ will be higher than $ln(\mathbb{S}(s))$, which approaches towards the $ln(\widehat{\mathbb{S}}(\hat{s}))_{high}$. And when $ln(\mathbb{S}(s))>0.6$, the consumption adjustment $ln(\widehat{\mathbb{S}}(\hat{s}))$ will be lower than $ln(\mathbb{S}(s))$, which approaches towards the $ln(\widehat{\mathbb{S}}(\hat{s}))_{low}$.

To summarize: $\mathbb{S}(s)=\left|ln\left(\frac{p(s|w=1)}{p(s|w=0)}\right)\right|$ is measured by the magnitude of information uncertainty, when the difference between $p(s|w=1)$ and $p(s|w=0)$ is small, the degree of information uncertainty is high, the information entropy is larger, at this time, the value of $ln(\mathbb{S}(s))$ is small, and the non-Bayesian agent tends to choose $ln(\widehat{\mathbb{S}}(\hat{s}))_{high}$, overestimate the consumption adjustment. When the difference between $p(s|w=1)$ and $p(s|w=0)$ is large, the degree of information uncertainty is low, the information entropy is low, at this time, the value of $ln(\mathbb{S}(s))$ is large, the non-Bayesian agent tends to choose $ln(\widehat{\mathbb{S}}(\hat{s}))_{low}$, underestimate the consumption adjustment.
\end{proof}

\subsection{The Direction of Consumption Adjustment with Tax Rate}
\textbf{Entrepreneur and Worker Model}: In this section, we are going to introduce government behavior. We assume that government tax revenue reflects the economic emotion and situation of society. When the economic situation remains prosperous, economic uncertainty is low, and tax revenue is low. Conversely, when the economic situation is poor and economic uncertainty is high, tax revenue is high. We further assume that the government acts rationally, meaning that it adjusts tax rates to ensure relatively stable tax revenue regardless of whether the economic environment is favorable or unfavorable. Therefore, government employees' incomes are stable, and they receive wages from total tax revenue to consume $C_{Gov}$. For entrepreneur investment agents, we assume that within a continuous phase $t$, agents will be impacted by big data interactions, at the beginning of this phase, agents' cognitive resources have not yet been diluted, and their decisions will be rational, making them Bayesian agents with a consumption level of $C_{Bayes}$. At the end of the phase, the agent's cognitive resources are diluted to equilibrium, with a high degree of dilution, leading to irrational decisions as a non-Bayesian agent, with a consumption level of $C_{non-Bayes}$.

Referring to the studies of \citet{pastor2020political} and \citet{pastor2016income}, we assume that the above agents have similar preferences for final consumption $U_t$:
\[U_t(C_{b,t})=\frac{(C_{b,t})^{1-\gamma_b}}{1-\gamma_b}\]

Among these, $C_{b,t}$ represents the consumption level of agent $b$ at stage $t$, where $b\in\{Bayes,non-Bayes\}$, and $\gamma_b$ denotes the risk aversion coefficient. When $\gamma_b>1$, it corresponds to CRRA utility, when $\gamma_b=1$, it corresponds to Log utility. Defining Bayesian and non-Bayesian agents due to the varying degrees to which their cognitive resources are diluted, agents exhibit heterogeneity in their cognitive abilities. Agent $b$ is endowed with a cognitive ability level $\mu_b$, which follows $\mu_b{\sim}\mathcal{N}(0,\sigma_{\mu}^2)$. Thus, agents with higher cognitive abilities can make rational decisions, thereby increasing the output from their investments and achieving higher consumption levels.

The agent obtains output $Y_{b,t+1}$ through investment:
\begin{align}
Y_{b,t+1}=e^{\mu_b+\varepsilon_{t+1}+\varepsilon_{b,t+1}}G_t\label{12}
\end{align}

For an agent $b$, all shocks are independent and identically distributed. $\varepsilon_{t+1}{\sim}\mathcal{N}(-\frac{1}{2}\sigma^2,\sigma^2)$ is the aggregate shock. $\varepsilon_{b,t+1}{\sim}\mathcal{N}(-\frac{1}{2}\sigma_{1}^2,\sigma_{1}^2)$ is the heterogeneous shock. Therefore, we can get that $\mathbb{E}_t[e^{\varepsilon_{t+1}}]=\mathbb{E}_t[e^{\varepsilon_{b,t+1}}]=1$. $G_t$ is the government's contribution to output. Each agent holds assets $Y_{b,t+1}(1-\tau_t)$ at the beginning of the investment period, where $\tau_t$ is the tax rate. Agents can sell a portion of their assets to other agents and use the proceeds to purchase two types of financial assets: Shares of other agents and risk-free bonds. Bonds mature at the end of period $t$, with a net supply of zero. Each agent must retain ownership of at least a small portion $\theta_c\in(0,1)$ of their assets due to investment risk considerations. This friction results in market incompleteness. In a real fiscal redistribution system, investment agents are net contributors or net taxpayers, while government employees are net beneficiaries. Government employees include not only government staff but also retirees living on social security, recipients of disability or unemployment benefits, and others.

For different economic conditions and environments, government departments impose different tax rates $\tau_t$ on investor returns, and tax revenues are redistributed to government staff to ensure that the government operates on a balanced budget. Therefore, we assume that: (\textbf{A}) When economic uncertainty is high, the tax rate is high, $\tau_t=\tau^H$, when economic uncertainty is low, the tax rate is low, $\tau_t=\tau^L$, where $\tau^H>\tau^L$. (\textbf{B}) In the macroeconomic issues discussed in this paper, only investment agents (Bayesian agents $C_{Bayes}$ and non-Bayesian agents $C_{non-Bayes}$) and government staff ($C_{Gov}$) engage in consumption. Let agent type $b$ comes from the set $L_t$, then the size of the investment agents is $m_t=\int_{b\in L_t}db$, and the size of government workers is $1-m_t$. Based on this, Proposition 1 is proposed:

\begin{proposition}
\textit{When economic uncertainty is high, the government levies a tax of $\tau^H$, and the consumption level of investment agents decreases. When economic uncertainty is low, the government levies a tax of $\tau^L$, and the consumption level of investment agents increases.}
\end{proposition}

\begin{proof}
Given $L_t$, the expected total output of the economy is fixed. Part of this output will be allocated to government workers, which equals to tax rate $\tau_t$, and another part will be allocated to investors, which equals to $1-\tau_t$. The consumption of government workers comes from total tax revenue, which depends on total output $Y_{t+1}$:
\[Y_{t+1}=\int_{j\in L_t}Y_{j,t+1}dj\]

For a certain tax rate $\tau$, the total tax revenue is $\tau Y_{t+1}$:
\begin{align}
\text{Tax Revenue}=\tau\int_{j\in L_t}Y_{j,t+1}dj=\tau\left(\int_{j\in L_t}e^{\mu_j+\varepsilon_{t+1}+\varepsilon_{j,t+1}}dj\right)G_t\label{13}
\end{align}

Based on the \textit{Law of Large Numbers}:
\[\begin{gathered}\int_{j\in L_t}e^{\mu_j+\varepsilon_{j,t+1}}dj=m_t\mathbb{E}[e^{\mu_j+\varepsilon_{j,t+1}}|j\in L_t]=m_t\mathbb{E}[e^{\mu_j}|j\in L_t]\mathbb{E}_t[e^{\varepsilon_{j,t+1}}|j\in L_t]\\\int_{j\in L_t}e^{\mu_j+\varepsilon_{j,t+1}}dj=m_t\mathbb{E}[e^{\mu_j}|j\in L_t]\end{gathered}\]

Now we set $\mu_b{\sim}\mathcal{N}(\bar{\mu},\sigma_\mu^2),\bar{\mu}=0,m_t^k=1-m_t$, then:
\[\begin{gathered}m_{t}^{k}=\int_{K(k)}^{\infty}\phi\left(\mu_{b};\overline{\mu},\sigma_{\mu}^{2}\right)d\mu_{b}=1-\Phi\left(K(k);\overline{\mu},\sigma_{\mu}^{2}\right)\\\mathbb{E}[e^{\mu_{j}}|j\in L_{t}]=\frac{1}{m_{t}^{k}}\int_{K(k)}^{\infty}e^{\mu_{j}}\phi\left(\mu_{j};\overline{\mu},\sigma_{\mu}^{2}\right)d\mu_{j}=\frac{e^{\overline{\mu}+\frac{1}{2}\sigma_{\mu}^{2}}\left(1-\Phi\left(K(k);\overline{\mu}+\sigma_{\mu}^{2},\sigma_{\mu}^{2}\right)\right)}{m_{t}^{k}}\end{gathered}\]

We get:
\[Y_{t+1}=\int_{j\in L_t}Y_{j,t+1}=G_te^{\varepsilon_{t+1}}m_t\mathbb{E}[e^{\mu_j}|j\in L_t]=G_te^{\varepsilon_{t+1}}m_te^{\overline{\mu}+\frac{1}{2}\sigma_\mu^2}\frac{\left(1-\Phi(K(k);\overline{\mu}+\sigma_\mu^2,\sigma_\mu^2)\right)}{1-\Phi(K(k);\overline{\mu},\sigma_\mu^2)}\]

Let:
\[\mu_k=K(k)-\overline{\mu}\]

Then:
\[Y_{t+1}(\mu_k)=G_te^{\varepsilon_{t+1}}m_te^{\overline{\mu}+\frac{1}{2}\sigma_\mu^2}\frac{\left(1-\Phi(\mu_k;\sigma_\mu^2,\sigma_\mu^2)\right)}{1-\Phi(\mu_k;0,\sigma_\mu^2)}\]

Denote $f(\mu_k)=\frac{\left(1-\Phi(\mu_k;\sigma_\mu^2,\sigma_\mu^2)\right)}{1-\Phi(\mu_k;0,\sigma_\mu^2)}$, we get:
\[\frac{\partial f(\mu_k)}{\partial\mu_k}=\frac{-\phi(\mu_k;\sigma_\mu^2,\sigma_\mu^2)[1-\Phi(\mu_k;0,\sigma_\mu^2)]+[1-\Phi(\mu_k;\sigma_\mu^2,\sigma_\mu^2)]\phi(\mu_k;0,\sigma_\mu^2)}{\left[1-\Phi(\mu_k;0,\sigma_\mu^2)\right]^2}\]

And:
\[\frac{\phi(\mu_{k};0,\sigma_{\mu}^{2})}{1-\Phi(\mu_{k};0,\sigma_{\mu}^{2})}{>}\frac{\phi(\mu_{k}-\sigma_{\mu}^{2};0,\sigma_{\mu}^{2})}{1-\Phi(\mu_{k}-\sigma_{\mu}^{2};0,\sigma_{\mu}^{2})}=\frac{\phi(\mu_{k};\sigma_{\mu}^{2},\sigma_{\mu}^{2})}{1-\Phi(\mu_{k};\sigma_{\mu}^{2},\sigma_{\mu}^{2})}\]

Since the inequality above always exists, for $\mu_H>\mu_k>\mu_L$, the total output $Y^H_{t+1}(\mu_H)$ of entrepreneur with high investment capacity $\mu_H$ will be greater than the total output $Y^L_{t+1}(\mu_L)$ of entrepreneur with low investment capacity $\mu_L$. Hence, investors with high cognitive abilities have higher economic output than investors with low cognitive abilities, $Y^H_{t+1}(\mu_H)>Y^L_{t+1}(\mu_L)$.

Then, we analyze the Tax Revenue:
\[\text{Tax Revenue}=\tau\int_{j\in L_t}Y_{j,t+1}dj=\tau\left(\int_{j\in L_t}e^{\mu_j+\varepsilon_{t+1}+\varepsilon_{j,t+1}}dj\right)G_t=\tau G_te^{\varepsilon_{t+1}}m_t\mathbb{E}[e^{\mu_j}|j\in L_t]\]

Based on balanced budget constraints, total tax revenue will be distributed equally among government workers with a size of $1-m_t$. The consumption of per government worker is:
\[C_{Gov,t}=\frac{\tau G_te^{\varepsilon_{t+1}}m_t\mathbb{E}[e^{\mu_j}|j\in L_t]}{1-m_t}\]

Then, the utility of government workers is:
\[\begin{gathered}\mathbb{E}_t[U(C_{Gov,t+1},\gamma_b\neq1)|\tau]=\frac{\tau^{1-\gamma_b}}{1-\gamma_b}G_t^{1-\gamma_b}\mathbb{E}_t[e^{(1-\gamma_b)\varepsilon_{t+1}}]\mathbb{E}_t[e^{\mu_j}|j\in L_t]^{1-\gamma_b}\left(\frac{m_t}{1-m_t}\right)^{1-\gamma_b}\\\mathbb{E}_t[U(C_{Gov,t+1},\gamma_b=1)|\tau]=log(\tau)+\mathbb{E}_t\left[log[G_t e^{\varepsilon_{t+1}}m_t\mathbb{E}[e^{\mu_j}|j\in L_t]]\right]-log(1-m_t)\end{gathered}\]

Now, assuming that the enterprise investor size is $m_t$, each agent b of investor sells $1-\theta_c$ shares and retains $\theta_c$ shares. Then, the net income is:
\[M_{b,t}=\mathbb{E}_t[\pi_{t,t+1}Y_{b,t+1}(1-\tau_t)]\]

Where $\pi_{t,t+1}$ is the capital depreciation index and $\tau_t$ is the tax rate in stage $t$. Each investor purchases capital shares in enterprises through investment. Let $N_t^{bj}$ represent the percentage of shares in company $j$ purchased by investor $b$ at time $t$, and let $N_{bt}^0$ be the entrepreneur’s (long or short) position in the bond. Then the investment budget constraint is:
\[(1-\theta_c)M_{bt}=\int_{b\neq j}N_t^{bj}M_{jt}dj+N_{bt}^0\]

Standardize the stock price and fix it at 1. Then, for a tax rate $\tau_t=\tau$, the consumption of the investment agent in stage $t$ is:
\[C_{b,t+1}=\theta_cY_{b,t+1}(1-\tau_t)+\int_{j\in L_t}N_t^{bj}Y_{j,t+1}(1-\tau_t)dj+N_{bt}^0\]

Where $b\in\{Bayes,non-Bayes\}$. Under market equilibrium, the premiums on all firms' risk assets are the same, so the optimal investment strategy for entrepreneurs is to allocate amounts according to the market value weights of risk assets. Let the proportion of firms' investment in risk assets $N_t^{bj}$ be $\delta(\gamma)$ and the proportion of investment in risk-free assets $N_{bt}^0$ be $1-\delta(\gamma)$. Then, under equilibrium conditions:
\[\begin{gathered}N_{t}^{bj}M_{jt}=[\delta(\gamma)(1-\theta_{c})M_{bt}]\times\frac{M_{jt}}{M_{P}}=[\delta(\gamma)(1-\theta_{c})M_{bt}]\times\frac{M_{jt}}{\int_{p\in L_{t}}M_{pt}dp}\\N_{t}^{bj}=\frac{\delta(\gamma)(1-\theta_{c})M_{bt}}{\int_{p\in L_{t}}M_{pt}dp}\\N_{bt}^{0}=(1-\theta_{c})M_{bt}-\int_{b\neq j}N_{t}^{bj}M_{jt}dj=(1-\theta_{c})M_{bt}-\frac{[\delta(\gamma)(1-\theta_{c})M_{bt}]}{\int_{p\in L_{t}}M_{pt}dp}\int_{b\neq j}M_{jt}dj\end{gathered}\]

According to the \textit{Continuum Hypothesis}:
\[\int_{b\neq j}M_{jt}dj=M_P-M_{jt}\approx M_P\]

Hence:
\[N_{bt}^0=(1-\theta_c)M_{bt}-\delta(\gamma)(1-\theta_c)M_{bt}=[1-\delta(\gamma)](1-\theta_c)M_{bt}\]

Based on market clearing conditions (demand from b equals supply from j) and the \textit{Continuum Hypothesis}:
\[\begin{gathered}(1-\theta_c)=\int_{b\neq j}N_t^{bj}db=\delta(\gamma)(1-\theta_c)\int_{b\neq j}\frac{M_{bt}}{\int_{p\in L_t}M_{pt}dp}db=\delta(\gamma)(1-\theta_c)\frac{\int_{b\neq j}M_{bt}db}{\int_{p\in L_t}M_{pt}dp}\\\delta(\gamma)=1\end{gathered}\]

Then, we can obtain the equilibrium conditions:
\[\begin{gathered}N_t^{bj}=(1-\theta_c)\frac{e^{\mu_b}}{\int_{p\in L_t}e^{\mu_p}dp}\\N_{bt}^0=0\end{gathered}\]

The consumption of investment agents is:
\[C_{b,t+1}=(1-\tau)G_te^{\mu_b}e^{\varepsilon_{t+1}}[\theta_ce^{\varepsilon_{b,t+1}}+(1-\theta_c)]\]

Therefore, we can obtain the analytical expression for utility:
\[\begin{gathered}\mathbb{E}_t[U(C_{b,t+1},\gamma_b\neq1)|\tau]=\frac{(1-\tau)^{1-\gamma_b}G_t^{1-\gamma_b}e^{(1-\gamma_b)\mu_b}}{1-\gamma_b}\mathbb{E}_t[e^{(1-\gamma_b)(\varepsilon_{t+1})}]\mathbb{E}[[\theta_ce^{\varepsilon_{i,t+1}}+(1-\theta_c)]^{1-\gamma_b}]\\\mathbb{E}_t[U(C_{b,t+1},\gamma_b=1)|\tau]=log(1-\tau)+log[G_te^{\mu_b}]+\mathbb{E}_t\left[log[e^{\varepsilon_{t+1}}(\theta_c e^{\varepsilon_{i,t+1}}+(1-\theta_c))]\right]\end{gathered}\]

Therefore, for investment agents:
\[\mathbb{E}_t[U(C_{b,t+1},\gamma_b\neq1\&\gamma_b=1)|\tau^L]>\mathbb{E}_t[U(C_{b,t+1},\gamma_b\neq1\&\gamma_b=1)|\tau^H]\]

If and only if:
\[\tau^H>\tau^L\]

The next key question to explore is: What economic conditions will cause the tax rates to increase? This is because it involves the direction of consumption adjustments by investment agents in the corresponding economic environment, i.e., whether to increase or decrease consumption levels.

\noindent\textbf{Empirical Evidence and Analysis}: Our theory has shown that an increase in tax rates will increase the taxes levied on entrepreneurs, which will reduce the utility of enterprise investors and thus reduce consumption. Next, we will design a simple econometric experiment to prove that increased uncertainty will increase the taxes paid by enterprises.

Using data from Chinese A-share listed companies from 2000 to 2023 as the analysis sample: We selected the annual taxes and fees paid by listed companies as the dependent variable (\textbf{Taxation/100 million RMB yuan}). The independent variable needs to measure uncertainty factors. Considering that the uncertainty factors affecting the economic conditions of listed companies mainly stem from information asymmetry (information friction) and economic policy uncertainty, we will comprehensively consider these two indicators to construct the independent variable measuring uncertainty.

(\textbf{A}) Regarding information asymmetry, refer to the studies of \citet{bharath2009does}, \citet{pastor2003liquidity} and \citet{amihud2002illiquidity}: First, calculate the first-order indicators of three variables considered as proxies for information asymmetry: Step one, calculate the Illiquidity Ratio (ILL) based on the relationship between trading volume and price changes. Step two, calculate its inverse indicator, the Liquidity Ratio (LR) (\citealp{amihud2002illiquidity}). Step three, estimate the extent of the impact of order flow on yield reversals using a regression model to obtain the Liquidity Indicator (GAM) (\citealp{pastor2003liquidity}). Second, perform principal component analysis (PCA) on the above three primary liquidity indicators and extract the first principal component as the final proxy variable (\textbf{ASY}) for measuring the degree of information asymmetry between capital providers and firms. The higher the value of this indicator, the greater the degree of information asymmetry.

(\textbf{B}) For economic policy uncertainty, the economic policy uncertainty index constructed by \citet{baker2016measuring} is used. This index is based on news reports and is jointly published by Stanford University and the University of Chicago, covering major economies around the world. \citet{baker2016measuring} selected the South China Morning Post as the news reporting retrieval platform and constructed the China Economic Policy Uncertainty Index using text retrieval and filtering methods. This paper calculates the annual economic policy uncertainty index using eight methods and takes the average of the eight results as the economic policy uncertainty variable (\textbf{EUP}) in the econometric analysis\footnote{The eight methods are as follows: (\textbf{A}) The arithmetic average of the monthly economic policy uncertainty index within the year divided by 100. (\textbf{B}) The natural logarithm of the arithmetic average of the monthly economic policy uncertainty index within the year. (\textbf{C}) The economic policy uncertainty index of the last month of each year divided by 100. (\textbf{D}) The natural logarithm of the economic policy uncertainty index for the last month of each year. (\textbf{E}) The weighted average of economic policy uncertainty at the annual level divided by 100. (\textbf{F}) The natural logarithm of the weighted average of economic policy uncertainty at the annual level. (\textbf{G}) The geometric average of economic policy uncertainty at the annual level divided by 100. (\textbf{H}) The natural logarithm of the geometric average of economic policy uncertainty at the annual level.}.

When constructing the independent variable, it is necessary to comprehensively consider information asymmetry ($ASY\in(-\infty,+\infty)$) and economic policy uncertainty($EUP\in(1,+\infty)$). Therefore, we set variable \textbf{Uncertainty}: $Uncertainty=ASY\div EUP$ if $ASY<0$, and $Uncertainty=ASY\times EUP$ If $ASY>0$. This allows us to thoroughly consider the magnitude of overall uncertainty.

We selected nine commonly used economic indicators of listed companies as control variables for empirical analysis: Debt-to-equity ratio (\textbf{DER}). Return on assets (\textbf{ROA})\footnote{It is necessary to use ROA as one of the control variables because we want to avoid the possibility that the increase in taxes is due to factors related to corporate profit growth.}. Accounts receivable ratio (\textbf{REC}). Inventory ratio (\textbf{INV}). Capital intensity (\textbf{CAP}). Book-to-market ratio (\textbf{BM}). Tobin's Q ratio (\textbf{TobinQ}). Major shareholder fund occupation (\textbf{Occupy}). Gross profit margin (\textbf{GrossProfit})\footnote{Since the control variables at the corporate level are relatively fixed, and to focus on presenting the core content and conclusions of the article, detailed explanations and calculation formulas for the control variables will not be included in the main text. If needed, please contact us for further information.}. The results of the econometric analysis of this part are shown in Table 1\footnote{Standard errors in parentheses: $^{*} p<0.05$, $^{**} p<0.01$, $^{***} p<0.001$. All data is sourced from the official websites of listed companies, stock exchange websites, and the CSMAR database. The same below.}:

\begin{table}[ht]
\centering
\caption{Regression Results A}
\label{tab:results}
\small
\setlength{\tabcolsep}{4pt}
\begin{tabular}{lccccccc}
\toprule
 &  & \multicolumn{3}{c}{Taxation} & \multicolumn{3}{c}{Friction} \\
\cmidrule(lr){3-5} \cmidrule(lr){6-8}
 &  & (1) & (2) & (3) & (4) & (5) & (6) \\
\midrule
Uncertainty &  & $0.101^{***}$ & $0.150^{***}$ & $0.112^{***}$ & $-0.006^{***}$ & $-0.004^{***}$ & $-0.003^{**}$ \\
 &  & (0.022) & (0.023) & (0.024) & (0.001) & (0.001) & (0.001) \\
DER &  &  &  & $-0.081^{***}$ &  &  & $0.028^{***}$ \\
 &  &  &  & (0.021) &  &  & (0.001) \\
ROA &  &  &  & $1.810^{***}$ &  &  & $0.291^{***}$ \\
 &  &  &  & (0.359) &  &  & (0.013) \\
REC &  &  &  & $-0.358$ &  &  & $0.058^{***}$ \\
 &  &  &  & (0.294) &  &  & (0.012) \\
INV &  &  &  & $-0.045$ &  &  & $0.093^{***}$ \\
 &  &  &  & (0.217) &  &  & (0.009) \\
CAP &  &  &  & $-0.052^{***}$ &  &  & $0.010^{***}$ \\
 &  &  &  & (0.012) &  &  & (0.000) \\
BM &  &  &  & $0.502^{***}$ &  &  & $-0.007^{***}$ \\
 &  &  &  & (0.027) &  &  & (0.001) \\
TobinQ &  &  &  & $0.058^{**}$ &  &  & $-0.007^{***}$ \\
 &  &  &  & (0.019) &  &  & (0.001) \\
Occupy &  &  &  & $3.343^{***}$ &  &  & $-0.336^{***}$ \\
 &  &  &  & (0.598) &  &  & (0.026) \\
GrossProfit &  &  &  & $0.308$ &  &  & $0.013$ \\
 &  &  &  & (0.201) &  &  & (0.008) \\
Constant &  & $-0.446^{**}$ & $-0.146$ & $-0.378^{*}$ & $0.050^{***}$ & $0.088^{***}$ & $0.026^{***}$ \\
 &  & (0.144) & (0.118) & (0.161) & (0.001) & (0.003) & (0.005) \\
Year &  & $\surd$ & $\surd$ & $\surd$ & $\times$ & $\surd$ & $\surd$ \\
ID &  & $\times$ & $\surd$ & $\surd$ & $\times$ & $\surd$ & $\surd$ \\
\midrule
$R^{2}$ &  & 0.012 & 0.012 & 0.021 & 0.002 & 0.031 & 0.089 \\
Number &  & 54861 & 54861 & 53750 & 44604 & 44604 & 43583 \\
\bottomrule
\end{tabular}

\vspace{5pt}
\raggedright
\end{table}

From columns (1) to (3) of Table 1, we can see that there is a significant positive correlation between taxation and uncertainty. An increase in overall uncertainty can increase the taxes that enterprises need to pay at a $1\%$ significance level\footnote{The dependent variable \textbf{Friction} in columns (4) to (6) of Table 1 is used to measure firms' financial friction as the “Credit Availability" variable. The experimental results in these three columns will be explained and used in the analysis and applications of section 4.}.

Therefore, we have reason to believe that in a highly uncertain economic environment, the government will increase taxes on enterprises. Theoretical analysis shows that when taxes increase, the utility of enterprise investment agents decreases and consumption levels decline and vice versa.
\end{proof}

\subsection{Utility Acquisition According to CAWF}
\textbf{Consumption Adjustment Weight Function}: Again, the core assumption of our article is that: Cognitive resources are one of the agent's resource endowments, which are similar to wealth or capital. Cognitive resources determine the agent's rationality, the more cognitive resources an agent possesses, the more rational their decisions will be. In this context, \textbf{Theorem 1} and \textbf{Theorem 2} reveal that an agent's cognitive resources will be continuously diluted and reduced as the time and scale of big data interaction increasing. This means that the agent's rationality will gradually shift to lower levels as the interaction with big data increases (i.e., transforming from a rational Bayesian agent to non-Bayesian agents with varying degrees of rationality). \textbf{Theorem 3} measures the value of big data, where the data value $D_t\in(0,1)$ linearly represents the uncertainty of macro information. Now, considering the time attribute of data generation (i.e., data is randomly generated in a continuous-time spacetime), we define the generation of data $D_t$ to follow a mean-reverting process with the mean value $\bar{D}=0.5$:
\[dD(t)=\varphi_D(\bar{D}-D(t))dt+\phi_DdZ_t^D\]

Among them, $\varphi_D>0$ is the regression rate, $\bar{D}=0.5$ is the average data value index, $\phi_D>0$ is the volatility, and $Z_t^D$ is the standard Brownian motion.

Based on \textbf{Theorem 4} and \textbf{Proposition 1}, we know the amount and direction of consumption adjustments made by agents in environments with varying degrees of uncertainty. Therefore, we further summarize: Agents will continue to engage in interactive behavior with big data and undergo rational state transitions. When uncertainty factors in the macro environment are high, the data value $D_t$ is low, and agents tend to reduce consumption. Furthermore, as rationality shifts to lower states, non-Bayesian agents often overestimate the risks of uncertain information, leading to greater reductions in consumption compared to Bayesian agents. Conversely, when uncertainty in the macro environment is low, the level of data value $D_t$ is high, and agents tend to increase consumption. However, as rationality shifts to lower states, non-Bayesian agents underestimate the welfare effects of information, resulting in a smaller increase in consumption compared to Bayesian agents.

More specifically, the Consumption Adjustment Weighting Function (CAWF) is as follows:
\[\mathrm{CAWF}\equiv C_\Delta(t,n)=\left(s_\Delta\times e^{(D(t)-\bar{D})}-1\right)\left(\frac{1}{1+\frac{n}{\omega}}\right)+\left(1-s_\Delta\times e^{(\bar{D}-D(t))}\right)\left(1-\frac{1}{1+\frac{n}{\omega}}\right)\]

In CAWF, $s_{\Delta}\geq1$ indicates the agent's sensitivity to data. $n>0$ indicates the scale of data. $\omega>0$ indicates the dilution weight of big data on individual cognitive resources. And $D(t)\in(0,1)$ indicates data value. Setting $s_{\Delta}=1.15$ and $\omega=100$, the agent's consumption adjustment $C_{\Delta}(t,n)$ is shown in Figure 5. Furthermore, substituting $dD(t)$ and using \textit{Monte Carlo} simulation data to form $D(t){\sim}\mathcal{N}(\bar{D},\frac{\phi_D^2}{2\varphi_D})$ with a sample size of 1000, where $\varphi_D=0.1,\phi_D = 0.8$, and $D(t)\in[D^-_{min},D^+_{max}]$, reflected and boundaries if it ever reaches them. Figure 6 presents the dynamic of the agent's consumption adjustment $C_{\Delta}(t,n)$ as the data size $n$ increases.

\begin{figure}[htbp]
\centering
\includegraphics[width=12cm]{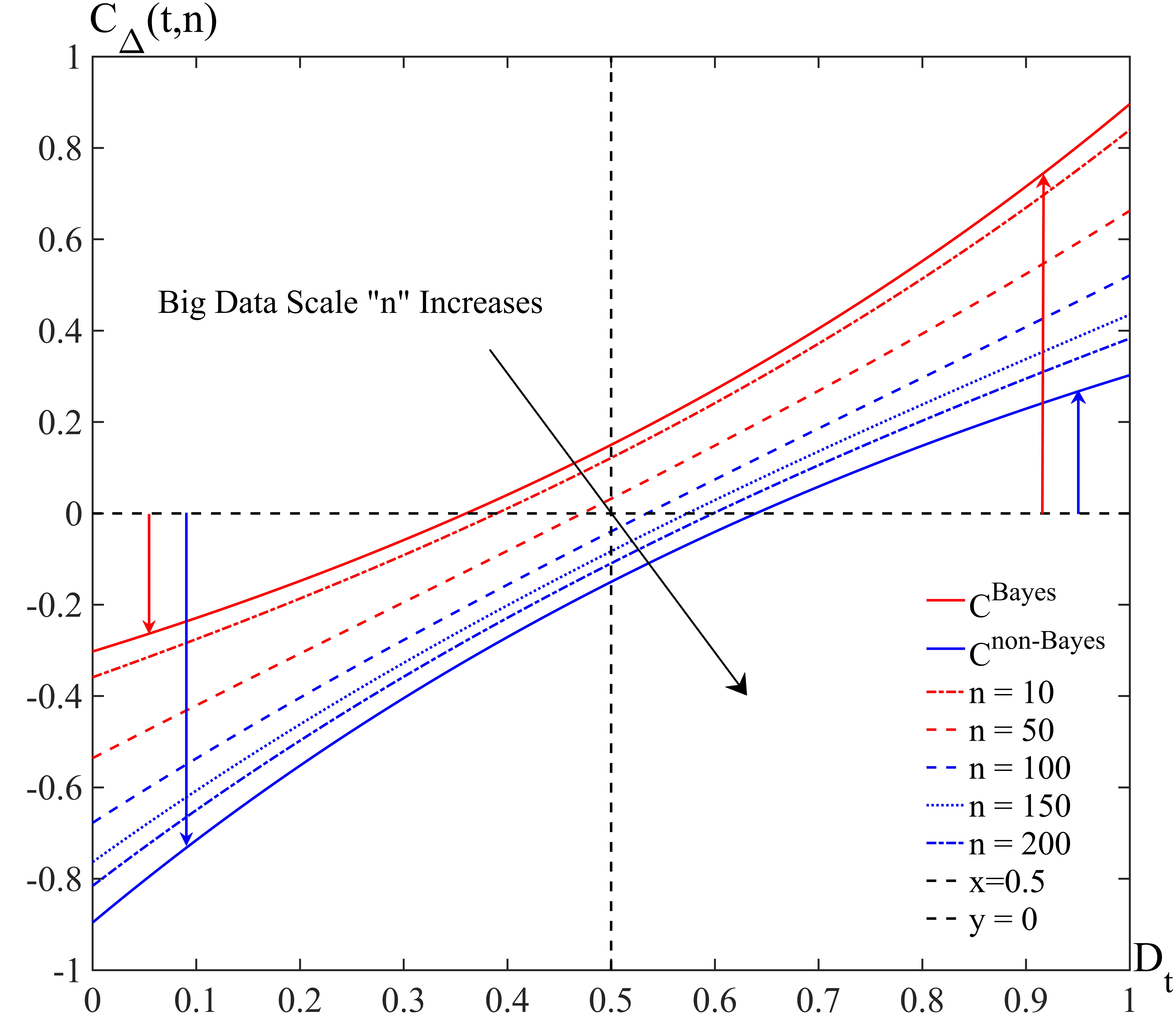}
\caption{\label{fig:F5}The CAWF with Variable $D_t$}
\end{figure}
\begin{figure}[htbp]
\centering
\includegraphics[width=12cm]{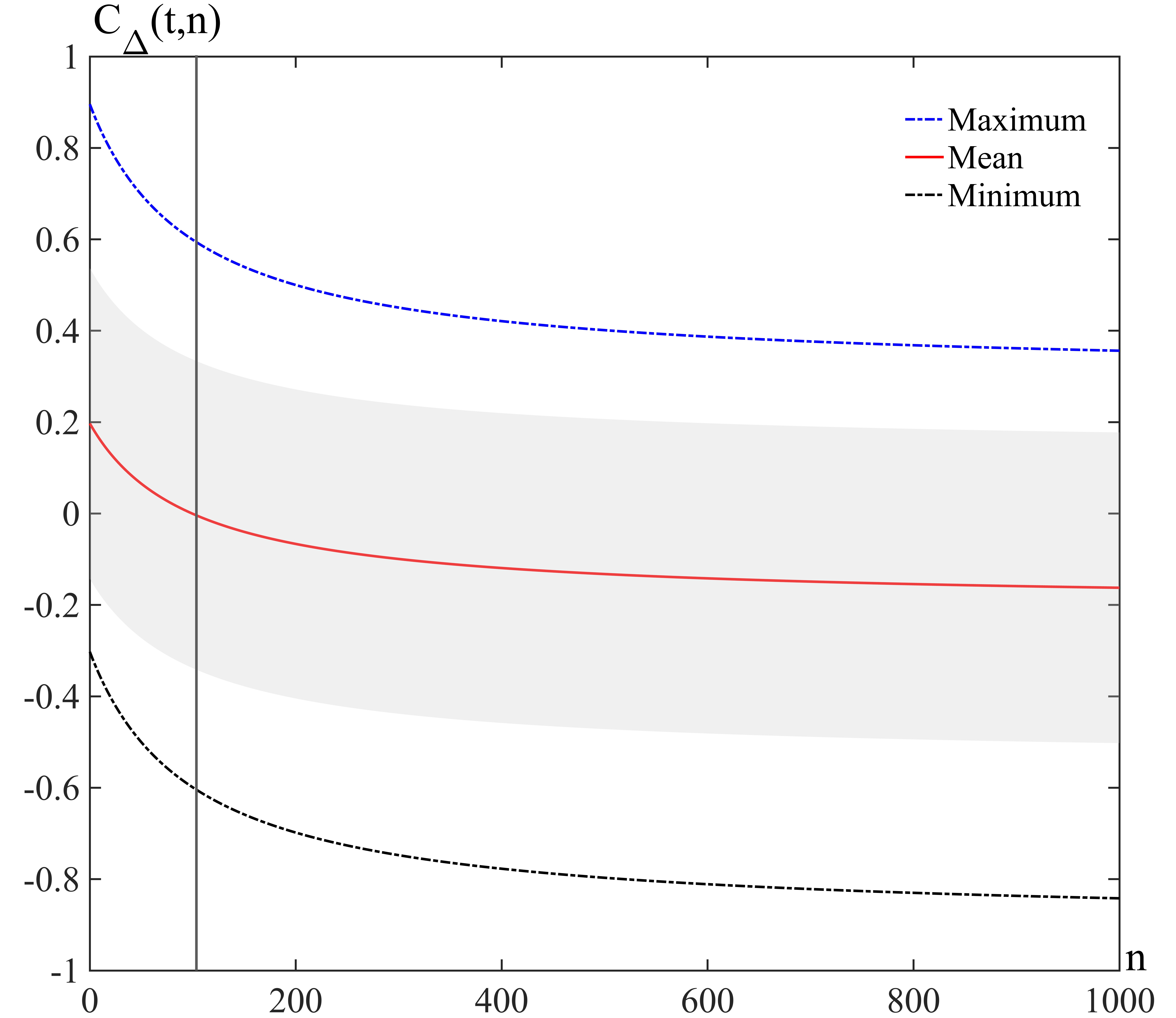}
\caption{\label{fig:F6}The CAWF with Monte Carlo Simulation}
\end{figure}

From Figure 5, we can see that:
\[\begin{gathered}\mathrm{Bayesian}\equiv\lim_{t\to0,n\to0}C_\Delta(t,n)=\left(s_\Delta\times e^{(D(t)-\bar{D})}-1\right)\\\mathrm{non-Bayesian}\equiv\lim_{t\to\infty,n\to\infty}C_\Delta(t,n)=\left(1-s_\Delta\times e^{(\bar{D}-D(t))}\right)\end{gathered}\]

Therefore, according to Figure 5, as the value of data $D_t\in(0,1)$ increases, agents' consumption adjustment strategies gradually shift from reducing consumption to increasing consumption\footnote{Based on \textbf{Proposition 1}, according to \citet{ding2024consumer}, consumer welfare derived from data elements, such as the issuance of digital consumption vouchers or shopping subsidies, has a stimulating effect on consumption growth. However, the substitution effect of data welfare across different consumption categories is limited. Consumers demand consumption welfare, leading to increased consumption in their target consumption categories, rather than displacing consumption expenditures in other categories. Consequently, consumers have no incentive to reduce expenditures in other consumption categories to the same extent. Therefore, when applying \textbf{Proposition 1} of this article, we believe that the impact of data value determined by information entropy on consumer behavior is straightforward: When data value is high, information uncertainty is low, and consumers are clearer about which areas should be further strengthened in consumption and which should be avoided. In such cases, there will always be areas where consumers increase consumption, and total consumption should increase. Conversely, when information uncertainty is high, consumers have a blurred perception of the boundaries between areas where they should increase consumption and those where they should avoid it. Additionally, due to limited funds, they struggle to determine in which areas increasing consumption would effectively enhance utility. As a result, the overall consumption level decreases.}. The red solid line represents the adjustment decisions of Bayes agents in the absence of big data interaction, while the blue solid line represents the adjustment decisions of non-Bayes agents under deep big data interaction. When the data value is low, non-Bayes agents will reduce consumption to a greater extent than Bayes agents. When the value of data is high, the extent to which non-Bayes agents increase consumption will be lower than that of Bayes agents. The red solid line, red dashed line, blue dashed line, blue dotted line, and blue solid line represent the transition of agents' consumption adjustment decisions with the scale of big data interaction $n$ increasing.

We restricted the threshold for data value generation in the \textit{Monte Carlo} simulation to $D(t)\in(0,1)$, getting the results shown in Figure 6, where the blue dotted line shows changes in consumption adjustment decisions in a high-value data ($D(t)\rightarrow1$) environment, and the black dotted line shows changes in consumption adjustment decisions in a low-value data ($D(t)\rightarrow0$) environment. The red solid line represents the average of the previous two situations, which better reflects the randomness of data value generation in reality\footnote{In the absence of any exogenous shocks, such as policy-making which will affect the generation of data value, high-value and low-value data are randomly generated.}. Therefore, we focus on the results presented by the average. As shown in Figure 6, an increase in the scale of big data interaction $n$ causes consumption adjustment decisions to shift toward lower states and eventually converge. At approximately $\frac{1}{10}n$, consumption adjustment decisions are reduced to the initial level of $\frac{1}{2}C_{\Delta}(t_0,n_0)$, indicating that the impact of big data interaction on consumption adjustment decisions is highly evident in the early stages but gradually converges toward an equilibrium value later. We treat the transfer curve of consumption adjustment decisions as the transfer function of the utility that consumption can obtain, representing the proportion of consumption that is effective and can generate utility for agents in big data interactions. As shown by the red curve, during the initial stage of big data interactions $n_0<\frac{1}{10}n$, agents' effective consumption will exceed actual consumption, i.e., $C_{\Delta}(t,n)>0$, indicating that the welfare effect of big data is significant during this period. When $n_0>\frac{1}{10}n$, the effective consumption adjustment function gradually converges to an equilibrium value, i.e., $C_{\Delta}(t^*,n^*)<0$, indicating that during this process, agents' effective consumption will be less than actual consumption, resulting in agents' actual utility being less than the utility corresponding to actual consumption. This is due to the irrational consumption caused by the dilution effect of big data on the agent's cognitive resources, causing the effective consumption $C^{utility}_t$ that can obtain utility $U^{net}_t$, gradually become a weight of actual total consumption $C^{total}_t$, as shown below:
\[\begin{gathered}C^{utility}_t=C^{total}_t\times\left(1+C_{\Delta}(t,n)\right)\\U^{net}_t=\frac{\left(C^{utility}_t\right)^{1-\gamma_t}}{1-\gamma_t}=\frac{\left(C^{total}_t\times(1+C_{\Delta}(t,n))\right)^{1-\gamma_t}}{1-\gamma_t}\end{gathered}\]

In summary, we create a static framework of consumption but with a dynamic transition of cognitive resources, the rationality of consumption decisions is commensurate with the cognitive resources that the agent possesses. Therefore, the agent only consumes once after a period of interacting with big data, and the overarching purpose of our analysis is to find the weight of total consumption which is effective and forms utility after the cognitive resources are diluted by big data at that time. Then, the CAWF measures the weight of effective consumption that can provide utility to total consumption.

\section{Application: Wealth Distribution with Financial Friction}
\textbf{Heterogeneity}: For the fourth part, we apply the CAWF model to the issue of firm wealth distribution with financial frictions. Based on the analysis results in Figure 6, we can identify two key foundational points: (\textbf{A}) When the value of data decreases from high to low, the uncertainty of the economic environment (information entropy) increases, causing the CAWF curve of the agent to shift downward overall. This indicates that the agent's consumption adjustments will decrease, leading to a decline in effective consumption. (\textbf{B}) For each category of economic uncertainty (high, moderate, or low uncertainty), the CAWF curves of agents decrease and eventually converge as the degree of interaction with the scale of big data $n$ increases. This indicates that interaction with big data dilutes agents' cognitive resources, reduces their rationality, and leads to a decline in effective consumption. 

Based on these above, we will define two types of agents in the firm wealth distribution problem with financial frictions according to the CAWF model, whose heterogeneity is from whether they interact with big data and perceive economic uncertainty: (\textbf{A}) The first type of agent does not engage in big data interaction or perceive economic uncertainty, and their total consumption can be fully converted into utility. (\textbf{B}) The second type of agent engages in big data interaction and perceives economic uncertainty, and the portion of their consumption that can be converted into utility is a weight of their total consumption.

A closely related paper to the discussion of this section is given by \citet{achdou2022income}, which is the first paper in the macroeconomic literature that studies how to better build income and wealth distribution models both theoretically and numerically. For more related literature on the topic of wealth distribution, we can see: \citet{fernandez2023financial}, \citet{bilal2023solving}, \citet{ahn2018inequality}, \citet{gabaix2016dynamics}, \citet{brunnermeier2014macroeconomic}, \citet{adrian2010liquidity}, \citet{krusell1998income}, \citet{aiyagari1994uninsured}, \citet{huggett1993risk}.

\subsection{Wealth Distribution with the First Type of Agent}
\textbf{The Prerequisite Model}: We first set up the basic framework of the wealth distribution in an entrepreneur-worker model where financial frictions are to be introduced. In this model, we assume that productivity $z$ is constant. Each entrepreneur owns a private firm which uses $k$ units of capital and $l$ units of labor to produce $y=(zk)^{\alpha}l^{1-\alpha}$ units of output, where $\alpha\in(0,1)$. Capital depreciates at the rate $\delta$. Define the entrepreneur’s profit function as:
\begin{align}
\pi(a)=\max_{k,l}(zk)^\alpha l^{1-\alpha}-wl-(r+\delta)k\label{14}
\end{align}

Function (14) is subject to collateral constraints: $k\leq\lambda a$, where $\lambda\geq1$. Here $\lambda$ reflects financial frictions. Note that as \citet{moll2014productivity} discusses in his paper that this formulation of capital market imperfections is analytically convenient: For $\lambda\geq1$, as $\lambda\rightarrow+\infty$, this indicates that the financial market is approaching a perfect state, and entrepreneurs will face minimal borrowing resistance. When $\lambda=1$, the financial market will be completely closed, and the funds required for enterprise production will be provided entirely by the entrepreneurs themselves. When $\lambda\in(0,+\infty)$, there will be an upper limit on entrepreneurs' borrowing capacity, namely $\lambda$ times their personal net assets. Therefore, it can be concluded that financial friction in imperfect markets decreases as $\lambda$ increases. Moreover, by placing a restriction on an entrepreneur’s leverage ratio $k/a$, it captures the common intuition that the amount of capital available to an entrepreneur is limited by his personal assets. Different underlying frictions can give rise to such collateral constraints.

Unlike the setup in \citet{moll2014productivity}, this section of our paper further assumes entrepreneurs now have access to a risky asset $\kappa_t$ in addition to the riskless bond denoted by $b_t$ (\citealp{achdou2022income}). Therefore, the entrepreneur’s budget constraint becomes:
\[da_t=d\kappa_t+db_t=(\pi(a_t)+R_t\kappa_t+rb_t-c_t)dt\]

Where $R_t$ is the return on the risky asset, $c_t$ is the consumption of the entrepreneur, and $a_t$ denotes the total wealth of the entrepreneur. The return of the risky asset is stochastic and given by:
\[R_tdt=\theta dt+\sigma dW_t\]

Where $\theta$ denotes the average return of the risky asset, $\sigma$ is the diffusion (volatility) coefficient of the risky asset return process, and $W_t$ are the standard Brownian motions.

By plugging in the stochastic process of the risky asset return to the budget constraint, the entrepreneurs’ budget constraint can be rewritten as:
\begin{align}
da_t=(\pi(a_t)+ra_t+(\theta-r)\kappa_t-c_t)dt+\sigma\kappa_tdW_t\label{15}
\end{align}

Assume productivity $z$ is greater than or equal to the productivity cutoff so that all the entrepreneurs are active, then:
\begin{align}
z\geq\frac{r+\delta}{\alpha\left(\frac{1-\alpha}{w}\right)^{\frac{1-\alpha}{\alpha}}}\label{16}
\end{align}

Since this problem is linear, it follows immediately that $k$ is either zero or $\lambda a$, and therefore, at optimum the collateral constraint will be binding, i.e., $k=\lambda a$, then according to \citet{moll2014productivity}, the entrepreneur’s profit maximization problem above can be solved as:
\[\begin{gathered}k(a_t)=\lambda a_t\\l(a_t)=\left(\frac{1-\alpha}{w}\right)^{\frac{1}{\alpha}}z\lambda a_t\\\pi(a_t)=\left(\alpha z{\left(\frac{1-\alpha}{w}\right)}^{\frac{1-\alpha}{\alpha}}-r-\delta\right)\lambda a_t\\y(a_t)=(z\lambda a_t)^\alpha{\left(\left(\frac{1-\alpha}{w}\right)^{\frac{1}{\alpha}}z\lambda a_t\right)^{1-\alpha}}\end{gathered}\]
\textbf{Entrepreneur’s Problem with HJB Equation}: Next, we are ready to solve for this model with two assets. Given the utility function for the entrepreneur takes the form of CRRA: $u(c)=\frac{c^{1-\gamma}}{1-\gamma}$, the discount factor is denoted by $\rho=\rho_0+\beta$, and $r=r_0+\beta$, where $\beta$ denotes rate of wealth dissipation shock (\citealp{moll2022uneven}), $\rho_0$ is the pure time-preference discount factor and $r_0$ is the interest rate, the HJB equation for the entrepreneur’s problem can be written as:
\begin{align}
\rho v(a)=\max_{c,\kappa}\frac{c^{1-\gamma}}{1-\gamma}+(\pi(a)+ra+(\theta-r)\kappa-c)v^{\prime}(a)+\frac{1}{2}\sigma^2\kappa^2v^{\prime\prime}(a)\label{17}
\end{align}

Plugging in the expression for $\pi(a)$ from the solution of the entrepreneur’s profit maximization problem, then the HJB equation can be rewritten as:
\[\rho v(a)=\max_{c,\kappa}u(c)+\left(\left(\alpha z\left(\frac{1-\alpha}{w}\right)^{\frac{1-\alpha}{\alpha}}-r-\delta\right)\lambda a+ra+(\theta-r)\kappa-c\right)v^{\prime}(a)+\frac{1}{2}\sigma^2\kappa^2v^{\prime\prime}(a)\]

Conjecture value function takes the form of $v(a)=Ba^{1-\gamma}$, then the policy functions can be solved as:
\[\begin{gathered}\kappa(a)=\frac{\theta-r}{\gamma\sigma^2}a\\c(a)=\frac{1}{\gamma}\left(\rho-(1-\gamma){\left({\left(\alpha z{\left(\frac{1-\alpha}{w}\right)^{\frac{1-\alpha}{\alpha}}-r-\delta}\right)\lambda+r}\right)-\frac{1}{2}(1-\gamma)\frac{(\theta-r)^2}{\gamma\sigma^2}}\right)a\end{gathered}\]

We now can find the stochastic process of the entrepreneur’s wealth $a$ by plugging the solutions of the policy functions $\pi(a),\kappa(a),c(a)$ into the entrepreneur’s budget constraint which is given by:
\begin{align}
da=(\pi(a)+ra+(\theta-r)\kappa(a)-c(a))dt+\sigma\kappa(a)dW_t\label{18}
\end{align}

Let $\Pi=\left(\alpha z{\left(\frac{1-\alpha}{w}\right)}^{\frac{1-\alpha}{\alpha}}-r-\delta\right)\lambda$, i.e., $\pi(a)=\Pi a$, then the entrepreneur’s budget constraint can be rewritten as: 
\[da=\left(\Pi+r+\frac{\left(\theta-r\right)^{2}}{\gamma\sigma^{2}}-\frac{1}{\gamma}{\left(\rho-(1-\gamma)(\Pi+r)-\frac{1}{2}(1-\gamma)\frac{\left(\theta-r\right)^{2}}{\gamma\sigma^{2}}\right)}\right)adt+\frac{\theta-r}{\gamma\sigma}adW_{t}\]

Note that above stochastic process suggests that the entrepreneur’s wealth a follows \textit{Geometric Brownian Motion} (GBM), and therefore, if we let $x=log(a)$, i.e., $x$ denotes logarithmic wealth, and $\Sigma=\frac{\theta-r}{\gamma\sigma}$, then the corresponding stochastic process of $x$ can be written as:
\[dx=\left(\Pi+r+\frac{(\theta-r)^2}{\gamma\sigma^2}-\frac{1}{\gamma}{\left(\rho-(1-\gamma)(\Pi+r)-\frac{1}{2}(1-\gamma)\frac{(\theta-r)^2}{\gamma\sigma^2}\right)-\frac{1}{2}\frac{(\theta-r)^2}{\gamma^2\sigma^2}}\right)dt+\Sigma dW_t\]

Let:
\[\mu=\left(\Pi+r+\frac{(\theta-r)^2}{\gamma\sigma^2}-\frac{1}{\gamma}{\left(\rho-(1-\gamma)(\Pi+r)-\frac{1}{2}(1-\gamma)\frac{(\theta-r)^2}{\gamma\sigma^2}\right)-\frac{1}{2}\frac{(\theta-r)^2}{\gamma^2\sigma^2}}\right)\]

Hence, we get the stochastic process of logarithmic wealth in real time:
\begin{align}
dx=\mu dt+\Sigma dW_t\label{19}
\end{align}
\textbf{Wealth Distribution}: The steady-state system of equations according to \textit{Mean Field Game} (MFG) in the entrepreneur-worker model is given by:
\[\begin{gathered}\rho v(a)=\max_{c,\kappa}\frac{c^{1-\gamma}}{1-\gamma}+(\pi(a)+ra+(\theta-r)\kappa-c)v^{\prime}(a)+\frac{1}{2}\sigma^2\kappa^2v^{\prime\prime}(a)\\0=-\frac{d}{da}[(\pi(a)+ra+(\theta-r)\kappa-c)\bar{p}(a)]+\frac{1}{2}\frac{d}{da^2}\left[\sigma^2\kappa^2\bar{p}(a)\right]-\beta\bar{p}(a)+\beta\delta(a-1)\\\int_0^\infty\lambda a\bar{p}(a)da=\int_0^\infty{\left(a-\left(\frac{\theta-r}{\gamma\sigma^2}\right)a\right)}\bar{p}(a)da\\\int_0^\infty\left(\frac{1-\alpha}{w}\right)^{\frac{1}{\alpha}}z\lambda a\bar{p}(a)da=1\end{gathered}\]

Where the \textbf{first} equation is HJB equation. The \textbf{second} equation is \textit{Kolmogorov Forward Equation} (KFE) in steady state. The \textbf{third} equation is capital market clearing condition. The \textbf{fourth} equation is the labor market clearing condition given total number of labor supply is assumed to be 1. And $\bar{p}(a)$ is the invariant density function of wealth $a$. To present and simplify the core analysis of wealth distribution in this section, we use the \textbf{first} and the \textbf{second} equations to calculate the analytical solutions of the invariant density of wealth $\bar{p}(a)$\footnote{MFG in the steady state of the entrepreneur-worker model can calculate equilibrium prices of the model, i.e., equilibrium wage rate and equilibrium interest rate. Therefore, the MFG determines the equilibrium interest rate $r^*$ and the equilibrium wage rate $w^*$. We have obtained the analytical forms of $r^*$ and $w^*$, then, we attempted to substitute them into the final result of $\bar{p}(a)$ to obtain $\bar{p}^*(a)$, which is the invariant density function of wealth in equilibrium state. However, we found that the result was too complex and did not produce clear and accurate numerical simulation results. Furthermore, considering that $\bar{p}(a)$ in the equilibrium and non-equilibrium states does not have a fundamental influence on the topic we are discussing, this section only uses the first two equations of MFG to obtain $\bar{p}(a)$ in non-equilibrium situation.}. Hence, when $x=log(a)$, $\bar{p}(x)$ satisfies the following \textit{Kolmogorov Forward Equation} (KFE): 
\begin{align}
0=-\mu\bar{p}^{\prime}(x)+\frac{1}{2}\Sigma^2\bar{p}^{\prime\prime}(x)-\beta\bar{p}(x)+\beta\delta(x)\label{20}
\end{align}

Where the term $\delta\bar{p}(x)$ representing the proportion of the entrepreneurs that loses all of their wealth due to the wealth dissipation shock at the rate of $\beta$ is subtracted from the distribution and hence exit the distribution, and $\beta\delta(x)$, i.e., $\delta(x)$ is the \textit{Dirac Delta Function} centered at zero, represents those entrepreneurs who just suffered an exit now reenter the distribution with zero wealth at the rate of $\beta$. The KFE can be easily solved as follows:

\[\begin{gathered}\bar{p}(x)=K\exp\left(\frac{\mu+\sqrt{\mu^2+2\beta\Sigma^2}}{\Sigma^2}x\right),x<0\\\bar{p}(x)=K\exp\left(\frac{\mu-\sqrt{\mu^2+2\beta\Sigma^2}}{\Sigma^2}x\right),x>0\end{gathered}\]

By normalization of $\bar{p}(x)$, we get:
\[K{\left[\int_{-\infty}^0\exp{\left(\frac{\mu+\sqrt{\mu^2+2\beta\Sigma^2}}{\Sigma^2}x\right)}dx+\int_0^\infty\exp{\left(\frac{\mu-\sqrt{\mu^2+2\beta\Sigma^2}}{\Sigma^2}x\right)}dx\right]}=1\]

We can get $K$ as:
\[K=\frac{\beta}{\sqrt{\mu^2+2\beta\Sigma^2}}\]

Hence, the invariant distribution of the first type of agent is given by:
\begin{align}
\bar{p}(x)=\frac{\beta}{\sqrt{\mu^2+2\beta\Sigma^2}}\exp{\left(\frac{\mu+\sqrt{\mu^2+2\beta\Sigma^2}}{\Sigma^2}x\right)},x<0\label{21}\\\bar{p}(x)=\frac{\beta}{\sqrt{\mu^2+2\beta\Sigma^2}}\exp{\left(\frac{\mu-\sqrt{\mu^2+2\beta\Sigma^2}}{\Sigma^2}x\right)},x>0\label{22}
\end{align}

\subsection{Wealth Distribution with the Second Type of Agent}
For the second type of agent, the only difference is that in the HJB equation of function (17), the consumption variable $c^{utility}$ in the CRRA utility function of the second type of agent is a weight of total consumption $c$. Since we set that the second type of agent engages in big data interactions, cognitive resources are diluted, and the effective consumption which can be converted into utility is reduced are exogenous assumptions in this scenario. Furthermore, CAWF$=C_{\Delta}(t,n)$ is essentially a function of information entropy $\sigma_t$. Therefore, we define the consumption adjustment weight of the second type of agent as a decreasing function of information entropy $f(\sigma_t)$, i.e., the higher the information entropy $\sigma_t$, the higher the uncertainty, the lower the consumption adjustment weight $f(\sigma_t)$, and the less effective consumption could be converted into utility. We denote it as $f_{\sigma}\in(0,1)$:
\[c^{utility}=c\times(1+C_{\Delta}(t,n))=c\times f(\sigma_t)=cf_{\sigma}\]

Then, for the second type of agent, the HJB equation could be  written as follows:
\begin{align}
\rho v(a)=\max_{c,\kappa}\frac{(cf_{\sigma})^{1-\gamma}}{1-\gamma}+(\pi(a)+ra+(\theta-r)\kappa-cf_{\sigma})v^{\prime}(a)+\frac{1}{2}\sigma^2\kappa^2v^{\prime\prime}(a)\label{23}
\end{align}

Hence, according to the solution process in 4.1, we solve function (23):
\[\begin{gathered}c(a)=\frac{1}{\gamma f_{\sigma}}\left(\rho-(1-\gamma){\left({\left(\alpha z{\left(\frac{1-\alpha}{w}\right)^{\frac{1-\alpha}{\alpha}}-r-\delta}\right)\lambda+r}\right)-\frac{1}{2}(1-\gamma)\frac{(\theta-r)^2}{\gamma\sigma^2}}\right)a\\da=\left(\Pi+r+\frac{\left(\theta-r\right)^{2}}{\gamma\sigma^{2}}-\frac{1}{\gamma f_{\sigma}}{\left(\rho-(1-\gamma)(\Pi+r)-\frac{1}{2}(1-\gamma)\frac{\left(\theta-r\right)^{2}}{\gamma\sigma^{2}}\right)}\right)adt+\frac{\theta-r}{\gamma\sigma}adW_{t}\end{gathered}\]

Then $\Pi=\left(\alpha z{\left(\frac{1-\alpha}{w}\right)}^{\frac{1-\alpha}{\alpha}}-r-\delta\right)\lambda$ and $\Sigma=\frac{\theta-r}{\gamma\sigma}$, we get: 
\[dx=\left(\Pi+r+\frac{(\theta-r)^2}{\gamma\sigma^2}-\frac{1}{\gamma f_{\sigma}}{\left(\rho-(1-\gamma)(\Pi+r)-\frac{1}{2}(1-\gamma)\frac{(\theta-r)^2}{\gamma\sigma^2}\right)-\frac{1}{2}\Sigma^2}\right)dt+\Sigma dW_t\]

Therefore, we set:
\[\mu^{\dagger}=\left(\Pi+r+\frac{(\theta-r)^2}{\gamma\sigma^2}-\frac{1}{\gamma f_{\sigma}}{\left(\rho-(1-\gamma)(\Pi+r)-\frac{1}{2}(1-\gamma)\frac{(\theta-r)^2}{\gamma\sigma^2}\right)-\frac{1}{2}\frac{(\theta-r)^2}{\gamma^2\sigma^2}}\right)\]

Then, for the second type agent, we get:
\begin{align}
dx=\mu^{\dagger} dt+\Sigma dW_t\label{24}
\end{align}

As we can see from the equation (24), the drift term of $dx$ is changed from the first type's $\mu$ to the second type's $\mu^{\dagger}$ after introducing the CAWF in our analysis. Now, we solve invariant density function of logarithmic wealth $\bar{p}^{\dagger}(x)$ for the second type of agent according to \textit{Kolmogorov Forward Equation} (KFE):
\[\begin{gathered}0=-\mu^{\dagger}\bar{p}^{{\dagger}\prime}(x)+\frac{1}{2}\Sigma^2\bar{p}^{{\dagger}\prime\prime}(x)-\beta\bar{p}^{\dagger}(x)+\beta\delta(x)\\\bar{p}^{\dagger}(x)=K^{\dagger}\exp\left(\frac{\mu^{\dagger}+\sqrt{{\mu^{\dagger}}^2+2\beta\Sigma^2}}{\Sigma^2}x\right),x<0\\\bar{p}^{\dagger}(x)=K^{\dagger}\exp\left(\frac{\mu^{\dagger}-\sqrt{{\mu^{\dagger}}^2+2\beta\Sigma^2}}{\Sigma^2}x\right),x>0\\K^{\dagger}{\left[\int_{-\infty}^0\exp{\left(\frac{\mu^{\dagger}+\sqrt{{\mu^{\dagger}}^2+2\beta\Sigma^2}}{\Sigma^2}x\right)}dx+\int_0^\infty\exp{\left(\frac{\mu^{\dagger}-\sqrt{{\mu^{\dagger}}^2+2\beta\Sigma^2}}{\Sigma^2}x\right)}dx\right]}=1\\K^{\dagger}=\frac{\beta}{\sqrt{{\mu^{\dagger}}^2+2\beta\Sigma^2}}\end{gathered}\]

Then, we get the invariant distribution of the second type of agent: 
\begin{align}
\bar{p}^{\dagger}(x)=\frac{\beta}{\sqrt{{\mu^{\dagger}}^2+2\beta\Sigma^2}}\exp{\left(\frac{\mu^{\dagger}+\sqrt{{\mu^{\dagger}}^2+2\beta\Sigma^2}}{\Sigma^2}x\right)},x<0\label{25}\\\bar{p}^{\dagger}(x)=\frac{\beta}{\sqrt{{\mu^{\dagger}}^2+2\beta\Sigma^2}}\exp{\left(\frac{\mu^{\dagger}-\sqrt{{\mu^\dagger{}}^2+2\beta\Sigma^2}}{\Sigma^2}x\right)},x>0\label{26}
\end{align}

\subsection{The Numerical Simulation of Wealth Distribution}

Based on equations (21),(22),(25) and (26), the parameter values for the numerical simulations must satisfy two constraints: (\textbf{A}) According to equation (16), the firm's productivity $z$ must be greater than or equal to the productivity threshold $z_{min}$, ensuring that the firm does not go bankrupt and remains active in production activities. (\textbf{B}) Based on columns (4) to (6) of Table 1, we use the credit availability indicator Friction of Chinese listed companies as a proxy variable for financial friction. It is calculated as “new corporate debt (difference between total debt at the beginning and end of the year)/total corporate assets.” The magnitude of the Friction variable measures the level of annual available capital for firms, i.e., the larger Friction is, the lower the financial friction, and the larger $\lambda$ is. As shown in the results of columns (4) to (6) in Table 1, uncertainty significantly reduces corporate credit availability. That means firms face greater financial friction in environments with higher uncertainty. Therefore, based on the empirical results, we establish a corresponding relationship between the consumption adjustment weight $f_{\sigma}$ and financial friction $\lambda$ for the second type of agent, where a higher $f_{\sigma}$ corresponds to a higher $\lambda$, indicating that financial friction is lower in environments with lower uncertainty and vice versa. The parameters and their values required for the simulation are shown in Table 2:

\begin{table}[ht]
\centering
\caption{Parameter Settings}
\label{tab:parameters}
\small
\renewcommand{\arraystretch}{1.2}
\begin{tabular}{lccc}
\toprule
Variable & Parameter & The First Type & The Second Type \\
\midrule
Discount Factor & $\rho$ & 0.05 & 0.05 \\ \addlinespace[0.5em]
Elasticity of Utility & $\gamma$ & 2 & 2 \\ \addlinespace[0.5em]
Capital Share & $\alpha$ & 0.3 & 0.3 \\ \addlinespace[0.5em]
Capital Depreciation & $\delta$ & 0.6 & 0.6 \\ \addlinespace[0.5em]
Wealth Dissipation Rate & $\beta$ & 0.3 & 0.3 \\ \addlinespace[0.5em]
Wage Rate & $w$ & 1 & 1 \\ \addlinespace[0.5em]
Interest Rate & $r$ & 0.01 & 0.01 \\ \addlinespace[0.5em]
\multirow{2}{*}{Average Return Rate} & $\theta^{-}$ & 0.05 & 0.05 \\
 & $\theta^{+}$ & 0.5 & 0.5 \\ \addlinespace[0.5em]
\multirow{2}{*}{Volatility Rate} & $\sigma^{-}$ & 0.05 & 0.05 \\
 & $\sigma^{+}$ & 0.5 & 0.5 \\ \addlinespace[0.5em]
\multirow{3}{*}{Financial Friction} & $\lambda^{L}$ & 5 & 5 \\
 & $\lambda^{M}$ & 25 & 25 \\
 & $\lambda^{H}$ & 50 & 50 \\ \addlinespace[0.5em]
\multirow{3}{*}{Utility Weight} & $f_{\sigma}^{L}$ & 1 & 0.2 \\
 & $f_{\sigma}^{M}$ & 1 & 0.5 \\
 & $f_{\sigma}^{H}$ & 1 & 0.8 \\ \addlinespace[0.5em]
\multirow{2}{*}{Productivity} & $z_{\min}$ & 4.674 & 4.674 \\
 & $z$ & 5 & 5 \\
\bottomrule
\end{tabular}
\end{table}

According to the parameters settings above, Figure 7 shows the logarithmic wealth distribution density of the first type $\bar{p}(x)$ with $(\theta^-,\sigma^-)$, and Figure 8 shows the logarithmic wealth distribution density of the first type $\bar{p}(x)$ with $(\theta^+,\sigma^+)$. Figure 9 shows the logarithmic wealth distribution density of the second type $\bar{p}^{\dagger}(x)$ with $(\theta^-,\sigma^-)$, and Figure 10 shows the logarithmic wealth distribution density of the second type $\bar{p}^{\dagger}(x)$ with $(\theta^+,\sigma^+)$. In each figure, the black, blue, and red solid lines represent different sizes of financial friction $\lambda^i,i\in\{L,M,H\}$ or utility conversion weights $f_{\sigma}^i,i\in\{L,M,H\}$.

\begin{figure}[htbp]
\centering
\includegraphics[width=11cm]{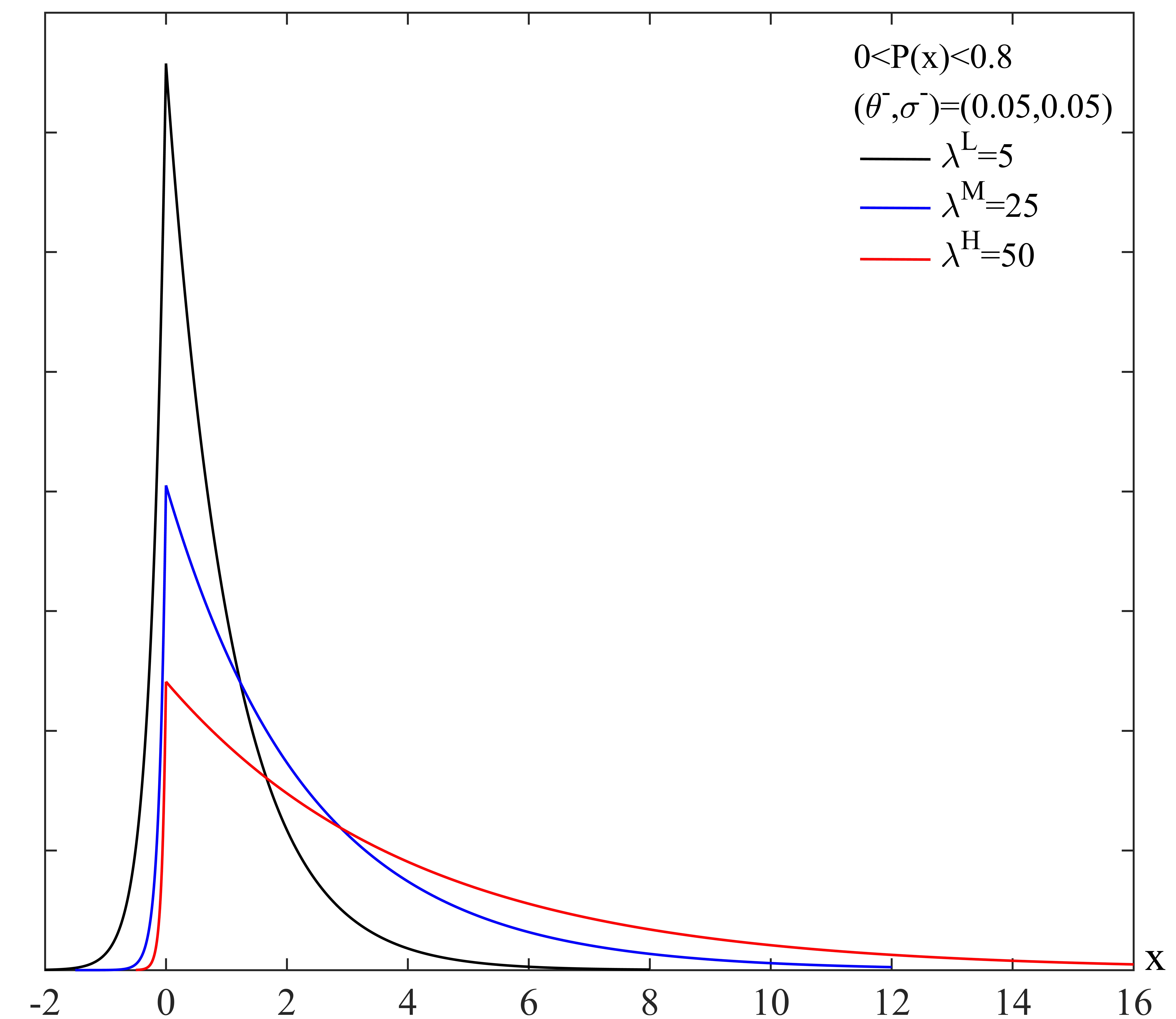}
\caption{\label{fig:F7}Logarithmic Wealth Distribution $\bar{p}(x)$ with $(\theta^-,\sigma^-)$}
\end{figure}
\begin{figure}[htbp]
\centering
\includegraphics[width=11cm]{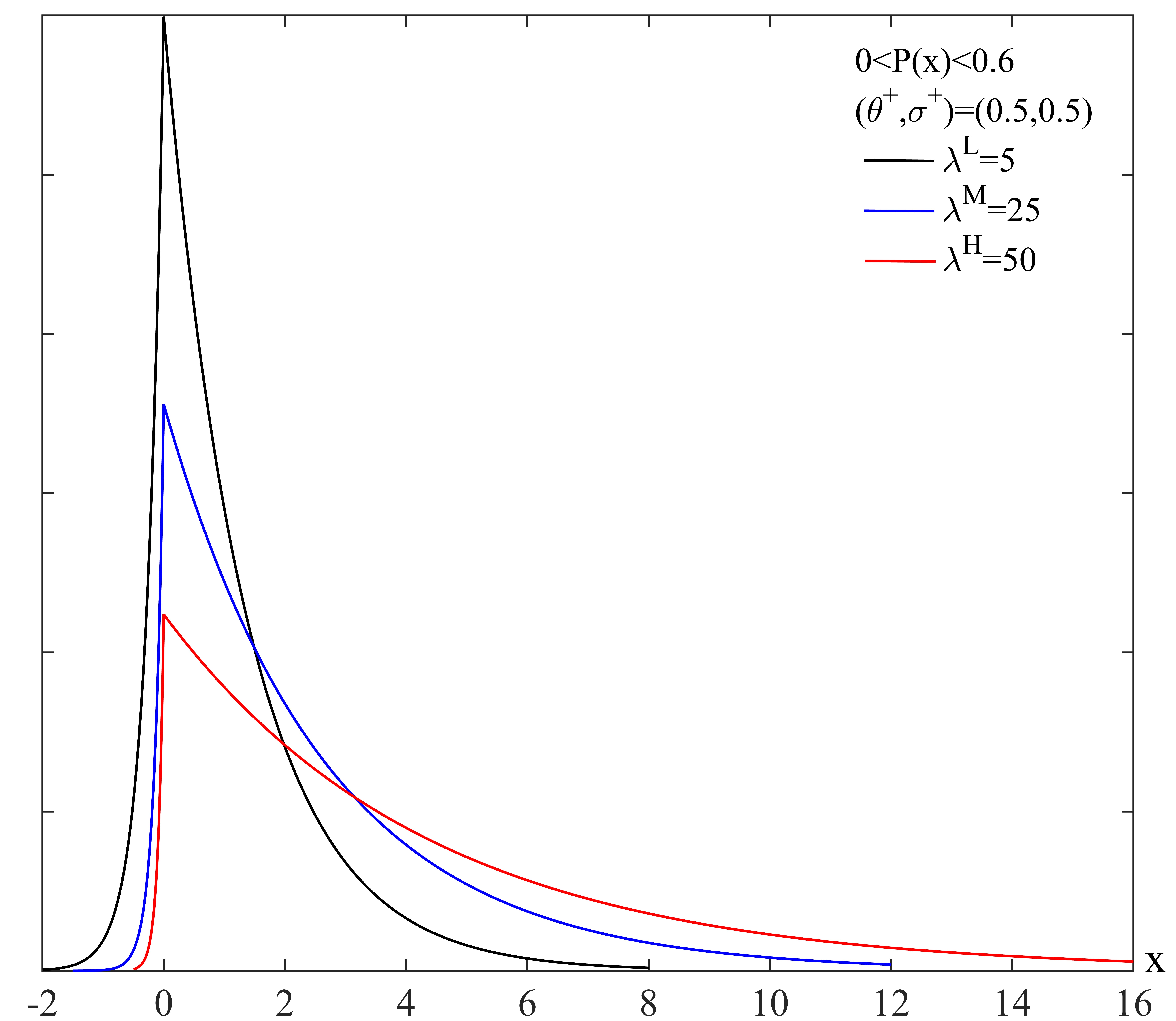}
\caption{\label{fig:F8}Logarithmic Wealth Distribution $\bar{p}(x)$ with $(\theta^+,\sigma^+)$}
\end{figure}
\begin{figure}[htbp]
\centering
\includegraphics[width=11cm]{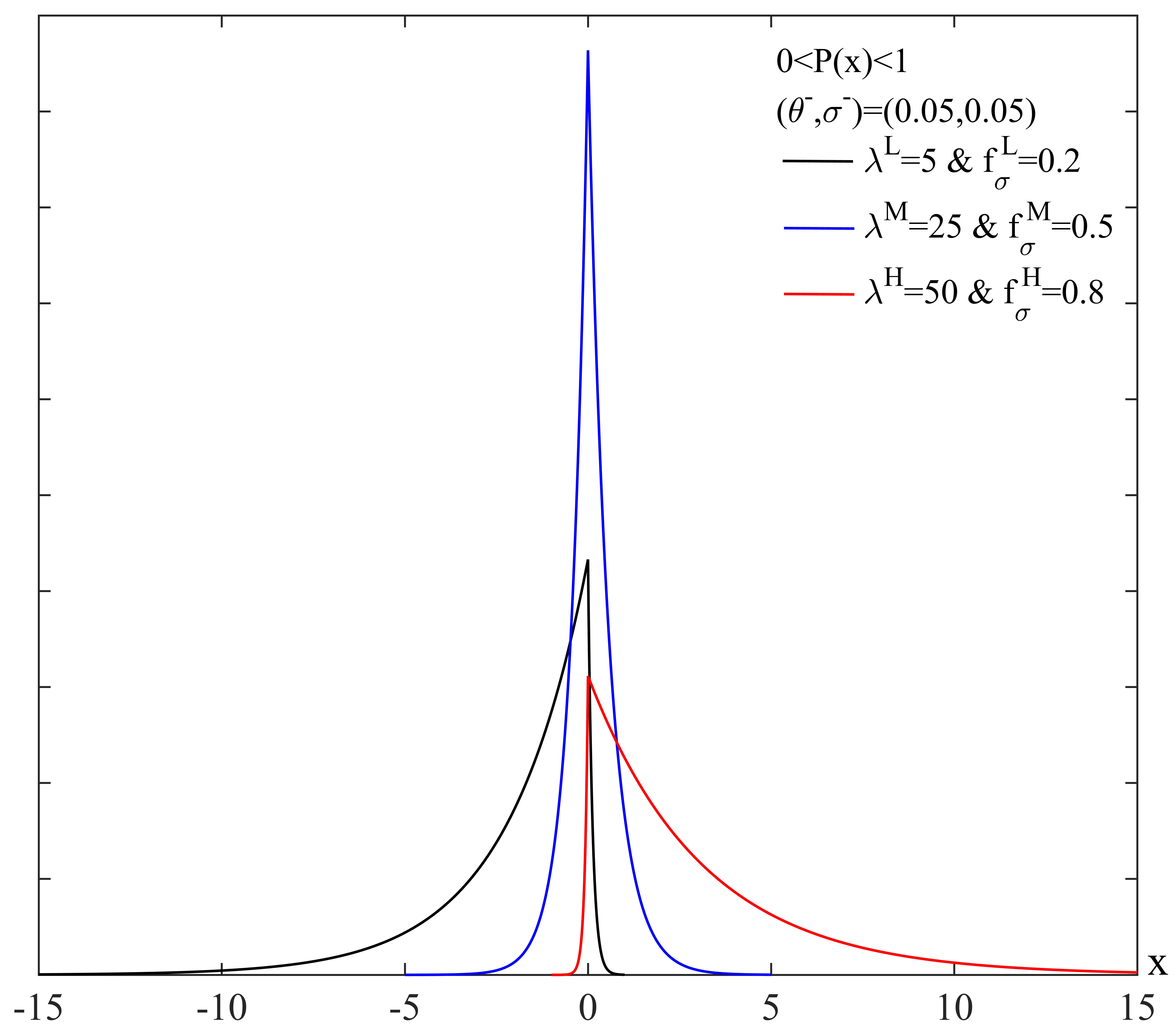}
\caption{\label{fig:F9}Logarithmic Wealth Distribution $\bar{p}^{\dagger}(x)$ with $(\theta^-,\sigma^-)$}
\end{figure}
\begin{figure}[htbp]
\centering
\includegraphics[width=11cm]{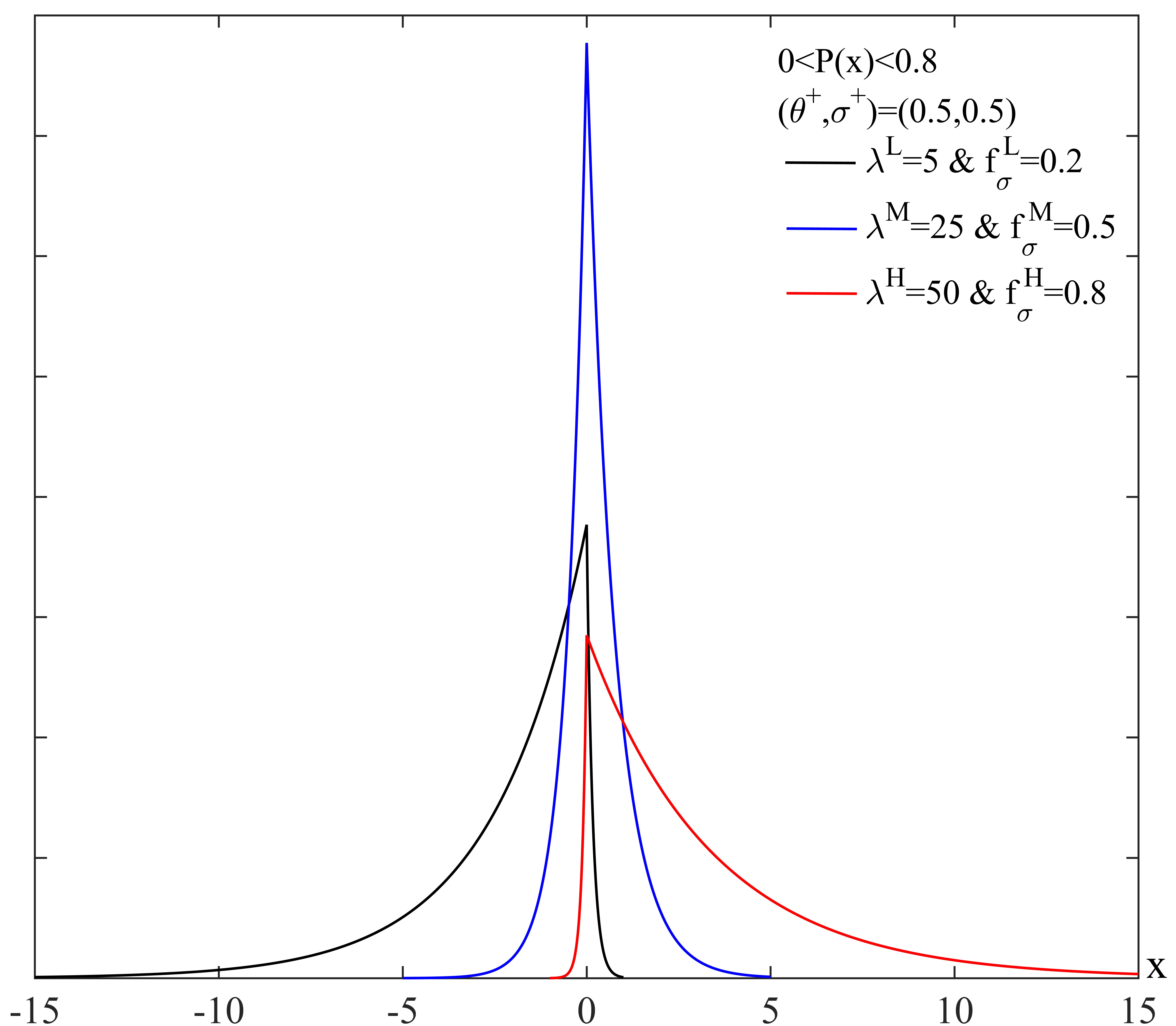}
\caption{\label{fig:F10}Logarithmic Wealth Distribution $\bar{p}^{\dagger}(x)$ with $(\theta^+,\sigma^+)$}
\end{figure}

For all Figures from 7 to 10, note that $\lambda$ stands for the degree of financial frictions: The larger the $\lambda$ is, the smaller the financial friction it stands for. That is, larger $\lambda$ means smaller financial frictions. Meanwhile, $\theta$ and $\sigma$ stand for the average return of the risky assets and the volatility of the risky assets, respectively. $f_{\sigma}\in(0,1)$ is a weight factor for total consumption. The cognitive resources of the second type of agent are diluted when interacting with big data, resulting in a decrease in the rationality of decision-making. In our discussion, this is an exogenous condition that distinguishes the second type of agent from the first type. Furthermore, information uncertainty (big data value) increases with the growth of financial friction. Therefore, the consumption of the second type of agent that can be converted into utility is a weight of total consumption, measured by $f_{\sigma}$, and the larger the $\lambda$ is, the greater the $f_{\sigma}$ is. In this change of $\lambda$ and $f_{\sigma}$, the cognitive resources of the second type of agent will be continuously diluted through interaction with big data. Therefore, the change in $\lambda$ ensures the external condition of $f_{\sigma}$, while the dilution of cognitive resources ensures the internal condition. Under these conditions, $c^{utility}_t=f_{\sigma}c_t<c_t$ can be established.

(\textbf{A}) The First Type of Agent: Figure 7 and Figure 8 plot the (logarithmic) wealth distribution when $\lambda^L=5,\lambda_M=25,\lambda^H=50$ with $(\theta^-,\sigma^-)=(0.05,0.05)$ and $(\theta^+,\sigma^+)=(0.5,0.5)$. We see that as $\lambda$ gets larger (i.e., as financial frictions get smaller), the mean of the (logarithmic) wealth distribution ($\mathbb{E}(x)$) becomes larger and the Pareto tail (Var(x)) of the (logarithmic) wealth gets thick, which means that as financial frictions get smaller, although the average wealth will increase, the (top) wealth inequality gets larger.

(\textbf{B}) The Second Type of Agent: Figure 7 and Figure 8 plot the (logarithmic) wealth distribution when $(\lambda^L,f_{\sigma}^L)=(5,0.2),(\lambda^M,f_{\sigma}^M)=(25,0.5),(\lambda^H,f_{\sigma}^H)=(50,0.8)$ with $(\theta^-,\sigma^-)=(0.05,0.05)$ and $(\theta^+,\sigma^+)=(0.5,0.5)$. Unlike the (logarithmic) wealth distribution of the first type of agent, which is all right-skewed, after introducing $f_{\sigma}$, the joint change of $\lambda$ and $f_{\sigma}$ will cause the (logarithmic) wealth distribution graph to shift from a left-skewed distribution to a normal distribution and then to a right-skewed distribution. And we can see that as $\lambda$ and $f_{\sigma}$ get larger, the mean of the (logarithmic) wealth distribution ($\mathbb{E}(x)$) becomes larger. However, for the Pareto tail (Var(x)) of the (logarithmic) wealth, it does not change linearly with $(\lambda^i,f_{\sigma}^i),i\in\{L,M,H\}$, as we can see, the Pareto tail (Var(x)) of the (logarithmic) wealth is thick under the condition of $(\lambda^L,f_{\sigma}^L)$ and $(\lambda^H,f_{\sigma}^H)$ while the the Pareto tail (Var(x)) is thin under the condition of $(\lambda^M,f_{\sigma}^M)$. That means: Although the reduction in financial friction and the increase in utility conversion weights have led to an increase in the average wealth, when we introduce the consumption-utility conversion weight $f_{\sigma}$, we find that wealth distribution inequality is significantly large under both high weight ($f_{\sigma}^H$) and low weight ($f_{\sigma}^L$). Only under a medium weight, i.e., $f_{\sigma}^M=0.5$, the wealth distribution is relatively concentrated, which is similar to a short-tailed normal distribution.

Now, if we fix $\lambda$ at the levels $\lambda^L=5,\lambda^M=25,\lambda^H=50$, we can explore how the volatility rate $\sigma$ and average return rate $\theta$ of risk assets affect the wealth inequality by observing the thickness of the respective Pareto tail. By analyzing Figures 7, 8, 9, and 10, we can see that in both the first and second types of situations, when the coefficient group of risk assets increases from $(\theta^-,\sigma^-)$ to $(\theta^+,\sigma^+)$, the peak of all logarithmic wealth distributions decreases, which means that the Pareto tail of all distributions becomes thicker. Therefore, it seems like a larger coefficient combination of $(\theta,\sigma)$ will lead to greater wealth inequality. Based on this, we will next use the \textit{Control Variable Method} on $(\theta,\sigma)$ to conduct further robust research. The robust tests see Figure 11, 12, 13 and 14:

\begin{figure}[htbp]
\centering
\includegraphics[width=13cm]{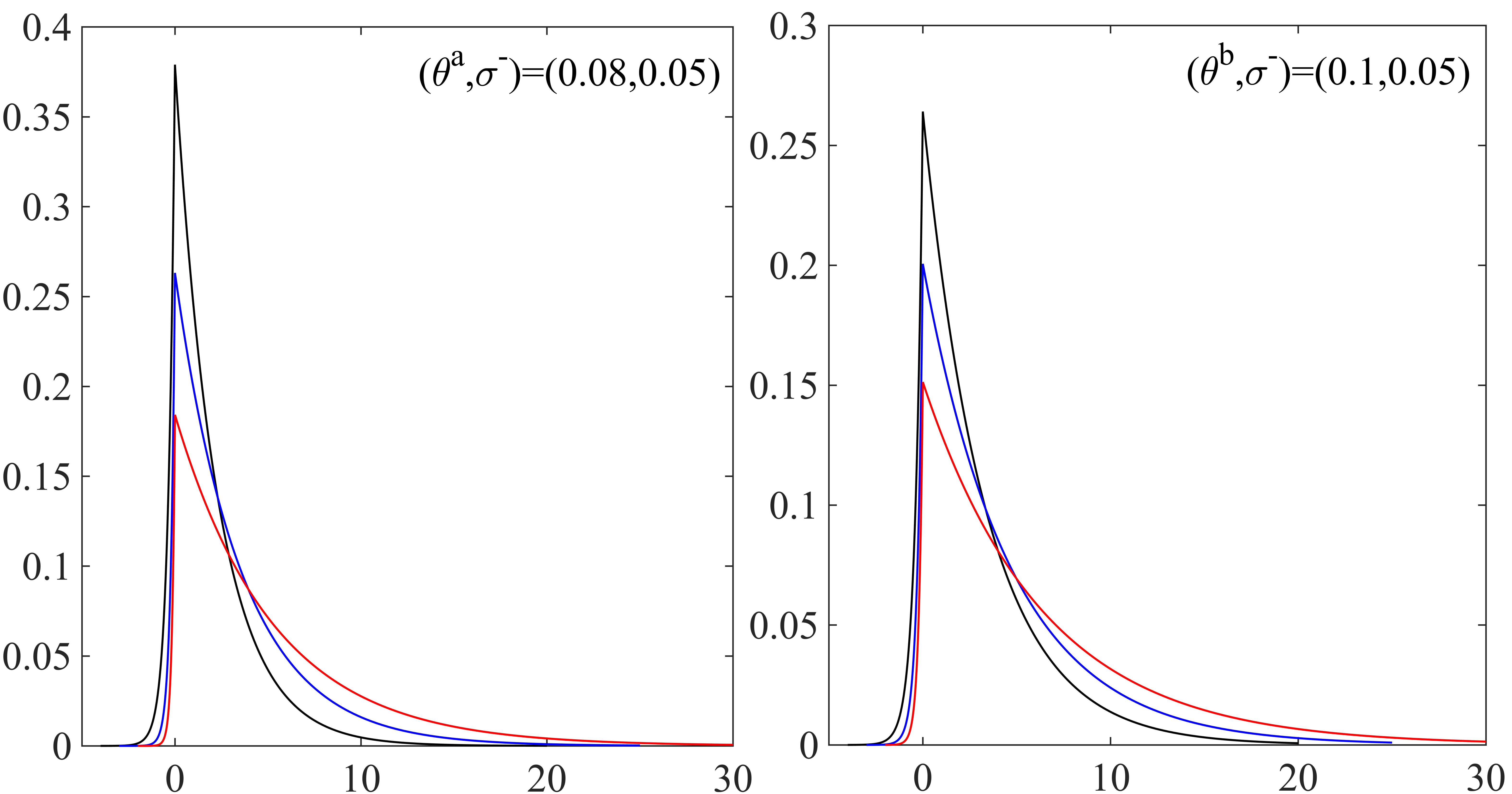}
\caption{\label{fig:F11}Logarithmic Wealth Distribution $\bar{p}(x)$ with $(\theta^a,\sigma^-)$ and $(\theta^b,\sigma^-)$}
\end{figure}
\begin{figure}[htbp]
\centering
\includegraphics[width=13cm]{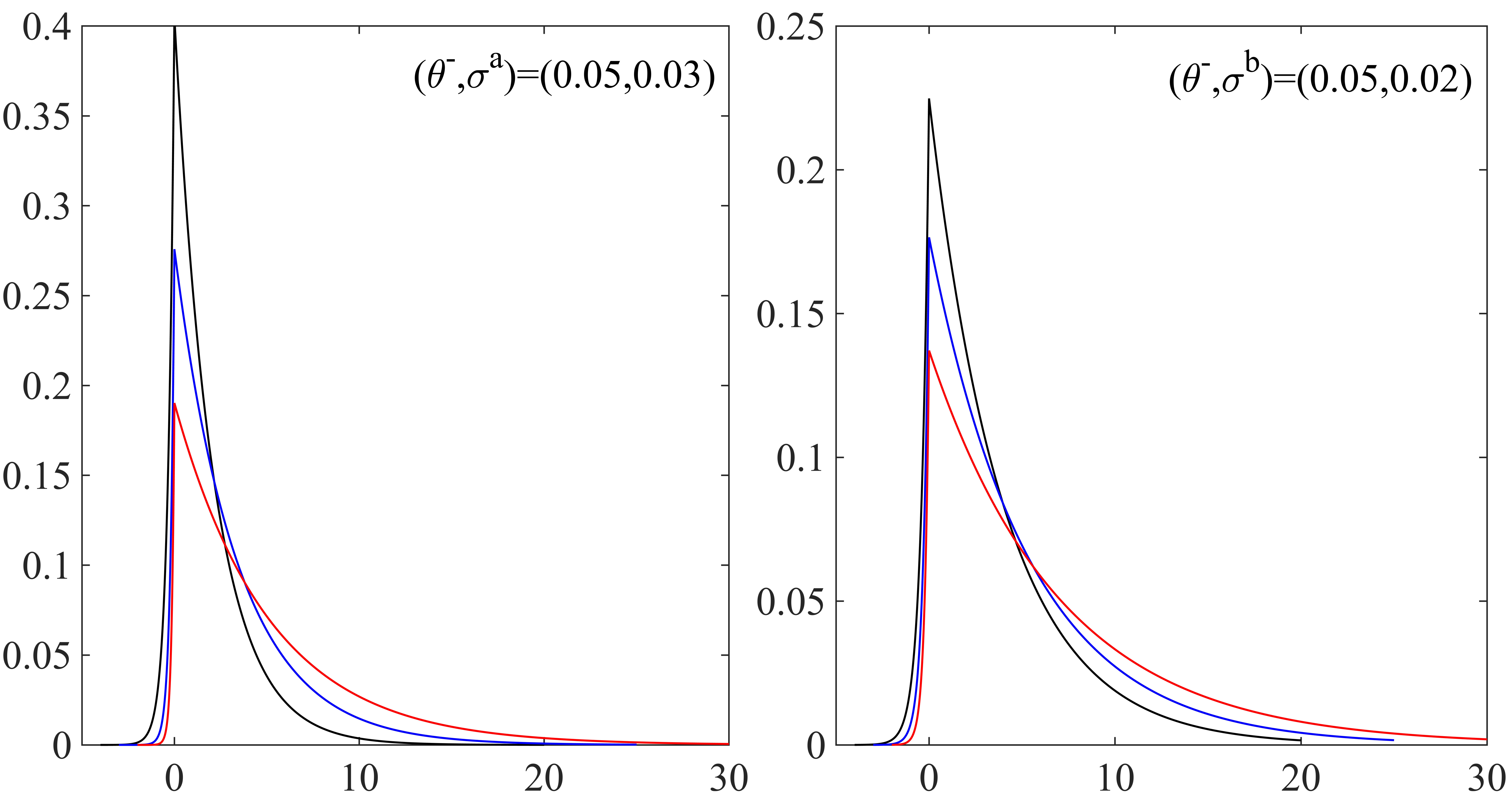}
\caption{\label{fig:F12}Logarithmic Wealth Distribution $\bar{p}(x)$ with $(\theta^-,\sigma^a)$ and $(\theta^-,\sigma^b)$}
\end{figure}
\begin{figure}[htbp]
\centering
\includegraphics[width=13cm]{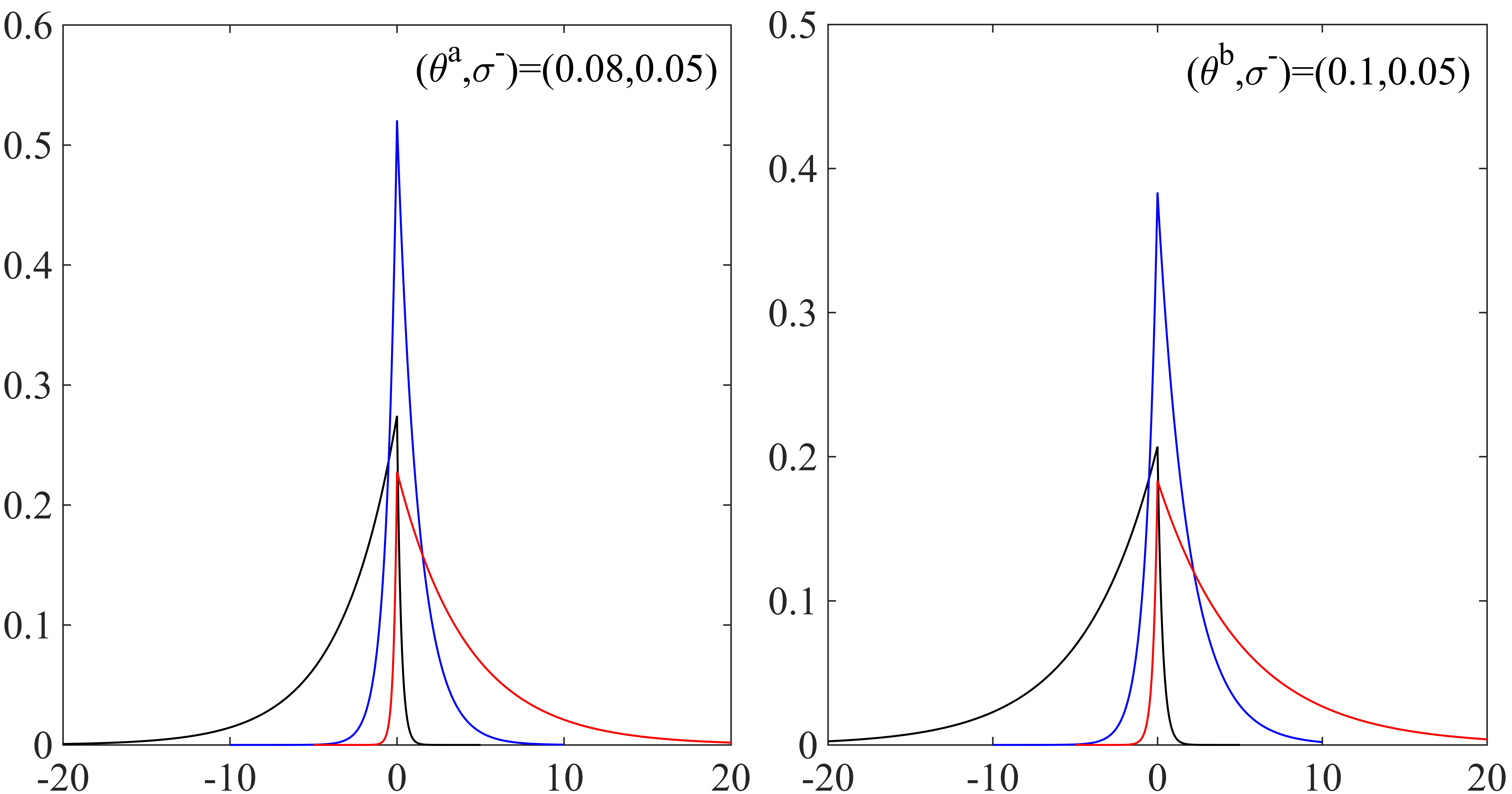}
\caption{\label{fig:F13}Logarithmic Wealth Distribution $\bar{p}^{\dagger}(x)$ with $(\theta^a,\sigma^-)$ and $(\theta^b,\sigma^-)$}
\end{figure}
\begin{figure}[htbp]
\centering
\includegraphics[width=13cm]{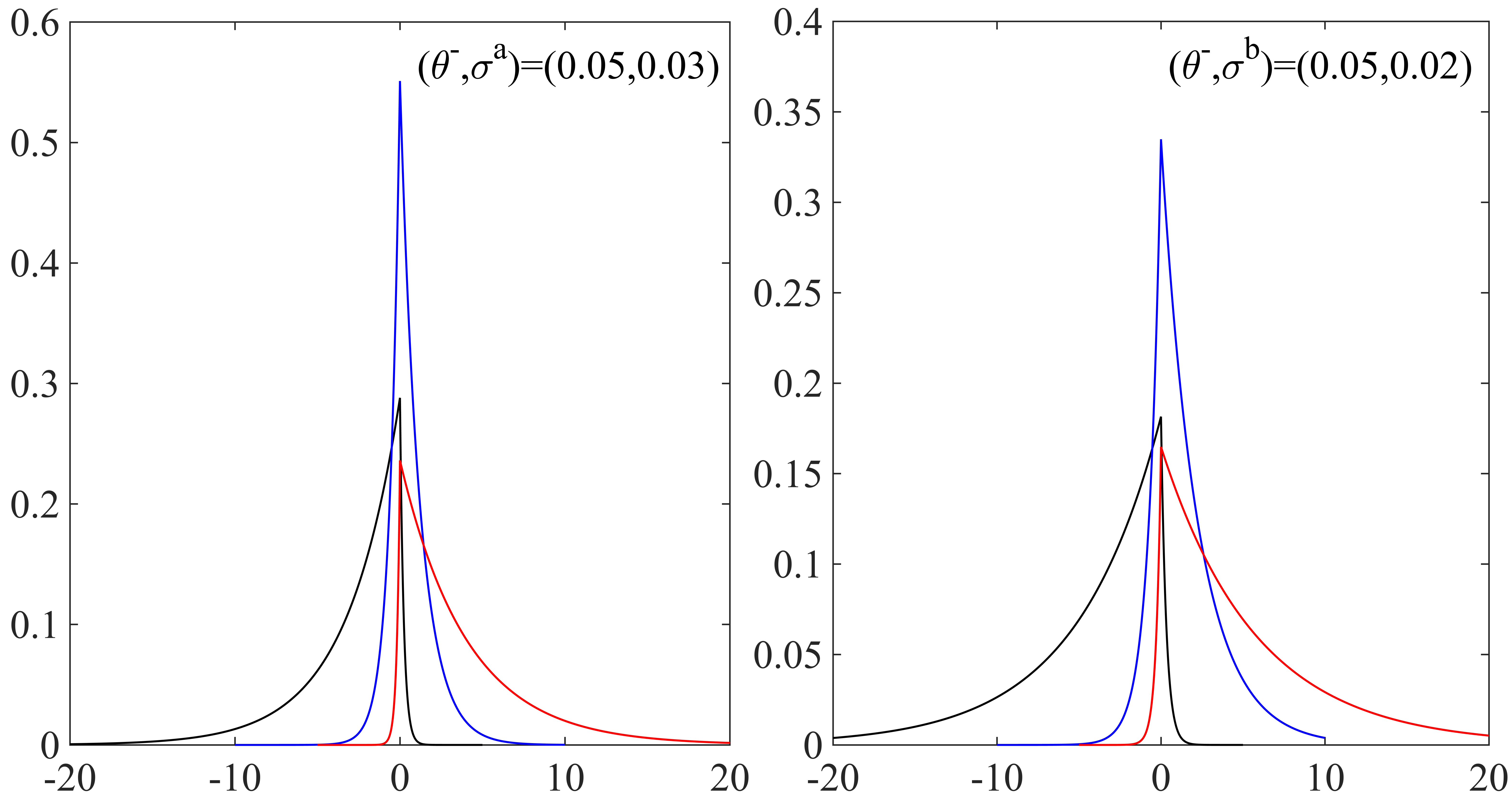}
\caption{\label{fig:F14}Logarithmic Wealth Distribution $\bar{p}^{\dagger}(x)$ with $(\theta^-,\sigma^a)$ and $(\theta^-,\sigma^b)$}
\end{figure}

The robust tests are shown\footnote{The color labels for the curves in Figures 11 to 14 are consistent with those in Figures 7 to 9, with black representing high financial friction $\lambda^L$, blue representing medium financial friction $\lambda^M$, and red representing low financial friction $\lambda^H$.} by Figure 11,12,13 and 14. We ultimately found that for both types of agents: when the average return rate $\theta$ of risk assets is fixed, the Pareto tail becomes thicker as the volatility rate $\sigma$ decreases. When the volatility rate $\sigma$ of risk assets is fixed, the Pareto tail becomes thicker as the average return rate $\theta$ increases. Therefore, we can summarize the third conclusion:

(\textbf{C}) Both Types of Agents: Comparing Figures 7 and 8 or Figures 9 and 10, we can see that investing in risk assets with higher average returns and volatility leads to greater wealth inequality. At the same time, comparing Figures 7 and 11, 12, or Figures 9, 13, and 14, we can see that investing only in assets with higher average returns leads to greater wealth inequality, and investing only in risk assets with lower volatility also leads to greater wealth inequality.

\noindent\textbf{Empirical Analysis}:According to the analysis in Section 4.3 on numerical simulation, we can see that for both the first and second types of agents, a reduction in financial friction (increase in $\lambda$) will be beneficial to the growth of average wealth and lead to greater wealth inequality. For the second type of agent, when we introduce the utility conversion weight $f_{\sigma}$, we find that changes from f and lambda have similar effects on the firm's wealth. Specifically, when $f_{\sigma}$ and $\lambda$ are high, the financial friction in the economic environment is low and the proportion of agents’ total consumption converted into utility is high, which will result in a higher average wealth. Based on this, we further propose Proposition 2:

\begin{proposition}
\textit{Lower financial friction and information uncertainty will be beneficial to firm wealth accumulation, increase market share, and promote more firms to enter the high market share group. However, high financial friction and information uncertainty will be negative for firm market share growth.}
\end{proposition}

According to the econometric analysis in Section 3.2, we will use data from Chinese A-share listed companies from 2000 to 2023 as the analysis sample. This section has two dependent variables. The first is market share variable (\textbf{Wealth}) that measures firm wealth, which is measured by the ratio of a firm's operating revenue to the total revenue of the industry. The second is a grouping variable (\textbf{Group}) used to determine whether a firm is in the high market share group. Firm market share groups are defined based on whether a firm’s market share exceeds the median of the sample. If a firm is in the high market share group, \textbf{Group}=1, if not, \textbf{Group}=0. For the binary variable \textbf{Group}, we will use the \textit{Logit} regression model to analyze. For the independent variables, the first is the credit availability variable (\textbf{Friction}) to measure financial friction. It is calculated as “new corporate debt (difference between beginning-of-year and end-of-year debt) / total corporate assets.” The magnitude of \textbf{Friction} measures the level of annual available capital for the company, meaning that the larger \textbf{Friction} is, the lower the financial friction. The second is the uncertainty variable (\textbf{Uncertainty}), which in this section serves as a proxy for the utility conversion weight $f_{\sigma}$, meaning that in an environment of higher uncertainty, the proportion of agents’ total consumption converted into utility is lower. The calculation method for \textbf{Uncertainty} is consistent with Section 3.2, requiring consideration of both information uncertainty and economic policy uncertainty. The control variables are consistent with those used in Table 1 of Section 3.2. The regression results are shown in Table 3\footnote{Standard errors in parentheses: $^{*} p<0.05$, $^{**} p<0.01$, $^{***} p<0.001$.}:

\begin{table}[ht]
\centering
\caption{Regression Results B}
\label{tab:results B}
\small
\setlength{\tabcolsep}{3pt}
\begin{tabular}{l*{8}{c}}
\toprule
 & \multicolumn{2}{c}{Wealth} & \multicolumn{2}{c}{Group} & \multicolumn{2}{c}{Wealth} & \multicolumn{2}{c}{Group} \\
\cmidrule(lr){2-3} \cmidrule(lr){4-5} \cmidrule(lr){6-7} \cmidrule(lr){8-9}
 & (1) & (2) & (3) & (4) & (5) & (6) & (7) & (8) \\
\midrule
Fraction & $0.007^{***}$ & $0.005^{***}$ & $1.540^{***}$ & $1.677^{***}$ &  &  &  &  \\
 & (0.001) & (0.001) & (0.164) & (0.192) &  &  &  &  \\
Uncertainty &  &  &  &  & $-0.002^{***}$ & $-0.002^{***}$ & $-0.635^{***}$ & $-0.745^{***}$ \\
 &  &  &  &  & (0.000) & (0.000) & (0.036) & (0.041) \\
DER &  & $0.001^{***}$ &  & $0.154^{***}$ &  & $0.001^{***}$ &  & $0.146^{***}$ \\
 &  & (0.000) &  & (0.028) &  & (0.000) &  & (0.024) \\
ROA &  & $0.018^{***}$ &  & $6.438^{***}$ &  & $0.022^{***}$ &  & $6.289^{***}$ \\
 &  & (0.002) &  & (0.505) &  & (0.002) &  & (0.435) \\
REC &  & $-0.000$ &  & $0.047$ &  & $-0.000$ &  & $0.847^{*}$ \\
 &  & (0.002) &  & (0.428) &  & (0.002) &  & (0.335) \\
INV &  & $-0.004^{**}$ &  & $-0.500$ &  & $-0.003$ &  & $0.061$ \\
 &  & (0.002) &  & (0.317) &  & (0.001) &  & (0.250) \\
CAP &  & $-0.001^{***}$ &  & $-0.543^{***}$ &  & $-0.001^{***}$ &  & $-0.511^{***}$ \\
 &  & (0.000) &  & (0.024) &  & (0.000) &  & (0.018) \\
BM &  & $0.000$ &  & $0.545^{***}$ &  & $0.001^{***}$ &  & $0.644^{***}$ \\
 &  & (0.000) &  & (0.039) &  & (0.000) &  & (0.035) \\
TobinQ &  & $-0.001^{***}$ &  & $-0.288^{***}$ &  & $-0.001^{***}$ &  & $-0.237^{***}$ \\
 &  & (0.000) &  & (0.027) &  & (0.000) &  & (0.024) \\
Occupy &  & $-0.006$ &  & $-1.954^{*}$ &  & $-0.021^{***}$ &  & $-2.821^{***}$ \\
 &  & (0.005) &  & (0.966) &  & (0.004) &  & (0.681) \\
GrossProfit &  & $-0.008^{***}$ &  & $-0.956^{***}$ &  & $-0.009^{***}$ &  & $-1.213^{***}$ \\
 &  & (0.001) &  & (0.287) &  & (0.001) &  & (0.235) \\
Constant & $0.026^{***}$ & $0.032^{***}$ &  &  & $0.036^{***}$ & $0.044^{***}$ &  &  \\
 & (0.001) & (0.001) &  &  & (0.001) & (0.001) &  &  \\
Year & $\surd$ & $\surd$ & $\surd$ & $\surd$ & $\surd$ & $\surd$ & $\surd$ & $\surd$ \\
ID & $\surd$ & $\surd$ & $\surd$ & $\surd$ & $\surd$ & $\surd$ & $\surd$ & $\surd$ \\
\midrule
$R^{2}$ & 0.043 & 0.057 &  &  & 0.063 & 0.079 &  &  \\
Number & 44577 & 43559 & 16587 & 16036 & 54781 & 53682 & 23525 & 22890 \\
\bottomrule
\end{tabular}

\vspace{5pt}
\raggedright
\end{table}

In Table 3, the result is consistent with \textbf{Proposition 2}: The independent variables in columns (1) to (4) are financial friction (\textbf{Friction}). According to the regression results, after controlling for the control variables and using a two-way fixed effects model, the coefficients of the independent variables are all significantly positive at the $1\%$ level, indicating that a reduction in financial friction, i.e., an increase in credit availability (\textbf{Friction}), effectively promotes firms' wealth growth. Furthermore, as shown by the Logit regression results, an increase in \textbf{Friction} significantly enhances the probability of \textbf{Group}=1, meaning that in an environment with lower financial friction, more firms enter the high market share group. The independent variables in columns (5) to (8) are uncertainty (\textbf{Uncertainty}). From the regression results, after controlling for variables and using a two-way fixed effects model, the coefficients of the independent variables are all significantly negative at the $1\%$ level, indicating that an increase in uncertainty (\textbf{Uncertainty}) reduces the weight of agents' consumption converted into utility, which inhibits firms' wealth growth. Furthermore, as indicated by the Logit regression results, an increase in uncertainty significantly reduces the probability of \textbf{Group}=1, meaning that in a higher uncertainty environment, fewer firms enter the high market share group.

\subsection{Extended Analysis: The Equilibrium Situation according to MFG}
\textbf{The Analytical Solution of Equilibrium}: Again, the steady-state system according to \textit{Mean Field Game} (MFG) (\citealp{achdou2022income}) in the entrepreneur-worker model is given by:
\[\begin{gathered}\rho v(a)=\max_{c,\kappa}\frac{c^{1-\gamma}}{1-\gamma}+(\pi(a)+ra+(\theta-r)\kappa-c)v^{\prime}(a)+\frac{1}{2}\sigma^2\kappa^2v^{\prime\prime}(a)\\0=-\frac{d}{da}[(\pi(a)+ra+(\theta-r)\kappa-c)\bar{p}(a)]+\frac{1}{2}\frac{d}{da^2}\left[\sigma^2\kappa^2\bar{p}(a)\right]-\beta\bar{p}(a)+\beta\delta(a-1)\\\int_0^\infty\lambda a\bar{p}(a)da=\int_0^\infty{\left(a-\left(\frac{\theta-r}{\gamma\sigma^2}\right)a\right)}\bar{p}(a)da\\\int_0^\infty\left(\frac{1-\alpha}{w}\right)^{\frac{1}{\alpha}}z\lambda a\bar{p}(a)da=1\end{gathered}\]

According to the equilibrium conditions above, i.e., the capital market clearing condition and the labor market clearing condition, we can solve equilibrium interest rate $r^*$ and equilibrium wage rate $w^*$.

It follows directly from capital market clearing condition:
\[\int_0^\infty\lambda a\bar{p}(a)da=\int_0^\infty{\left(a-\left(\frac{\theta-r}{\gamma\sigma^2}\right)a\right)}\bar{p}(a)da\]

Then, we get $r^*$:
\begin{align}
r^*=\theta-\gamma\sigma^2(1-\lambda)\label{27}
\end{align}

To determine the equilibrium wage rate $w^*$, we have to use the labor market clearing condition, when let $x=log(a)$, we get:
\begin{align}
\left(\frac{1-\alpha}{w}\right)^{\frac{1}{\alpha}}z\lambda\int_{-\infty}^\infty e^x\bar{p}(x)dx=1\label{28}
\end{align}

And:
\[\int_{-\infty}^\infty\bar{p}(x)dx=1\]

Plugging equations (21) and (22) into equation (28), and solving the new equation:
\[\begin{gathered}
\frac{\left(\frac{1-\alpha}{w}\right)^{\frac{1}{\alpha}}\beta z\lambda}{\sqrt{\mu^2+2\beta\Sigma^2}}{\left[\int_{-\infty}^0e^x\exp\left(\frac{\mu+\sqrt{\mu^2+2\beta\Sigma^2}}{\Sigma^2}x\right)dx+\int_0^\infty e^x\exp\left(\frac{\mu-\sqrt{\mu^2+2\beta\Sigma^2}}{\Sigma^2}x\right)dx\right]}=1\\\frac{\left(\frac{1-\alpha}{w}\right)^{\frac{1}{\alpha}}\beta z\lambda}{\sqrt{\mu^2+2\beta\Sigma^2}}\times\frac{2\sqrt{\mu^2+2\beta\Sigma^2}}{2\beta-2\mu-\Sigma^2}=1\\\frac{\left(\frac{1-\alpha}{w}\right)^{\frac{1}{\alpha}}\beta z\lambda}{2\beta-2\mu-\Sigma^2}=\frac{1}{2}\end{gathered}\]

Then, we set a series of algebraic equations and insert them into the solution:
\[\begin{gathered}\mu=\left(\Pi+r+\frac{\left(\theta-r\right)^2}{\gamma\sigma^2}-\frac{1}{\gamma}{\left(\rho-(1-\gamma)(\Pi+r)-\frac{1}{2}(1-\gamma)\frac{\left(\theta-r\right)^2}{\gamma\sigma^2}\right)}-\frac{1}{2}\frac{\left(\theta-r\right)^2}{\gamma^2\sigma^2}\right)\\\Pi=\left(\alpha z\left(\frac{1-\alpha}{w}\right)^{\frac{1-\alpha}{\alpha}}-r-\delta\right)\lambda\\t=\left(\frac{1-\alpha}{w}\right)^{\frac{1}{\alpha}}\\\Sigma=\frac{\theta-r}{\gamma\sigma}\end{gathered}\]

And We get:
\begin{align}
2\beta z\lambda\gamma t+2\alpha z\lambda t^{1-\alpha}+2(r-\lambda r-\lambda\delta-\rho)+\frac{(1+\gamma)(\theta-r)^2}{\gamma\sigma^2}-2\beta\gamma=0\label{29}
\end{align}

In general, the analytic formula for $t$ and hence the equilibrium wage rate $w^*$ is hard to find. However, if we let $\alpha=0.5$ which is not very far from $\alpha=0.3$ that is widely used in literature, we then can apply the root formula for quadratic function to solve for $t$ and hence the equilibrium wage rate $w^*$. That is, let $\alpha=0.5$, then solve equation (29), we get:
\[\begin{gathered}2\beta z\lambda\gamma\left(\sqrt{t}\right)^2+z\lambda\sqrt{t}+2(r-\lambda r-\lambda\delta-\rho)+\frac{(1+\gamma)(\theta-r)^2}{\gamma\sigma^2}-2\beta\gamma=0\\\sqrt{t}=\frac{-z\lambda+\sqrt{z^2\lambda^2-8\beta z\lambda\gamma\left(2(r-\lambda r-\lambda\delta-\rho)+\frac{(1+\gamma)(\theta-r)^2}{\gamma\sigma^2}-2\beta\gamma\right)}}{4\beta z\lambda\gamma}\\2(r-\lambda r-\lambda\delta-\rho)+\frac{(1+\gamma)(\theta-r)^2}{\gamma\sigma^2}-2\beta\gamma<0\end{gathered}\]

Hence:
\[\left(\frac{1-\alpha}{w}\right)^{\frac{1}{\alpha}}=\left(\frac{-z\lambda+\sqrt{z^2\lambda^2-8\beta z\lambda\gamma\left(2(r-\lambda r-\lambda\delta-\rho)+\frac{(1+\gamma)(\theta-r)^2}{\gamma\sigma^2}-2\beta\gamma\right)}}{4\beta z\lambda\gamma}\right)^2\]

Then, given $\alpha=0.5$, and compute the equilibrium wage rate $w^*$:
\[\begin{gathered}w^*=(1-\alpha)\left(\frac{-z\lambda+\sqrt{z^2\lambda^2-8\beta z\lambda\gamma{\left(2(r-\lambda r-\lambda\delta-\rho)+\frac{(1+\gamma)(\theta-r)^2}{\gamma\sigma^2}-2\beta\gamma\right)}}}{4\beta z\lambda\gamma}\right)^{-2\alpha}\\w^*=\frac{2\beta z\lambda\gamma}{-z\lambda+\sqrt{z^2\lambda^2-8\beta z\lambda\gamma\left(2(r-\lambda r-\lambda\delta-\rho)+\frac{(1+\gamma)(\theta-r)^2}{\gamma\sigma^2}-2\beta\gamma\right)}}\end{gathered}\]

Now, plugging in the equilibrium interest rate $r^*$, we finally get $w^*$ clearly:
\begin{align}
w^*=\frac{2\beta z\lambda\gamma}{-z\lambda+\sqrt{\begin{bmatrix}z^2\lambda^2\\-16\beta z\lambda\gamma(\theta(1-\lambda)-\gamma\sigma^2(1-\lambda)^2-\lambda\delta-\rho)\\-8\beta z\lambda\gamma(\gamma(1+\gamma)\sigma^2(1-\lambda)^2-2\beta\gamma)\end{bmatrix}}}\label{30}
\end{align}

To summarize, the \textit{Mean Field Game} determines the equilibrium interest rate $r^*$ and the equilibrium wage rate $w^*$, which are given by equation (27) and (30) respectively. Now, according to the equilibrium interest rate $r^*$ and equilibrium wage rate $w^*$, our algebraic equations of wealth distribution with financial friction could be rewritten as:
\[\begin{gathered}\mu^*=\left(\Pi^*+r^*+\frac{\left(\theta-r^*\right)^2}{\gamma\sigma^2}-\frac{1}{\gamma}{\left(\rho-(1-\gamma)(\Pi^*+r^*)-\frac{1}{2}(1-\gamma)\frac{\left(\theta-r^*\right)^2}{\gamma\sigma^2}\right)}-\frac{1}{2}\frac{\left(\theta-r^*\right)^2}{\gamma^2\sigma^2}\right)\\\Pi^*=\left(\alpha z{\left(\frac{1-\alpha}{w^*}\right)}^{\frac{1-\alpha}{\alpha}}-r^*-\delta\right)\lambda\\\Sigma^*=\frac{\theta-r^*}{\gamma\sigma}\end{gathered}\]

Then, equation (19) becomes:
\begin{align}
dx=\mu^*dt+\Sigma^*dW_t\label{31}
\end{align}

 Therefore, we can solve the wealth distribution density function under the equilibrium situation according to MFG:
 \begin{align}
\bar{p}^*(x)=\frac{\beta}{\sqrt{{\mu^{*}}^2+2\beta{\Sigma^*}^2}}\exp{\left(\frac{\mu^*+\sqrt{{\mu^{*}}^2+2\beta{\Sigma^*}^2}}{{\Sigma^*}^2}x\right)},x<0\label{32}\\\bar{p}^{*}(x)=\frac{\beta}{\sqrt{{\mu^{*}}^2+2\beta{\Sigma^*}^2}}\exp{\left(\frac{\mu^{*}-\sqrt{{\mu^{*}}^2+2\beta{\Sigma^*}^2}}{{\Sigma^*}^2}x\right)},x>0\label{33}
\end{align}

Theoretically, equations (32) and (33) prove the existence of an equilibrium wealth distribution for firms.

\noindent\textbf{Further Prospect}: In Section 4, the setting of heterogeneity is exogenous, which means that the utility conversion weight $f_{\sigma}$ for the second type of agent is exogenously given. This is because we assume that the second type of agent interacts with big data, resulting in the dilution of their cognitive resources, then, the consumption that can be converted into utility is a weight of total consumption, and information uncertainty (data value) is commensurate with the magnitude of financial friction $\lambda$. Next, we will preliminarily attempt to propose a conceptual approach that makes effective consumption of agents with their cognitive resources endogenous\footnote{The rationality degree of agents' consumption stems from the quantity of cognitive resources they possess, more cognitive resources lead to higher state of rationality.}. Referring to the methods used by \citet{chang2024heterogeneity}: Based on the analysis in Section 2.2 (\textbf{Theorem 2}) of our paper, we obtained the dynamic distribution of cognitive resources $R_t$ and the equilibrium solution of cognitive resources $R^*$ in big data interaction. Now, we assume that the agent's effective consumption is a function of cognitive resources, $C_t(R_t)=\mathbb{F}(R_t)$, so different cognitive resources lead to different effective consumption levels, and the equilibrium effective consumption is $C^*(R^*)=\mathbb{F}(R^*)$. Based on the analysis in Section 4.4, let the equilibrium wealth of the firm be denoted as $W^*$. We further decompose the effective consumption $C_t$ and firm wealth $W_t$ under general conditions into deterministic equilibrium states $(C^*,W^*)$ and uncertain fluctuations $(\tilde{C}_t,\tilde{W}_t)$:
\[\begin{gathered}W_t=W^*+\tilde{W}_t\\C_t=C^*+\tilde{C}_t\end{gathered}\]

Then, $(W_t,C_t(R))$ evolve jointly according to the following linear functional VAR:

\[\begin{gathered}\tilde{W}_t=\mathcal{B}_{ww}\tilde{W}_{t-1}+\int \mathcal{B}_{wc}(\tilde{R})\tilde{C}_{t-1}(\tilde{R})d\tilde{R}+v_{w,t}\\\tilde{C}_t(R)=\mathcal{B}_{cw}(R)\tilde{W}_{t-1}+\int \mathcal{B}_{cc}(R,\tilde{R})\tilde{C}_{t-1}(\tilde{R})d\tilde{R}+v_{c,t}(R)\end{gathered}\]

Among these two functions above: $\mathcal{B}_{ww}$ represents\footnote{This is different from \citet{chang2024heterogeneity}, where $\mathcal{B}_{ww}$ is a multidimensional matrix, i.e., $\mathcal{B}_{ww}\in\mathbb{R}^n$, representing the impact of lagged values of various macroeconomic variables on their own and other variables' current values. In our discussion, we only focus on the single macro variable of firm wealth $W_t$.} the impact of the previous period's firm wealth $W_{t-1}$ on the current period's wealth $W_{t}$, $\mathcal{B}_{ww}\in\mathbb{R}^1$. $\mathcal{B}_{wc}$ is a function that takes cognitive resource points $\tilde{R}$ as input\footnote{For convenience, when we talk about a specific cognitive resource level $R_0$, we are talking about the consumption $C_t(R_0)$ determined by the cognitive resource point $R_0$, and the same applies below.}. Its value at each point $\tilde{R}$ quantifies the marginal impact of a small change in the distribution at the cognitive resource level $\tilde{R}$ on the macro variable $\tilde{W}_t$. $\mathcal{B}_{cw}$ is a function that takes cognitive resource point $R$ as its input\footnote{This can be understood as the heterogeneity of the impact of exogenous shocks contained in $\tilde{W}_{t-1}$ on consumption determined by cognitive resources, i.e., different shocks have different effects on agents’ consumption with varying levels of cognitive resources.}. Its value at each point $R$ quantifies the marginal impact of the lagged macroeconomic variable $\tilde{W}_{t-1}$ on the distribution density of the cognitive resource at level $R$. $\mathcal{B}_{cc}$ is a kernel density function that describes how changes in the density at each cognitive resource point $\tilde{R}$ in the previous period's cognitive resource distribution affect the density at another cognitive resource level point $R$ in the current period's cognitive resource distribution. That is, agents at the $\tilde{R}$ level in the previous period have a certain “probability” of moving to the $R$ level in the current period, and hence changing the density at the $R$ point\footnote{For example, if $R>\tilde{R}$ and $\mathcal{B}_{cc}$ is positive, then agents with a consumption level of $C_{t-1}(\tilde{R})$ in the previous period have a certain probability of shifting to a consumption level of $C_t(R)$ in the current period. When $R=\tilde{R}$ and $\mathcal{B}_{cc}$ is at its peak, this indicates that agents have a very high probability of maintaining the same consumption level in the current period as in the previous period.}. Then, $v_{w,t}$ and $v_{c,t}(R)$ both represent simplified stochastic shocks. Here $v_{w,t}$ is mean-zero random vector with covariance $\Omega_{ww}$ and $v_{c,t}(R)$ is a random element in a \textit{Hilbert Space} with covariance function $\Omega_{cc}(R,\tilde{R})$. Now, according to the function (5), which could be shifted to a \textit{Geometric Brownian Motion} (GBM), we can get the restrictions of $R_t$:
\[\frac{dR_i(t)}{R_i(t)}=\left[\mu^c(\theta^c-1)-\eta^c\sigma^c(n-1)^{1-\gamma^c}\right]dt+\psi^c\left[\theta^c-\frac{\mu^c\theta^c}{\mu^c+\eta^c\sigma^c(n-1)^{1-\gamma^c}}\right]dW(t)\]

Let:
\[\begin{gathered}X=\left(\mu^c(\theta^c-1)-\eta^c\sigma^c(n-1)^{1-\gamma^c}\right)\\Y=\psi^c\left(\theta^c-\frac{\mu^c\theta^c}{\mu^c+\eta^c\sigma^c(n-1)^{1-\gamma^c}}\right)
\end{gathered}\]

We get:
\[R_t{\sim}\mathcal{N}\left(R_0e^{Xt},R_0^2e^{2Xt}\left(e^{Y^2t}-1\right)\right)\]

Based on the initial analysis above, we can further transform the research framework of this paper from static to dynamic. This is the basis for introducing policy shocks that may affect agents' cognitive resources in the future research, so that we can explore the impact of policies on effective consumption and wealth distribution.

\section{Conclusion and Marginal Contribution}

\subsection{Summary}
\textbf{Overarching Issue}: Our article defines the third category of information beyond the signals from macroeconomic fluctuations (the state of the economic cycle and market) and microeconomic factor (the attributes of the transactions, auctions and products): Big data composed of general information entropy. This type of data is a creation of the rapid development of platform technology and artificial intelligence in the digital era. It is typically based on underlying algorithmic designs and disseminated through networks and platforms in professional or entertainment-oriented formats, continuously and on a large scale, to be received by agents.

Unlike conventional information economics analysis frameworks, since this type of big data is entirely composed of information entropy, its value lies solely in the information uncertainty measured by the magnitude of information entropy, rather than being a form of knowledge or signal that serves as an actual economic variable proxy. This is different from all “Knowledge Economy” based on \citet{hayek1945use}, so acquiring this type of big data does not require agents to pay material costs. Instead, we treat “cognitive resources” as one of the agents' resource endowments, viewing the formation and circulation of big data as an interaction between algorithm-driven intelligent entities and agents' “cognitive resources.” Agents obtain big data and engage in interactive behaviors by expending cognitive resources. Theoretical analysis demonstrates that as the time and scale of agents' interactions with big data expand, their cognitive resources will decrease. This ensures that agents' rational state will be influenced by big data, and thus affect the effectiveness of consumption. Based on this, we further employ prospect theory, an exogenous “investment-tax” model, and empirical evidence to respectively demonstrate the direction and specific amount of agents’ consumption adjustments under big data interaction, hence concretely constructing the Consumption Adjustment Weight Function “CAWF” for agents under big data interaction. We attempt to summarize the issues we have explored, including CAWF, using the following mathematical analysis:

Assuming that the accumulation of big data doesn't exceed the time dimension\footnote{Big data interaction is an optional behavior to agents, meaning that the influence of big data interaction on agents' consumption behavior occurs only if the agent engaged in big data interaction in the previous stage. For this reason, we define that the accumulation of big data does not exceed the time dimension. This differs from normal belief updating. In standard information economics, since macroeconomic or microeconomic signals are objectively present, agents inevitably receive information and undergo belief shifts while demanding and consuming, resulting in continuous consumption with belief updating.}. Hence, under infinite time, the accumulation of big data does not deviate from the existence of time. Thus, we denote time as $t$ and big data as $n(t)$:
\[\int_0^\infty n(t)dt<\int_0^\infty tdt\]

Let:
\[n(t)=t^{\frac{1}{k}},k\in(1,+\infty)\]

Hence:
\[\lim_{n\to+\infty,t\to+\infty}\frac{n}{t}=\lim_{t\to+\infty}\frac{1}{k}{\left(\frac{1}{t}\right)}^{\left(1-\frac{1}{k}\right)}\to0\]

We introduce the \textit{Dirac Delta Function} $\delta$ and denote consumption as a continuous function $C(t)$ with time:
\[\begin{gathered}\delta(t-t^*)=\left\{\begin{array}{cc}0&t-t^*\neq0\\\infty&t-t^*=0\end{array}\right.\\\int_0^\infty\delta(t-t^*)C(t)dt=\int_0^\infty\delta(t-t^*)C(t^*)dt=C(t^*)\end{gathered}\]

Therefore, under the influence of the $\delta(t-t^*)$ function, the consumption $C(t)$ at continuous time $t$ is changed to the consumption $C(t^*)$ at a specific time $t^*$. Based on the CRRA utility function $(\gamma\neq1)$ and CAWF, the core content of our article can be summarized by the following equation:

\[\frac{C(t^*)^{utility}}{C(t^*)^{total}}=\frac{\sqrt[1-\gamma]{U_{t^*}(1-\gamma)}}{C(t^*)^{total}}=(1+\mathrm{CAWF})\]

This may explains why the big data interaction in our article differs from the setting in conventional information economics, where macro-level and micro-level information drive changes in agents' beliefs, thus influencing consumption, as discussed in \citet{acemoglu2025big} or \citet{bhandari2025survey} and other similar information economics articles, in their discussions, agents' consumption can typically be described by a function $C(t)$ that does not explicitly incorporate time with the probability measure of subjective beliefs and rational expectations, making their consumption dynamic and continuous. In contrast, our big data interaction shock is analogous to the \textit{Dirac Delta Function} mentioned above, and its effect on agents is manifested in their consumption at time $t^*$. Simply explain: During a period of time $t$, the agents engage in big data interaction behavior. After their cognitive resources are diluted and reduced, their rational states change. At time $t^*=t+1$, the agents make consumption decision, how much of this consumption decision is effective? That is, how much utility does it provide? The weight of effective consumption, 1+CAWF, is the consequence of the impact of big data on the agent's cognitive resources.

Finally, we use CAWF to connect uncertainty and financial friction together, exploring two types of agents: Ordinary investors and investors who interact with big data, perceive environmental uncertainty, and convert effective consumption into utility as a weight of total consumption. Their wealth distributions are affected by financial friction or utility conversion weight. Combining numerical simulations and empirical analysis, we find that : (\textbf{A}) A reduction in financial friction significantly increases the average wealth of investors, but this also exacerbates wealth inequality. (\textbf{B}) When investors can convert consumption into utility as a weight of total consumption, this weight is commensurate with changes in financial friction. The average wealth of investors will increase as the utility conversion weight rises. Furthermore, the impact of the utility conversion weight on wealth inequality follows a U-shaped trend, with the smallest wealth inequality occurring when the weight approaches to 0.5.

\subsection{Supplementary Hypothesis according to Lucas Critique}
\textbf{The Critique}: Lucas Critique is an economic theory proposed by Robert Lucas Jr. in 1976. According to his article (\citealp{lucas1976econometric}). Lucas pointed out that once macroeconomic policies are implemented, people incorporate policy factors into their expectations, thereby influencing current behavior. This alters the policy variables referenced when the macroeconomic policies were initially formulated. Therefore, macroeconomic policies are inherently endogenous variables, and policies cannot achieve their intended effects.

According to the analysis of our article, cognitive resources are treated as the endowments of agents\footnote{We consider that an agent's rationality is not evaluated by the economic results of their behaviors, but determined by the cognitive resources the agent possesses. That is, Cognitive Resources$\rightarrow$Rationality$\rightarrow$Economic Results, rather than Economic Results$\rightarrow$Rationality. Compared to economic results, cognitive resource is a more micro-level and smaller variable.}, their rational behaviors essentially occur through the allocation of cognitive resources. Therefore, the allocation of cognitive resources is similar to the allocation of other general economic resources and can be influenced by policy shocks. Based on this and the Incentive Theory of Asymmetric Information (\citealp{mirrlees1999theory}; \citealp{mirrlees1976optimal})\footnote{They argued that markets can achieve incentive compatibility even if all economic agents are self-interested, provided that effective rules are designed to guide them, thereby achieving Pareto optimal allocation or other social objectives.}, it is necessary to model and make agents' decision-making endogenous in order to reasonably predict their responses, hence, we propose the following hypothesis as a supplement to Lucas Critique:

\begin{Hypothesis}
\textit{When a group of economic policies could be effective, it must include at least two dimensions of policies to achieve one economic goal: The first dimension is called the “Master Policy”, which defines a quantifiable economic objective by generating and allocating material resources of economy. The second dimension is called the “Auxiliary Policy”, which influences the rational state of economic agents, guiding the allocation of their cognitive resources to ensure that changes in their current economic behavior based on future expectations are conducive to achieving the economic objective.}
\end{Hypothesis}

In a short explanation: Policies that include both the \textbf{Master Policy} and \textbf{Auxiliary Policy} dimensions have a stronger “directionality” than quantitative target policies that only include the \textbf{Master Policy} dimension. The former can better ensure that economic resources (material and immaterial) are generated and allocated in a direction which is conducive to achieving policy objectives, and hence effectively achieve economic goal as expected.

\newpage

\singlespacing
\setlength{\bibsep}{2pt}
\bibliographystyle{plainnat}
\bibliography{main}

\end{document}